%% file: main.tex
\documentclass[authoryear,natbib]{sigplanconf}
\usepackage{enumitem}
\usepackage{amsmath}
\usepackage{amsthm}
\usepackage{color}
\usepackage{calc}
\usepackage{xspace}
\usepackage{upgreek}
\usepackage{graphicx}
\usepackage{natbib}
\usepackage{alltt}
\usepackage{array}
\usepackage{tabu}
\usepackage{multirow}
\usepackage{mathpartir}
\usepackage{balance}

\usepackage{ifxetex}
\ifxetex
   \usepackage{fontspec}
   \usepackage{pifont}
   \newcommand\yay{\ding{51}}
   \newcommand\nay{\ding{55}}
   \usepackage[colorlinks=true,
     draft=false,
     linkcolor=black,
     citecolor=black,
     urlcolor=black, 
     pdfauthor={Nguyên, Tobin-Hochstadt, Van Horn}, 
     pdftitle={Soft Contract Verification}]{hyperref}
\else
   \usepackage[breaklinks=true]{hyperref}
   \usepackage{unixode}
   \usepackage[T1]{fontenc}
   \usepackage[scaled=0.77]{DejaVuSansMono}
   \usepackage{bbding}
   \usepackage{setspace}
   \newcommand\yay{\scalebox{.85}{\Checkmark}}
   \newcommand\nay{\scalebox{.85}{\XSolidBrush}}
\fi

\begin{document}

\conferenceinfo{ICFP~'14}{September 1--6, 2014, Gothenburg, Sweden}
\copyrightyear{2014}
\copyrightdata{978-1-4503-2873-9/14/09}
\doi{2628136.2628156}
\exclusivelicense

\title{Soft Contract Verification}
\ifxetex
\authorinfo{Phúc C. Nguy\~{\^e}n}
           {University of Maryland}
           {pcn@cs.umd.edu}
\else
\authorinfo{Ph\'uc C. Nguy\~{\^e}n}
           {University of Maryland}
           {pcn@cs.umd.edu}
\fi
\authorinfo{Sam Tobin-Hochstadt}
           {Indiana University}
           {samth@cs.indiana.edu}
\authorinfo{David Van Horn}
           {University of Maryland}
           {dvanhorn@cs.umd.edu}

\maketitle 

\input{preamble}
\techreporttrue

\begin{abstract}
  Behavioral software contracts are a widely used mechanism for
  governing the flow of values between components.
  However, run-time monitoring and enforcement of contracts imposes
  significant overhead and delays discovery of faulty components to
  run-time.

  To overcome these issues, we present \emph{soft contract
    verification}, which aims to statically prove either complete or
  partial contract correctness of components, written in an untyped,
  higher-order language with first-class contracts.
  Our approach uses higher-order symbolic execution, leveraging
  contracts as a source of symbolic values including unknown
  behavioral values, and employs an updatable heap of contract
  invariants to reason about flow-sensitive facts.
  We prove the symbolic execution soundly approximates the dynamic
  semantics and that \emph{verified programs can't be blamed}.

  The approach is able to analyze first-class contracts, recursive
  data structures, unknown functions, and control-flow-sensitive
  refinements of values, which are all idiomatic in dynamic languages.  It
  makes effective use of an off-the-shelf solver to decide problems
  without heavy encodings.  The approach is competitive with a wide
  range of existing tools---including type systems, flow analyzers, 
  and model checkers---on their own benchmarks.
\end{abstract}


%% Non-mandatory
%% \terms
%% Languages, Theory, Verification


\section{Static verification for dynamic languages}

Contracts~\cite{dvanhorn:meyer-eiffel, dvanhorn:Findler2002Contracts} have become a prominent
mechanism for specifying and enforcing invariants in dynamic
languages~\cite{local:contracts-coffee,
  dvanhorn:Plosch1997Design,
  dvanhorn:Austin2011Virtual,
  dvanhorn:Strickland2012Chaperones,  
  % local:codecontracts, -- not dynamic
  local:clojure-contracts}.  
They offer the expressivity and flexibility of programming in a
dynamic language, while still giving strong guarantees about the
interaction of components.  However, there are two downsides: (1)
contract monitoring is expensive, often prohibitively so, which causes
programmers to write more lax specifications, compromising correctness
for efficiency; and (2) contract violations are found only at
run-time, which delays discovery of faulty components with the usual
negative engineering consequences.

Static verification of contracts would empower programmers to state
stronger properties, get immediate feedback on the correctness of
their software, and avoid worries about run-time enforcement cost
since, once verified, contracts could be removed.
%
%% As research over the past decade has shown, contracts for higher-order
%% languages pose novel semantic and pragmatic
%% challenges~\cite{dvanhorn:Findler2002Contracts,
%%   dvanhorn:Blume2006Sound, dvanhorn:Findler2006Contracts,
%%   dvanhorn:Gronski2007Unifying, dvanhorn:Dimoulas2011Correct,
%%   dvanhorn:Dimoulas2011Contract, dvanhorn:Greenberg2012Contracts,
%%   dvanhorn:Dimoulas2012Complete, dvanhorn:Strickland2012Chaperones},
%% which must be reflected in static reasoning.
%% %
%% These contracts are no longer simple boolean predicates on flat values
%% but test objects and functions for infinitary, behavioral properties.
%
All-or-nothing approaches to verification of typed functional programs
has seen significant advances in the recent work on static contract
checking~\cite{dvanhorn:Xu2009Static,dvanhorn:Xu2012Hybrid,dvanhorn:Vytiniotis2013HALO},
refinement type
checking~\cite{dvanhorn:Terauchi2010Dependent,dvanhorn:DBLP:conf/vmcai/ZhuJ13,dvanhorn:Vazou2013Abstract,local:vazou-icfp2014},
and model
checking~\cite{dvanhorn:Kobayashi2009Types,dvanhorn:Kobayashi2010Higherorder,dvanhorn:Kobayashi2011Predicate}.
However, the highly dynamic nature of untyped languages makes
verification more difficult.

Programs in dynamic languages are often written in idioms that thwart
even simple verification methods such as type inference.  Moreover,
contracts themselves are written within the host
language in the same idiomatic style.  This suggests
that moving beyond all-or-nothing approaches to verification is necessary.

In previous work~\cite{dvanhorn:TobinHochstadt2012Higherorder}, we
proposed an approach to \emph{soft contract verification}, which
enables piecemeal and modular verification of contracts.  This
approach augments a standard reduction semantics for a functional
language with contracts and modules by endowing it with a notion of
``unknown'' values refined by sets of contracts.
Verification is carried out by executing programs on abstract values.

To demonstrate the essence of the idea, consider the following
contrived, but illustrative example.  Let {\tt pos?} and {\tt neg?} be
predicates for positive and negative integers.  Contracts can be
arbitrary predicates, so these functions are also contracts.  Consider
the following contracted function (written in Lisp-like notation):
\begin{alltt}
   (f : pos? \(\conarrow\) neg?)       ; contract
   (define (f x) (* x -1)) ; function
\end{alltt}
We can verify this program by (symbolically) running it on an
``unknown'' input.  Checking the domain contract refines the input to
be an unknown satisfying the set of contracts $\{\mbox{\tt pos?}\}$.
By embedding some basic facts about {\tt pos?}, {\tt neg?}, and {\tt
  -1} into the reduction relation for {\tt *}, we conclude $\mbox{\tt
  (* \leftcurly pos?\rightcurly\ -1)} \stdstep \mbox{\leftcurly{\tt
    neg?}\rightcurly}$, and voil\`a, we've shown once and for all {\tt
  f} meets its contract obligations and cannot be blamed.  We could
therefore soundly eliminate any contract which blames {\tt f}, in this
case {\tt neg?}.

This approach is simple and effective for many programs, but suffers
from several  shortcomings, which we solve in this paper:

%% FIXME: would this sentence be good here?
%% At its core, the system relies on a simple yet effective idea:
%% symbolic execution naturally breaks down programs into simpler
%% components, enabling effective reasoning about seemingly-complex
%% features.  

\paragraph{Solver-aided reasoning:} 
While embedding symbolic arithmetic knowledge for specific, known
contracts works for simple examples, it fails to reason about
arithmetic generally.
Contracts often fail to verify because equivalent
formulations of contracts are not hard-coded in the semantics of
primitives.
Many systems address this issue by incorporating an SMT
solver.
However, for a higher-order language, solver integration is often achieved by
reasoning in a theory of uninterpreted functions or semantic
embeddings~\cite{dvanhorn:Knowles2010Hybrid,dvanhorn:Rondon2008Liquid,dvanhorn:Vytiniotis2013HALO}.

In this paper, we observe that higher-order contracts can be
effectively verified using only a simple first-order solver.  The key
insight is that contracts delay higher-order checks and failures
always occur with a first order witness.  By relying on a (symbolic)
semantic approach to carry out higher-order contract monitoring, we
can use an SMT solver to reason about integers without the
need for sophisticated encodings.  (Examples in
\S\ref{sec:smt-example}.)

% \begin{alltt}
%    (f′ : (λ (x) (> x 0)) \(\conarrow\) (λ (x) (< x 0)))
%    (define (f′ x) (* x -1))
% \end{alltt}

\paragraph{Flow sensitive reasoning:}
Just as our semantic approach decomposes higher-order
contracts into first-order properties, first-order contracts naturally
decompose into conditionals.  Our prior approach fails
to reason effectively about conditionals, requiring
contract checks to be built-in to the semantics. As a result, even
simple programs with conditionals fail to verify:
\begin{alltt}
   (g : int? \(\conarrow\) neg?)
   (define (g x) (if (pos? x) (f x) (f 8)))
\end{alltt}
This is because the true-branch call to {\tt f} is {\tt (f \leftcurly
  int?\rightcurly)} by substitution, although we know from the guard
that {\tt x} satisfies {\tt pos?}.  

In this paper, we observe that flow-sensitivity can be achieved by
replacing substitution with \emph{heap-allocated} abstract
values. These heap addresses are then refined
as they flow through predicates and primitive operations, with no need
for special handling of contracts (\S\ref{sec:flow-sensitive-example}).
As a result, the system is not only effective for contract verification,
but can also handle safety verification for programs with no contracts
at all.

%% FIXME: add cites back in

%% Some type systems approximate this kind of reasoning
%% \cite{dvanhorn:Knowles2010Hybrid,dvanhorn:Rondon2008Liquid,dvanhorn:TobinHochstadt2010Logical},
%% but we can generalize \dots

%% of path sensitivity that discovers control-flow
%% based invariants (\S\ref{sec:flow-sensitive-example}).  Our approach
%% \emph{heap-allocates} abstract value, which are refined by contracts
%% as they flow through predicates, operations, and contract checks.

\paragraph{First-class contracts:}
Pragmatic contract systems enable first-class contracts so new
combinators can be written as functions that consume and produce
contracts.
%  (\(\mu\), \(\cup\), and {\tt pairof} are recursive,
% disjunctive, and pairing contract constructors, resp.):
% \begin{alltt}
%    (define (treeof c) 
%      (\(\mu\) tree/c (\(\cup\) c (pairof tree/c tree/c))))
% \end{alltt}
But to the best of our knowledge, no verification system currently
supports first class contracts (or refinements), and in most
approaches it appears fundamentally difficult to incorporate such a
notion.  

Because we handle contracts (and all other features) by
\emph{execution}, first-class contracts 
% such as \texttt{treeof} 
pose no significant technical challenge and our system reasons about
them effectively (\S\ref{sec:putting-it-all-together-example}).

\paragraph{Converging for non-tail recursion:}
Of course, simply executing programs has a fundamental drawback---it
will fail to terminate in many cases, and when the inputs are unknown,
execution will almost always diverge.  Our prior work used a simple
loop detection algorithm that handled only tail-recursive functions.
As a result, even simple programs operating over inductive data timed
out.
% such this timed out:
% \begin{alltt}
%    (h : (listof pos?) \(\conarrow\) (listof neg?))
%    (define (h xs)
%      (if (null? xs) null (cons (f (car x)) (h (cdr xs)))))
% \end{alltt}

In this paper, we accelerate the convergence of programs by
identifying and approximating regular accumulation of evaluation
contexts, causing common recursive programs to converge on unknown
values, while providing precise predictions
(\S\ref{sec:approxctxs}). As with the rest of our approach, this
happens during execution and is therefore robust to complex,
higher-order control flow.
\begin{center}
\vspace*{-2mm}
\rule{40pt}{1pt}
\vspace*{-1mm}
\end{center}
Combining these techniques yields a system competitive with a diverse
range of existing powerful static checkers, achieving many
of their strengths in concert, while   balancing
the benefits of static contract verification with the
flexibility of dynamic enforcement.

We have built a prototype soft verification engine, which we dub \SCV, based on these
ideas and used it to evaluate the approach (\S\ref{sec:impl}).  Our
evaluation demonstrates that the approach can verify properties typically
reserved for approaches that rely on an underlying type system, while
simultaneously accommodating the dynamism and idioms of untyped
programming languages.
We take examples from work on soft
typing~\cite{dvanhorn:Cartwright1991Soft,dvanhorn:Wright1997Practical},
type systems for untyped
languages~\cite{dvanhorn:TobinHochstadt2010Logical}, static
contract checking~\cite{dvanhorn:Xu2009Static,dvanhorn:Xu2012Hybrid},
refinement type checking~\cite{dvanhorn:Terauchi2010Dependent}, and model checking of
higher-order functional
languages~\cite{dvanhorn:Kobayashi2009Types,dvanhorn:Kobayashi2010Higherorder,dvanhorn:Kobayashi2011Predicate}. 

\SCV can prove all contract checks redundant for almost all of the
examples taken from this broad array of existing program analysis and
type checking work, and can handle many of the tricky higher-order
verification problems demonstrated by other systems.
In other words, our approach is competitive with type systems, model
checkers, and soft typing systems on each of their chosen
benchmarks---in contrast, work on higher-order model checking does not
handle benchmarks aimed at soft typing or occurrence typing, and vice
versa.
In the cases where \SCV does not prove the complete absence of
contract errors, the vast majority of possible dynamic errors are
ruled out, justifying numerous potential optimizations. Over this
corpus of programs, 99\% of the contract and run-time type
checks are proved safe, and could be
eliminated.
%% DVH: who cares -- this is just distracting.
%% \footnote{Additionally, a static analyzer based on our
%%   framework could reject programs with guaranteed-to-fail contracts,
%%   but this would change the underlying semantics of the program.}

We also evaluate the verification of three small interactive video
games which use first-class and dependent contracts pervasively.  The
results show the subsequent elimination of contract monitoring has a
dramatic effect: from a factor speed up of 7 in one case, to three
orders of magnitude in the others.  In essence, these results show the
games are infeasible without contract verification.

%% \cite{dvanhorn:Schwartz2010All}
%% \cite{dvanhorn:Boyer1975SELECTa}

\section{Worked examples} \label{sec:main-idea}

We now present the main ideas of our approach through a 
series of examples taken from work on other verification
techniques, starting from the simplest and working up to a complex
object encoding.

\subsection{Higher-order symbolic reasoning}

Consider the following simple function that transforms functions on
even integers into functions on odd integers.  It has been ascribed
this specification as a contract, which can be monitored at run-time.
% \ifxetex{\relax}\else{\newpage}\fi
\begin{alltt}
   (e2o : (even? \(\conarrow\) even?) \(\conarrow\) (odd? \(\conarrow\) odd?))
   (define (e2o f)
     (λ (n) (- (f (+ n 1)) 1)))
\end{alltt}
A contract monitors the flow of values between components.  In this
case, the contract monitors the interaction between the context and
the {\tt e2o} function.  It is easy to confirm that {\tt e2o} is
correct with respect to the contract; {\tt e2o} holds up its end of
the agreement, and therefore cannot be blamed for any run-time
failures that may arise.  The informal reasoning goes like this: First
assume {\tt f} is an {\tt even? \(\conarrow\) even?} function.  When applied, we
must ensure the argument is even (otherwise {\tt e2o} is at fault),
but may assume the result is even (otherwise the context is at fault).
Next assume {\tt n} is odd (otherwise the context is at fault) and
ensure the result is odd (otherwise {\tt e2o} is at fault).  Since
{\tt (+ n 1)} is even when {\tt n} is odd, {\tt f} is applied to an
even argument, producing an even result.  Subtracting one therefore
gives an odd result, as desired.

This kind of reasoning mimics the step-by-step computation of {\tt
  e2o}, but rather than considering some particular inputs, it
considers these inputs symbolically to verify all possible executions
of {\tt e2o}.  We systematize this kind of reasoning by augmenting a
standard reduction semantics for contracts with symbolic values that
are refined by sets of contracts.  At first approximation, the
semantics includes reductions such as:
\begin{align*}
\mbox{\tt (+ \leftcurly odd?\rightcurly\ 1)} &\stdstep \mbox{\tt \leftcurly even?\rightcurly}\text{, and}\\
\mbox{\tt (\leftcurly even? \(\conarrow\) even?\rightcurly\ \leftcurly even?\rightcurly)} &\stdstep \mbox{\tt \leftcurly even?\rightcurly}\text.
\end{align*}

This kind of symbolic reasoning mimics a programmer's informal
intuitions which employ contracts to refine unknown values and to
verify components meet their specifications.  If a component cannot be
blamed in the symbolic semantics, we can safely conclude it cannot be
blamed in general.

\subsection{Flow sensitive reasoning}
\label{sec:flow-sensitive-example}

Programmers using untyped languages often use a mixture of
type-based and flow-based reasoning to design programs. 
The analysis naturally takes advantage of type tests idiomatic in
dynamic languages even when the tests are buried in complex
expressions.  The following function taken from work on occurrence
typing ~\cite{dvanhorn:TobinHochstadt2010Logical} can be proven safe
using our symbolic semantics:
\begin{alltt}
(f : (or/c int? str?) cons? \(\conarrow\) int?)
(define (f x p)
  (cond 
    [(and (int? x) (int? (car p))) (+ x (car p))]
    [(int? (car p)) (+ (str-len x) (car p))]
    [else 0]))
\end{alltt}

Here, {\tt int?}, {\tt str?}, and {\tt cons?} are type predicates for
integers, strings, and pairs, respectively.  The contract {\tt
  (or/c int? str?)} uses the {\tt or/c} contract combinator to
construct a contract specifying a value is either an integer or a
string.

A programmer would convince themselves this program was safe by using
the control dominating predicates to refine the types of {\tt x} and
{\tt (car p)} in each branch of the conditional.%
\footnote{The call to {\tt str-len} is safe because
{\tt (and (int? x) (int? (car p)))} being false and
{\tt (int? (car p))} being true implies that
{\tt (int? x)} is false, which in turns implies
{\tt x} is a string as enforced by {\tt f}'s contract.}
Our symbolic
semantics accommodates exactly this kind of reasoning in order to
verify this example.  However, there is a technical challenge here.  A
straightforward substitution-based semantics would not reflect the
flow-sensitive facts.  Focusing just on the first clause, a
substitution model would give:
%{int?,str?} is later interpreted as a conjunction of int? and str?
%in the formalism, so we should use {∪ int? str?} here
\begin{alltt}
(cond 
  [(and (int? {\leftcurly}(or/c int? str?)\rightcurly) (int? (car {\leftcurly}cons?\rightcurly)))
   (+ {\leftcurly}(or/c int? str?){\rightcurly} (car {\leftcurly}cons?{\rightcurly}))] …)
\end{alltt}
At this point, it's too late to communicate the refinement of these
sets implied by the test evaluating to true, so the semantics would
report the contract on {\tt +} potentially being violated because the
first argument may be a string, and the second argument may be
anything.  We overcome this challenge by modelling symbolic values as
heap-allocated sets of contracts.  When predicates and data structure
accessors are applied to heap addresses, we refine the corresponding
sets to reflect what must be true.  So the program is modelled as:
\newcommand\where{\ensuremath{\text{\rm where}}}
\begin{alltt}
   (cond 
     [(and (int? \(\mloc\sb1\)) (int? (car \(\mloc\sb2\))))
      (+ \(\mloc\sb1\) (car \(\mloc\sb2\)))] …)
   \(\where \mloc\sb1\;\mapsto\;\){\leftcurly}(or/c int? string?)\rightcurly\(, \mloc\sb2\;\mapsto\;\){\leftcurly}cons?\rightcurly
\end{alltt}
In the course of evaluating the test, we get to {\tt (int?
  \(\mloc\sb1\))}, the semantics conceptually forks the evaluator and
refines the heap:
\begin{align*}
\mbox{\tt (int? \(\mloc\sb1\))} &\stdstep \mbox{\tt true}\text{, where }\mloc\sb1\mapsto\leftcurly\mbox{\tt int?}\rightcurly\\
 &\stdstep \mbox{\tt false}\text{, where }\mloc\sb1\mapsto\leftcurly\mbox{\tt string?}\rightcurly
\end{align*}
Similar refinements to \(\mloc\sb2\) are communicated through the heap
for {\tt (int? (car \(\mloc\sb2\)))}, thereby making {\tt (+
  \(\mloc\sb1\) (car \(\mloc\sb2\)))} safe.  This
  simple idea is effective in achieving flow-based refinements.  It
naturally handles deeply nested and inter-procedural conditionals.

\subsection{Incorporating an SMT solver}
\label{sec:smt-example}

The techniques described so far are highly effective for reasoning
about functions and many kinds of recursive data structures. However,
effective reasoning about many kinds of base values, such as integers,
requires sophisticated domain-specific knowledge. Rather than build
such a tool ourselves, we defer to existing high-quality solvers for
these domains.  Unlike many solver-aided verification tools, however,
we use the solver \emph{only} for queries on base values, rather than
attempting to encode a rich, higher-order language into one that is
accepted by the solver.

To demonstrate our approach, we take an example ({\tt intro3}) from work on
model checking higher-order programs \cite{dvanhorn:Kobayashi2010Higherorder}.

\begin{alltt}
   (>/c : int? \(\conarrow\) any \(\conarrow\) bool?)
   (define (>/c lo) (λ (x) (and (int? x) (> x lo))))
\end{alltt}
\begin{alltt}
   (define (f x g) (g (+ x 1)))
\end{alltt}
\begin{alltt}
   (h : [x : int?] \(\conarrow\) [y : (>/c x)] \(\conarrow\) (>/c y))
   (define (h x) ...) ; unknown definition
\end{alltt}
\begin{alltt}
   (main : int? \(\conarrow\) (>/c 0))
   (define (main n) (if (≥ n 0) (f n (h n)) 1))
\end{alltt}

In this program, we define a contract combinator ({\tt >/c}) that
creates a check for an integer from a lower bound; a helper function
{\tt f}, which comes without a contract; and an \emph{unknown}
function {\tt h} that given an integer {\tt x}, returns a function
mapping some number {\tt y} that is greater than {\tt x} to an answer
greater than {\tt y}---here {\tt h}'s specification is given, but not
its implementation.  (Note {\tt h}'s contract is dependent.)
We verify {\tt main}'s correctness, which means
it definitely returns a positive integer and does not violate
{\tt h}'s contract. 

According to its contract, {\tt main} is passed an integer {\tt n}.
If {\tt n} is negative, {\tt main} returns {\tt 1}, satisfying the contract.
Otherwise the function applies {\tt f} to {\tt n} and {\tt (h n)}.
Function {\tt h}, by its contract, returns another function
that requires  a number greater than {\tt n}.
Examining  {\tt f}'s definition,
we see {\tt h} (now bound to {\tt g}) is eventually applied to {\tt (+ n 1)}.
Let {\tt n\(_1\)} be the result of {\tt (+ n 1)}.
And by {\tt h}'s contract, we know the answer is another integer
greater than {\tt (+ n 1)}.
Let us name this answer {\tt n\(_2\)}.
In order to verify that {\tt main} satisfies contract {\tt (>/c 0)},
we need to verify that {\tt n\(_2\)} is a positive integer.
% FIXME: we don't discuss verifying that n_1 > n

Once {\tt f} returns, the heap contains several addresses with
contracts:
\[
\begin{array}{lcl}
\mathtt{n} &  \mapsto & \{\texttt{int?}, \texttt{(≥/c 0)}\}\\
\mathtt{n_1} &  \mapsto & \{\texttt{int?}, \texttt{(=/c (+ n 1))}\}\\
\mathtt{n_2} &  \mapsto & \{\texttt{int?}, \texttt{(>/c n$_1$)}\}\\
\end{array}
\]
We then  translate this information to a query for an
external solver:
\begin{alltt}
  n, n\(\sb{1}\), n\(\sb{2}\): INT;
  ASSERT n ≥ 0;
  ASSERT n\(\sb{1}\) = n + 1;
  ASSERT n\(\sb{2}\) > n\(\sb{1}\);
  QUERY n\(\sb{2}\) > 0;
\end{alltt}
Solvers such as CVC4~\cite{dvanhorn:Barrett2011CVC4} and Z3~\cite{dvanhorn:DeMoura2008Z3}
easily verify this implication, proving {\tt main}'s correctness.

Refinements such as {\tt (≥/c 0)} are generated by \emph{primitive}
applications {\tt (≥ x 0)},
and queries are generated from translation of the heap, not arbitrary expressions.
This has a few consequences.
First, by the time we have value {\tt v} satisfying predicate {\tt p}
on the heap, we know that {\tt p} terminates successfully on {\tt v}.
Issues such as errors (from {\tt p} itself) or divergence 
are handled elsewhere in other evaluation branches.
Second, we only need to translate a small set
of simple, well understood contracts---not arbitrary expressions.
Evaluation naturally breaks down complex expressions,
and properties are discovered even when they are buried
in complex, higher-order functions.
Given a translation for {\tt (>/c 0)},
the analysis automatically takes advantage of the  solver
even when the predicate contains {\tt >} in a complex way,
such as {\tt (λ (x) (or (> x 0) $e$)} where $e$ is an arbitrary expression.
Predicates that lack translations to SMT only reduce precision, never soundness.

\subsection{Converging for non-tail recursion}
\label{sec:approxctxs}
The techniques sketched above provide high precision in the examples
considered, but simply executing programs on abstract values is
unlikely to terminate in the presence of recursion.
When an abstract value stands for an infinite set of concrete values,
execution may proceed infinitely, building up
an ever-growing evaluation context.
To tackle this problem, we \emph{summarize} this context to coalesce
repeated structures and enable termination on many
recursive programs.  Although guaranteed termination is not
our goal, the empirical results (\S\ref{sec:impl}) demonstrate that the method is
effective in practice.

The following example program is taken from work on model checking of higher-order
functional programs \cite{dvanhorn:Kobayashi2010Higherorder}, and
demonstrates checking non-trivial safety properties on recursive
functions. Note that no loop invariants need be provided by the user.

\begin{alltt}
   (main : (and/c int? ≥0?) \(\conarrow\) (and/c int? ≥0?))
   (define (main n)
     (let ([l (make-list n)])
       (if (> n 0) (car (reverse l empty)) 0)))
\end{alltt}
\begin{alltt}
   (define (reverse l ac)
     (if (empty? l) ac
         (reverse (cdr l) (cons (car l) ac))))
\end{alltt}
\begin{alltt}
   (define (make-list n)
     (if (= n 0) empty
         (cons n (make-list (- n 1))))))
\end{alltt}

Again, we aim to verify both the specified contract for {\tt main} as
well as the preconditions for primitive operations such as {\tt
  car}. Most significantly, we need to verify that \texttt{(reverse l
  empty)} produces a non-empty list (so that {\tt car} succeeds) and
that its first element is a positive integer.  The local functions
{\tt reverse} and {\tt make-list} do not come with a contract.

This problem is more challenging than the original OCaml version of
the same program, due to the lack of types.
This program represents a common idiom in dynamic languages:
not all values are contracted, and there is no
type system on which to piggy-back verification.
In addition, programmers often rely on inter-procedural reasoning
to justify their code's correctness, as here with {\tt reverse}.

We verify {\tt main} by applying it to an abstract (unknown) value {\tt n\(_1\)}.
The  contract ensures that within the body,
{\tt n\(_1\)} is a non-negative integer.

The integer {\tt n\(_1\)} is first passed to {\tt make-list}.
%All numeric operations on {\tt n\(_1\)} such as comparison and subtraction are safe.
The comparison {\tt (= n\(_1\)\ 0)} non-deterministically returns
{\tt $\strue$} and {\tt $\sfalse$},
updating the information known about {\tt n\(_1\)} to be either {\tt 0} or {\tt (>/c 0)}
 in each corresponding case.
In the first case, {\tt make-list} returns {\tt empty}.
In the second case, {\tt make-list} proceeds to the recursive application
{\tt (make-list n\(_2\))}, where {\tt n\(_2\)} is the abstract non-negative integer
obtained from evaluating {\tt (- n\(_1\) 1)}. However, {\tt (make-list
  n\(_2\))} is identical to the original call {\tt (make-list
  n\(_1\))} up to renaming, since both {\tt n\(_1\)} and {\tt n\(_2\)} are non-negative. Therefore, we pause here and use a summary of {\tt make-list}'s result
instead of continuing in an infinite loop.

Since we already know that {\tt empty} is one possible result of {\tt (make-list n\(_1\))},
we use it as the result of {\tt (make-list n\(_2\))}.
The application {\tt (make-list n\(_1\))} therefore produces the pair
{\tt $\sconsc{\tt n_1}\sempty$}, which is another answer for the original application.
We could continue this process and plug this new result into the
pending application {\tt (make-list n\(_2\))}.
But by observing that the application produces a list of one positive integer
when the recursive call produces {\tt empty},
we approximate the new result $\sconsc{\tt n_1}\sempty$ to a
non-empty list of positive integers, and then use this approximate answer as the result of
 the pending application {\tt (make-list n\(_2\))}.
This then induces another result for {\tt (make-list n\(_1\))},
 a list of two or more positive integers, but this
 is subsumed by the previous answer of non-empty integer list.
 We have now discovered \emph{all} possible return values of {\tt make-list}
when applied to a non-negative integer:
it maps {\tt 0} to {\tt empty}, and positive integers
to a non-empty list of positive integers.

Although our explanation made use of the order, the soundness of
analyzing {\tt make-list} does not depend on the order of exploring
non-deterministic branches.  Each recursive application with repeated
arguments generates a waiting context, and each function return
generates a new case to resume.  There is an implicit work-list
algorithm in the modified semantics (\S\ref{sec:termination}).

When {\tt make-list} returns to {\tt main},
we have two separate cases:
either {\tt n\(_1\)} is 0 and {\tt l} is {\tt empty},
or {\tt n\(_1\)} is positive and {\tt l} is non-empty.
In the first case, {\tt (> n\(_1\) 0)} is false
and {\tt main} returns {\tt 0}, satisfying the contract.
Otherwise, {\tt main} proceeds to reversing the list
before taking its first element.

Using the same mechanism as with {\tt make-list},
the analysis infers that {\tt reverse}
returns a non-empty list when either of its arguments
({\tt l} or {\tt acc}) is non-empty.
In addition, {\tt reverse} only receives arguments of proper lists,
so all partial operations on {\tt l}
such as $\scar$ and $\scdr$ are safe when {\tt l} is not $\sempty$, without needing an
 explicit check.
The function eventually returns a non-empty list of integers
to {\tt main}, justifying {\tt main}'s call to the partial function
{\tt car}, producing a positive integer.
Thus, {\tt main} never has a run-time error in any context.

While this analysis makes use of the implementation of {\tt make-list}
and {\tt reverse}, that does not imply that it is whole-program.
Instead, it is \emph{modular} in its use of unknown values abstracting
arbitrary behavior.  For example, {\tt make-list} could instead be an
abstract value represented by a contract that always produces lists of
integers.  The analysis would still succeed in proving all contracts
safe except the use of {\tt car} in {\tt main}---this shows the
flexibility available in choosing between precision and modularity.
In addition, the analysis does not have to be perfectly precise to be
useful.  If it successfully verifies most contracts in a module, that
already greatly improves confidence about the module's correctness and
justifies the elimination of numerous expensive dynamic checks.

\subsection{Putting it all together}
\label{sec:putting-it-all-together-example}

The following example illustrates all aspects of our system. 
For this, we choose a simple encoding of classes as functions that
produce objects, where objects are again functions that respond to
messages named by symbols.
We then verify the correctness of a \emph{mixin}: a
function from classes to classes. 
The {\tt vec/c} contract enforces the interface of a 2D-vector class
whose objects accept messages {\tt 'x}, {\tt 'y}, and {\tt 'add}
for extracting components and vector addition.
\begin{alltt}
   (define vec/c
     ([msg : (one-of/c 'x 'y 'add)]
      \(\conarrow\) (match msg 
          [(or 'x 'y) real?]
          ['add (vec/c \(\conarrow\) vec/c)])))
\end{alltt}
This definition demonstrates several powerful contract system
features which we are able to handle:
\begin{itemize}%[noitemsep,nolistsep]
\item contracts are first-class values, as in the  definition of {\tt vec/c},
\item contracts may include arbitrary predicates, such as {\tt real?},
\item contracts may be recursive, as in the contract for  {\tt 'add},
\item function contracts may express \emph{dependent} relationships
  between the domain and range---the contract of the result of method
  selection for {\tt vec/c} depends on the method chosen.
\end{itemize}

Suppose we want to define a mixin that takes any class that satisfies
the {\tt vec/c} interface and produces another class with
added vector operations such as {\tt 'len}
for computing the vector's length.
The {\tt extend} function defines this mixin,
and {\tt ext-vec/c} specifies the new interface.
We verify that {\tt extend} violates no contracts and
returns a class that respects specifications from {\tt ext-vec/c}.
% 'extend' does not implement 'add', so 'add' can only return vec/c
\begin{alltt}
   (extend : (real? real? \(\conarrow\) vec/c)
             \(\conarrow\) (real? real? \(\conarrow\) ext-vec/c))
   (define (extend mk-vec)
     (λ (x y)
       (let ([vec (mk-vec x y)])
         (λ (m)
           (match m
             ['len
              (let ([x (vec 'x)] [y (vec 'y)])
                (sqrt (+ (* x x) (* y y))))]
             [_ (vec m)])))))

   (define ext-vec/c
     ([msg : (one-of/c 'x 'y 'add 'len)]
      \(\conarrow\) (match msg 
          [(or 'x 'y) real?]
          ['add (vec/c \(\conarrow\) vec/c)]
          ['len (and/c real? (≥/c 0))])))
\end{alltt}

To verify {\tt extend}, we provide an arbitrary value,
which is guaranteed by its contract to be
a class matching {\tt vec/c}.
The mixin returns a new class whose objects
understand messages {\tt 'x}, {\tt 'y}, {\tt 'add}, and {\tt 'len}.
This new class defines method {\tt 'len}
and relies on the underlying class
to respond to  {\tt 'x}, {\tt 'y}, and {\tt 'add}.
Because the old class is constrained by contract {\tt vec/c},
the new class will not violate its contract when responding
to messages {\tt 'x}, {\tt 'y}, and {\tt 'add}.

For the {\tt 'len}  message, the object in the new vector class
extracts its components as abstract numbers {\tt x} and {\tt y},
according to interface {\tt vec/c}.
It then computes their squares and leaves the following information on the heap:
\[
\begin{array}{lcl}
\mathtt{x^2} &  \mapsto & \{\texttt{real?}, \texttt{(=/c (* x x))}\}\\
\mathtt{y^2} &  \mapsto & \{\texttt{real?}, \texttt{(=/c (* y y))}\}\\
\mathtt{s} &  \mapsto & \{\texttt{real?}, \texttt{(=/c (+ x\(^2\) y\(^2\)))}\}\\
\end{array}
\]
Solvers such as Z3 \cite{dvanhorn:DeMoura2008Z3}
can handle simple non-linear arithmetic and verify that the
sum {\tt s} is non-negative, thus the {\tt sqrt} operation is safe.
Execution proceeds to take the square root---now called {\tt l}---and refines
the heap with the following mapping:
\[
\begin{array}{lcl}
\mathtt{l} &  \mapsto & \{\texttt{real?}, \texttt{(=/c (sqrt s))}\}\\
\end{array}
\]
When the method returns, its result is checked by contract {\tt ext-vec/c}
to be a non-negative number.
We again rely on the solver to prove that this is the case.

Therefore, {\tt extend} is guaranteed to produce a new class that
is correct with respect to interface {\tt vec-ext/c},
justifying the elimination of expensive run-time checks.
In a Racket program computing the length of 100000 random vectors, %TODO cite Racket?
eliminating these contracts results in a 100-fold speed-up.
While such dramatic results are unlikely in full programs,
measurements of existing Racket programs suggests that 50\% speedups
are possible~\cite{dvanhorn:Strickland2012Chaperones}. 

%In summary, to verify a module, we place it in an unknown context
%and run the program under a modified semantics extended with opaque values.
%Values flowing across components are enforced by contracts,
%which refine their behavior in each module they flow in.
%We name opaque values and maintain a ``heap'' mapping each name
%to the set of contracts it satisfies as an upperbound to its behavior,
%which become more precise as execution proceeds and follows specific branches.
%Finally, we augment each execution branch with a history of function calls,
%detect repeated recursive applications, and summarize the result.

\section{A Symbolic Language with Contracts}
\label{sec:concrete}

In this section, we give a reduction system describing the core of our
approach.  Symbolic \lang\ is a model of a language with first-class
contracts and \emph{symbolic values}. We first present the semantics,
including handling of primitives and unknown functions.  We then
describe how the handling of primitive values integrates with external
solvers.  Finally, we show an abstraction of our symbolic system to
accelerate convergence.  For each abstraction, we relate concrete and
symbolic programs and prove a soundness theorem.

At a high level, the key idea of our semantics is that abstract values
behave non-deterministically in all possible ways that concrete values
might behave. Furthmore, abstract values can be bounded by
specifications in the form of contracts that limit these behaviors. As
a result, an operational semantics for abstract values explores all
the ways that the concrete program under consideration might be used.  

Given this operational semantics, we can then examine the results of
evaluation to see if any results are errors blaming the components we
wish to verify.  If they do not, then our soundness theorem implies
that there are no ways for the component to be blamed,
regardless of what other parts of the program do. Thus, we have
verified the component against its contract in all contexts.  We make
this notion precise in section~\ref{sec:soundness}.

\subsection{Syntax of Symbolic \lang}
\label{sec:syntax}
Our initial language models the functional core of many modern dynamic
languages, extended with behavioral, first-class contracts, as well as
symbolic values.  The abstract syntax is shown in
figure~\ref{fig:syntax-src}.  Syntax needed only for symbolic
execution is highlighted in \graybox{\mathrm{gray}}; we discuss it
after the syntax of concrete programs.

\begin{figure}[t!]
\[
\begin{array}{l@{\ \,}l@{\;}cl}
\mbox{Programs} & \mprg,\mprgo & ::= & \mvmod\:\mexp\\[.8mm]
\mbox{Modules}  & \mmod & ::= & \smod\mmodvar{\mpreval_c}\mpreval\\[.8mm]
\mbox{Expressions} & \mexp,\mcon,\mcono & ::= & \mval\ |\ \mvar^\mlab\ |\ \sapp\mexp\mexp^\mlab\ |\ \sapp\mop{\vec\mexp}^\mlab\
          |\ \sif\mexp\mexp\mexp\ |\ \slamc\mvar\mcon\mcono\\
&     &   | & \chk\mcon\mlab\mlab\mlab\mexp\
          |\ \simpleblm\mlab\mlab\ |\ \graybox{\assume\mval\mval}\\[.8mm]
\mbox{Pre-values} & \mpreval & ::= & \slam\mvar\mexp\ |\ \mbase\ |\ \sconsc\mval\mval\ |\ \slamc\mvar\mval\mcon\ |\ \graybox{\,\opaque\,}\\[.8mm]
\mbox{Values} & \mval & ::= & \graybox{\mLab}\ |\ \mpreval{\graybox{/\mvval}}\\[.8mm]
\mbox{Base values} & \mbase & ::= & \strue\ |\ \sfalse\ |\ \mnum\ |\ \sempty\\[.8mm]
\mbox{Operations} & \mop & ::= & \moppred\ |\ \sconsop\ |\ \scar\ |\ \scdr\ |\ \ssucc\ |\ \splus\ |\ \sequalp\\
& & | & \dots\\[.8mm]
\mbox{Predicates} & \moppred & ::= & \snump\ |\ \sfalsep\ |\ \sconsp\ |\ \semptyp\\
& & | & \sprocp\ |\ \sdepp\\[.8mm]
\mbox{Variables} & \mvar,\mmodvar,\mlab & \in & \mathit{identifier}\\[.8mm]
\mbox{Addresses} & \mLab & \in & \mathit{address}
\end{array}
\]
\caption{Syntax of Symbolic \lang}
\label{fig:syntax-src}
\end{figure}

A program $\mprg$ is a sequence of module definitions followed by a
top-level expression which may reference the modules.
Each module $\mmod$ has a name $\mmodvar$ and exports a single value $\mpreval$
with behavior enforced by contract $\mpreval_c$.
(Generalizing to multiple-export modules is straightforward.)

Expressions include standard forms such as values, variable and module
references, applications, and conditionals, as well as those for
constructing and monitoring contracts.  Contracts are first-class
values and can be produced by arbitrary expressions.  For clarity,
when an expression plays the role of a contract, we use the
metavariable $\mcon$ and $\mcono$, rather than $\mexp$.  A
\emph{dependent} function contract ($\slamc\mvar\mcon\mcono$) monitors
a function's argument with $\mcon$ and its result with the contract
produced by applying $\slam\mvar\mcono$ to the argument.

A contract violation at run-time causes \emph{blame},
an error with information about who violated the contract.
We write $\simpleblm\mlab\mlaboo$ to mean module $\mlab$ is blamed
for violating the contract from $\mlaboo$.
The form ($\chk\mcon\mlab\mlabo\mlaboo\mexp$)
monitors expression $\mexp$ with contract $\mcon$,
with $\mlab$ being the positive party,
$\mlabo$ the negative party, and $\mlaboo$ the source of the contract.
The system blames the positive party if $\mexp$ produces a value violating $\mcon$,
and the negative one if $\mexp$ is misused by the \emph{context} of the contract check.
To make context information available at run-time,
we annotate  references and applications with labels indicating the module
they appear in,
or  $\toplevel$ for the top-level expression.
For example, $\mvar^\toplevel$ denotes a reference to the name
$\mvar$ from the top level,
and ${(\sapp\ssucc\mvar)}^\mlab$ denotes an addition inside module $\mlab$.
When a module $\mlab$ causes a primitive error, such as applying $5$,
we also write $\simpleblm\mlab\Lambda$, indicating that it violates
a contract with the language.
We omit labels when they are irrelevant or can be inferred.

\emph{Pre-values} $\mpreval$---extended to values below---include abstractions, base values,
pairs of values, and dependent contracts with domain components evaluated.
Base values include numbers, booleans, and the empty list.
Primitive operations over values are standard,
including predicates $\moppred$
for dynamic testing of data types.%, as commonly found in untyped languages.

%\subsection{Symbolic \lang}
\label{sec:symbolic}

To reason about absent components,
we equip \lang\ with \emph{unknown},
or \emph{symbolic values},
which abstract over multiple concrete values exhibiting a range of behavior.
An unknown value $\opaque$ stands for
any value in the language.
For soundness, execution must account
for \emph{all} possible concretizations of abstract values,
and reduction becomes non-deterministic.
As the program proceeds through tests and contract checks,
more assumptions can be made about abstract values.
To remember these assumptions, we take the 
 \emph{pre-values}
and refine each with a set of contracts it is known to satisfy,
written $\with\mpreval\mvval$. 

Finally, to track refinements of unknown values,
we use heap addresses $\mLab$ as symbolic values 
and track them in a heap, which is  a finite map from addresses to refined pre-values:
\[
\begin{array}{l@{\quad}c}
\mbox{Heaps} & \msto ::= \vec{\spair{\mLab}{\with\mpreval{\vec\mval}}}\text.
\end{array}
\]
The heap $\msto$ maps addresses allocated for unknown values to
refinements expressed as contracts; these refinements are updated
during reduction and represent upper bounds on what they might be at
run-time.
Intuitively, any possible concrete execution can be obtained by substituting
addresses with concrete values within bounds specified by the heap.
We omit refinements when they are empty or irrelevant.
%% A heap $\msto$ is a finite function mapping labels to refined pre-values
%% as upper bounds on what they might be at run-time.

%% DVH: Out of place, too much detail at this point.  Maybe find a
%% place to re-insert:

%% Refinement sets are  useful even for concrete values:
%%  $\with{\syntax{7}}{\syntax{p?}}$  satisfies 
%%  $\syntax{p?}$ even if $\syntax{p?}$ is unknown.

%% Answers, which represent the final results of evaluation, include
%% values and stuck programs that blame a faulty party.

%% from execution includes blame with full information about the party at
%% fault.  Precise blame is a pragmatic necessity in our analysis,
%% because blame partitions errors which may occur in the portion of the
%% program being verified from those resulting from bugs in external
%% components. Without blame, \emph{any} error seen would be a potential
%% verification failure; with blame, irrelevant errors are discarded by
%% the verifier.
%
%% Furthermore, precise blame 
%% enables us to move beyond a binary view of contract verification.
%% The system aims to prove and justify the elimination of as many contracts
%% as possible, but it does not have to eliminate all of them to be helpful.
%% Finally, in cases where the analysis fails to discharge a contract,
%%  context information is provided, which is useful
%% for debugging the program. 

%Next, we present the revised semantics for our extended language.

\subsection{Semantics of Symbolic \lang}
\label{sec:semantics}
We now turn to the reduction semantics for Symbolic \lang, which
combines standard rules for untyped languages with  behavior for
unknown values. 
Reduction is defined as a relation on states, parameterized by a
module context:
$$
\vec\mmod \vdash \mstate \stdstep \mstate'
$$ 
States consist either of an expression paired with a heap, or
blame:
\[
\begin{array}{l@{\quad}c}
\mbox{States} & \mstate ::= (\mexp,\msto)\ |\ \sblm\mlab\mlab\text.
\end{array}
\]

%% Except where specified, we close these rules under the grammar of
%% evaluation contexts given in figure~\ref{fig:syntax-src}.

We present the rules inline; a full version of all rules is given in
the appendix of the \techrepcite. In the inline presentation of rules,
we systematically omit labels in contracts, these are presented in the
full rules.
We omit the module context whenever it is irrelevant.

\subsubsection{Basic rules}
\label{sec:basicrules}

Applications of primitives are interpreted by a $\delta$ relation,
which maps operations, arguments and heaps to results and new heaps.
\begin{mathpar}
\inferrule[Apply-Primitive]{\delta(\msto,\mop{\vec\mval},) \ni \mstate}
          {\spapp\mop{\vec\mval}, \msto
            \stdstep
            \mstate }
\end{mathpar}
The use of a $\delta$ relation in reduction semantics is standard, but
typically it is a function and is independent of the heap.
We make $\delta$ dependent on the heap in order to use and update the
current set of invariants; we make it a relation, since it may behave
non-deterministically on unknown values.
For example, in interpreting {\tt (> \(\bullet\) 5)}, the $\delta$
relation will produce two results: one \strue, with an updated heap
to reflect the unknown value is {\tt (>/c 5)}; the other \sfalse, with a
heap reflecting the opposite.
The $\delta$ relation is thus the hub of
the verification system and a point of interaction with the
SMT solver.  It is described in more detail in section~\ref{sec:prim}.

Applications of $\lambda$-abstractions follow standard
$\beta$-reduction; applications of non-functions result in blame.
\begin{mathpar}
\inferrule[Apply-Function]{\ }{\spapp{\slamp\mvar\mexp}{\mval},\msto
\stdstep
\subst\mval\mvar\mexp,\msto}

\inferrule[Apply-Non-Function]{\deltamap\msto\sprocp\mval\sfalse\mstoo}
{{\spapp{\mval}{\mval'}},\msto
\stdstep
\simpleblm\relax\relax,\mstoo}
\end{mathpar}
Notice that the $\delta$ relation is employed to determine whether the
value in operator position is a function using the $\sprocp$
primitive.
(Non-functions include concrete numbers and booleans as well as
abstract values known to exclude functions; application of abstract
values that may be functions is described in
section~\ref{sec:handling-unknown}.)

Conditionals follow a common treatment in untyped languages in which
values other than $\sfalse$ are considered true.
\begin{mathpar}
\inferrule[If-True]{\deltamap\msto\sfalsep\mval\sfalse{\msto'}}
          {\sif\mval{\mexp_1}{\mexp_2},\msto \stdstep \mexp_1,\msto' }

\inferrule[If-False]{\deltamap\msto\sfalsep\mval\strue{\msto'}}
          {\sif\mval{\mexp_1}{\mexp_2},\msto \stdstep \mexp_2,\msto'}
\end{mathpar}
Just as in the case of \emph{Apply-Non-Function}, the interpretation
of conditionals uses the $\delta$ relation to determine whether
\sfalsep holds, which takes into account all of the knowledge
accumulated in $\msto$ and in either branch that is taken, updates
the current knowledge to reflect whether \sfalsep of \mval holds.
This is the mechanism by which control-flow based refinements
are enabled.

The two rules for module references reflect the  approach in
which contracts are treated as \emph{boundaries} between
components~\cite{dvanhorn:Dimoulas2011Correct}: a module self-reference
 incurs no contract check, while cross-module references are
protected by the specified contract.
\begin{mathpar}
\inferrule[Module-Self-Reference]{\smod\mmodvar{\mpreval_c}\mpreval \in \mvmod}
{\vec\mmod\vdash\mmodvar^\mmodvar\!,\msto \stdstep \mpreval,\msto}

\inferrule[Module-External-Reference]{\smod\mmodvar{\mpreval_c}\mpreval \in \mvmod
\\ \mmodvar\neq\mlab}
{\vec\mmod\vdash\mmodvar^\mlab\!,\msto
\stdstep
 \achksimple{\mpreval_c}{\mpreval},\msto}
\end{mathpar}

Finally, any state that is stuck with blame inside an evaluation
context transitions to a final blame state that discards the
surrounding context and heap.
\begin{mathpar}
\inferrule[Halt-Blame]{ }
{\mctx[\sblm\relax\relax],\msto \stdstep {\sblm\relax\relax}}
\end{mathpar}
Evaluation contexts as defined as follows:
\[
\begin{array}{@{\ \,}l@{\;}cl}
 \mctx & ::= & [\;]\ |\ \sapp\mctx\mexp\ |\ \sapp\mval\mctx\ |\ \sapp\mop{\vec\mval \mctx \vec\mexp}
            \ |\ \sif\mctx\mexp\mexp\\
      &   | & \achksimple\mctx\mexp\ |\ \achksimple\mval\mctx\ |\ \slamc\mvar\mctx\mexp
\end{array}
\]

\subsubsection{Contract monitoring}
\label{sec:contract-monitoring}

Contract monitoring follows existing operational semantics for
contracts~\cite{dvanhorn:Findler2002Contracts}, with extensions to
handle and refine symbolic values.

There are several cases for checking a value against a contract.  If
the contract is not a function contract, we say it is \emph{flat},
denoting a first-order property to be checked immediately.  We thus
expand the checking expression to a conditional.
\begin{mathpar}
\inferrule[Monitor-Flat-Contract]{\deltamap\msto\sdepp{\mval_c}\sfalse\mstoo\\
  \ambig\mstoo\mval{\mval_c}}
          {\achksimple{\mval_c}\mval, \msto
            \stdstep
            \sif{\spapp{\mval_c}\mval}{\assume\mval{\mval_c}}{\simpleblm{}{}}, \mstoo}
\end{mathpar}
Since contracts are first-class, they can also be abstract values; we
rely on $\delta$ to determine whether a value is a flat contract by
using (the negation of) the predicate for dependent contracts, $\sdepp$,
instead of examining the syntax.
This rule is standard except for the use of $\assume\mval{\mval_c}$ and the
($\ambig\cdot\cdot\cdot$) judgment.
The $\assume\mval{\mval_c}$ form, which would normally just be
$\mval$, dynamically refines value $\mval$ and the heap to indicate
that $\mval$ satisfies $\mval_c$; \textsf{assume} is discussed further
in section~\ref{sec:handling-unknown}.
The judgment $\ambig\mstoo\mval{\mval_c}$, which would normally just
be omitted, indicates that the contract $\mval_c$ cannot be statically
judged to either pass or fail for $\mval$, which is why the predicate
must be applied.
This judgment and its closely related counterparts
($\proves\cdot\cdot\cdot$) and ($\refutes\cdot\cdot\cdot$), which
statically prove a value must or must not satisfy a given contract
respectively, are discussed in section~\ref{sec:smt}.  
If a flat contract can be statically proved or refuted, monitoring can
be short-circuited.
\begin{mathpar}
\inferrule[Monitor-Proved]{\deltamap\msto\sdepp{\mval_c}\sfalse\mstoo\\\\
  \proves\mstoo\mval{\mval_c}}
          {\achksimple{\mval_c}\mval, \msto
            \stdstep
            \mval, \mstoo}

\inferrule[Monitor-Refuted]{\deltamap\msto\sdepp{\mval_c}\sfalse\mstoo \\\\
  \refutes\mstoo\mval{\mval_c}}
          {\achksimple{\mval_c}\mval, \msto
            \stdstep
            \simpleblm{}{}, \mstoo}
\end{mathpar}

Monitoring a function contract against a function is interpreted the
standard $\eta$-expansion of contracts. 
\begin{mathpar}
\inferrule[Monitor-Function-Contract]{\deltamap\msto\sprocp\mval\strue{\mstoo}}
          {\achksimple{\slamc\mvar{\mval_c}\mcono}\mval, \msto
            \stdstep
            \slam\mvar{\achksimple\mcono{\spapp\mval{\achksimple{\mval_c}\mvar}}}, {\mstoo}}
\end{mathpar}

Monitoring a function contract against a non-function results in an error.
\begin{mathpar}
\inferrule[Monitor-Non-Function]{\deltamap\msto\sdepp{\mval_c}\strue{\msto_1}\\
  \deltamap{\msto_1}\sprocp\mval\sfalse{\msto_2}}
          {\achksimple{\mval_c}\mval, \msto
            \stdstep
            \simpleblm\relax\relax, {\msto_2}}
\end{mathpar}

When a dependent contract is represented by a address in the heap, we
look up the address and use the result.
\begin{mathpar}
\inferrule[Monitor-Unknown-Function-Contract]{\deltamap\msto\sdepp\mLab\strue{\msto_1}\\
  \deltamap{\msto_1}\sprocp\mval\strue{\msto_2}\\
           {\msto_2}(\mLab) = \slamc\mvar{\mval_c}\mcono }
  {\achksimple\mLab\mval, \msto
    \stdstep
    \slam\mvar{\achksimple\mcono{\sapp\mval{\achksimple{\mval_c}\mvar}}}, {\msto_2}}
\end{mathpar}

\subsubsection{Handling unknown values}
\label{sec:handling-unknown}

The final set of reduction rules concern unknown values and refinements.
\begin{mathpar}
\inferrule[Refine-Concrete]{\mpreval\neq\opaque}
          {\mpreval,\msto
            \stdstep            
            \with\mpreval\emptyset,\msto}

\inferrule[Refine-Unknown]{\mLab\notin\sdom\msto}
          {\opaque,\msto
            \stdstep
            \mLab,\sext\msto\mLab{\with\opaque\emptyset}}
\end{mathpar}
These two rules show reduction of pre-values, which initially have
no refinement.  If the pre-value is unknown, we additionally create a
fresh address and add it to the heap.

The \textsf{assume} form uses the \textsf{refine} metafunction to
update the heap of refinements to take into account the new
information; see figure~\ref{fig:semantics-ext} for the definition of \textsf{refine}.
\begin{mathpar}
\inferrule[Assume]{(\mstoo,\mvalo) = \refine\msto\mval{\mval_c}}
          {\assume\mval{\mval_c},\msto \stdstep \mvalo\!,\mstoo}
\end{mathpar}
Refinement is straightforward propagation of known contracts,
including expanding values known to be pairs via $\sconsp$ into pair
values, and values known to be function contracts (via $\sdepp$) into function contract values.

Finally, we must handle application of unknown values. The first rule
simply produces a new unknown value and heap address for the result
of a call. If the unknown function came with a contract, this new
unknown value will be refined by the contract via reduction.
\begin{mathpar}
\inferrule[Apply-Unknown]{\deltamap\msto\sprocp\mLab\strue\mstoo}
          {\sapp\mLab\mval, \msto
            \stdstep
            \mLab_a,\sext\mstoo{\mLab_a}\opaque}
          
\inferrule[Havoc]{\deltamap\msto\sprocp\mLab\strue\mstoo}
          {\sapp\mLab\mval, \msto \stdstep
            \sapp{\syntax{havoc}}\mval, \mstoo}
\end{mathpar}

The second reduction rule for applying an unknown function, labeled
\textit{Havoc}, handles the possible dynamic behavior of the unknown function.
A value passed to the unknown function may itself be a function
with behavior, whose implementation we hope to verify. This function
may further be invoked by the unknown function on unknown arguments.
To simulate this, we assume \emph{arbitrary} behavior from this unknown
function and put the argument in a so-called demonic context
implemented by the \textsf{havoc} operation, defined in a module added
to every program; the definition is given below.
\begin{gather*}
\begin{array}{@{}lr@{}}
\multicolumn{2}{@{}l@{}}{
\smod{\syntax{havoc}}
     {(\sarr{\sany}{\slam{\_}\sfalse}}}\\
\multicolumn{2}{@{}r@{}}{
\ \ (\slam\mvar{\amb(\{(\syntax{havoc}\ (\mvar\ \opaque)),
                                  (\syntax{havoc}\ (\scar\ \mvar)),
                                  (\syntax{havoc}\ (\scdr\
                                  \mvar))\})}))}
\end{array}
\\
\begin{array}{lcl}
\amb(\{\mexp\}) &=& \mexp\\
\amb(\{\mexp,\mexp_1,\dots\}) &=& \sif\opaque\mexp{\amb(\{\mexp_1,\dots\})}
\end{array}
\end{gather*}
The $\syntax{havoc}$ function never produces useful
results; its only purpose is to probe for all potential errors in the
value provided.  This context, and thus the $\syntax{havoc}$ module,
may be blamed for misuse of accessors and applications; we ignore
these, as they represent potential failures in \emph{omitted} portions
of the program.  Using $\syntax{havoc}$ is key to soundness in modular
higher-order static checking
~\cite{dvanhorn:Fahndrich2011Static,dvanhorn:TobinHochstadt2012Higherorder};
we discuss its role further in section~\ref{sec:soundness}.
Intuitively, precise execution of properly contracted functions prevents
$\syntax{havoc}$ from destroying every analysis.

\begin{figure}
\[
\begin{array}{@{}l@{\;}c@{\;}l@{\ }r@{}}

\refine\msto\mLab\mval&=&(\sext\mstoo\mLab\mvalo,\mLab)\\
\multicolumn{4}{r}{
\mbox{where }(\mstoo,\mvalo) = \refine\msto{\msto(\mLab)}\mval
}\\[1mm]

\refine\msto{\with\opaque{\vec\mval}}\sconsp
&=&
\multicolumn{2}{@{}l}{
(\sext{\sext\msto{\mLab_1}\opaque}{\mLab_2}\opaque,\with{\sconsc{\mLab_1}{\mLab_2}}{\vec\mval})
}\\
\multicolumn{4}{r}{
\mbox{where }{\mLab_1},{\mLab_2}\notin\sdom\msto
}\\[1mm]

\refine\msto{\with\opaque{\vec\mval}}\sdepp
&=&
\multicolumn{2}{@{}l}{
(\sext\msto\mLab\opaque,\with{\slamc\mvar\mLab\opaque}{\vec\mval})}\\
\multicolumn{4}{r}{
\mbox{where }\mLab\notin\sdom\msto
}\\[1mm] % FIXME: \opaque used to be `d`, but where did `d` come from?

\refine\msto{\with\mpreval{\vec\mval}}{\mval_i}
&=&
(\msto,\with\mpreval{{\vec\mval}\cup\{\mval_i\}})\\[1mm]

\end{array}
\]
\caption{Refinement for Symbolic \lang}
\label{fig:semantics-ext}
\end{figure}

\subsection{Primitive operations}
\label{sec:prim}

Primitive operations are the primary place where
unknown values in the heap are refined, in concert with
successful contract checks.
Figure~\ref{fig:delta} shows a representative excerpt of $\delta$'s
definition; the full definition is given in \techrep.

The first three rules cover primitive predicate checks.
Ambiguity never occurs for concrete values,
and an abstract value may definitely prove or refute the predicate
if the available information is enough for the conclusion.
If the proof system cannot decide a definite result for the predicate check,
$\delta$ conservatively includes \emph{both} answers in the possible results
and records assumptions chosen for each non-deterministic branch in
the appropriate heap.
The last three rules reveal possible refinements when applying
partial functions such as $\ssucc$, which fails when given non-numeric inputs.
This mechanism, when combined with the SMT-aided
proof system given below, is sufficient to provide the precision
necessary to prove the absence of contract errors.

\begin{figure}
\[
\begin{array}{rclr}
\delta(\msto,\moppred,\mval)
&\ni&
(\strue,\msto)
&\mbox{if }\proves\msto\mval\moppred\\[1mm]

\delta(\msto,\moppred,\mval)
&\ni&
(\sfalse,\msto)
&\mbox{if }\refutes\msto\mval\moppred\\[1mm]

\delta(\msto,\moppred,\mLab)
&\supseteq&
\multicolumn{2}{l}{\{(\strue,{\msto_t}),(\sfalse,{\msto_f})\}}\\
\multicolumn{4}{r}{
\mbox{if }\ambig\msto\mLab\moppred
\mbox{ and }(\msto_t,\_)=\refine\msto\mLab\moppred}\\
\multicolumn{4}{r}{
\mbox{ and }(\msto_f,\_)=\refine\msto\mLab{\neg\moppred}
}\\

\dots\\

\delta(\msto,\ssucc,\mnum)
&\ni&
\multicolumn{2}{l}{({\mnum+1},\msto)}\\[1mm]

\delta(\msto,\ssucc,\mval)
&\ni&
\multicolumn{2}{l}{(\mLab,\sext\mstoo\mLab{\with\opaque\snump})}\\
\multicolumn{4}{r}{\mbox{where }\deltamap\msto\snump\mval\strue\mstoo
                   \mbox{, }\mval \neq \mnum
                   \mbox{, and }\mLab \notin \mstoo }\\[1mm]

\delta(\msto,\ssucc,\mval)
&\ni&
\multicolumn{2}{l}{(\simpleblm\relax\Lambda,\mstoo)}\\
\multicolumn{4}{r}{\mbox{where }\deltamap\msto\snump\mval\sfalse\mstoo}\\[1mm]
\dots
\end{array}
\]
\caption{Selected primitive operations}
\label{fig:delta}
\end{figure}

\subsection{SMT-aided proof system}
\label{sec:smt}

Contract checking and primitive operations rely on a  proof
system to statically relate values and contracts.
We write $\proves\msto\mval{\mval_c}$ to mean value $\mval$ satisfies contract $\mval_c$,
where all addresses in $\mval$ are defined in $\msto$.  In other words, under any possible
instantiation of the unknown values in $\mval$, it would satisfy
$\mval_c$ when checked according to the semantics.
On the other hand, $\refutes\msto\mval{\mval_c}$ indicates that $\mval$ definitely fails $\mval_c$.
Finally, $\ambig\msto\mval{\mval_c}$ is a conservative answer when
information from the heap and refinement set is insufficient to draw a definite conclusion.
The effectiveness of our analysis depends on the precision of this
provability relation---increasing the number of contracts that can be
related statically to values prunes spurious paths
 and eliminates impossible error cases.

\subsubsection{Simple proof system}
\label{sec:simple-proof-system}

 A simple proof system can be obtained which returns definite answers
 for concrete values, uses heap refinements, and handles negation of
 predicates and disjointness of data types.
\[
\begin{array}{lr}
\multicolumn{2}{l}{\proves\msto\mnum\snump}\\
\refutes\msto\mnum\moppred
&\mbox{if }\moppred\in\mbox{\{\sconsp,\sprocp,etc.\}}\\[1mm]

\proves\msto{\with\mpreval{\vec\mval}}{\mval_i}
&\mbox{if }{\mval_i}\in\vec\mval\\
\refutes\msto{\with\mpreval{\vec\mval}}\moppred
&\mbox{if }{\neg \moppred}\in\vec\mval\\[1mm]

\proves\msto\mLab\mval
&\mbox{if }\proves\msto{\msto(\mLab)}\mval\\
\refutes\msto\mLab\mval
&\mbox{if }\refutes\msto{\msto(\mLab)}\mval\\
\ambig\msto\mLab\mval
&\mbox{if }\ambig\msto{\msto(\mLab)}\mval\\
\dots\\
\ambig\msto\mval{\mval_c}
&\mbox{(conservative default)}
\end{array}
\]

%% To save space,
%% we abbreviate $\syntax{\slam\mvar{(\sapp\moppred\mvar)}}$ with $\syntax{\moppred}$
%% and $\syntax{\slam\mvar{(\sapp\neg{(\sapp\moppred\mvar)})}}$ with $\syntax{\neg\moppred}$.
%
Notice that the proof system only needs to handle a small
number of well-understood contracts.
We rely on evaluation to naturally break down complex contracts
into smaller ones and take care of subtle issues such as
divergence and crashing.
By the time we have $\with\mpreval{\vec\mval}$,
we can assume all contracts in $\vec\mval$ have terminated with success on $\mpreval$.
With these simple and obvious rules, our system can already verify a
significant number of interesting programs.
With SMT solver integration, as described 
below, we can handle far more interesting
constraints, including relations between numeric values, without
requiring an encoding of the full language.

\subsubsection{Integrating an SMT solver}
\label{sec:integrating-smt}

We extend the simple provability relation by employing an external solver.

We first define the translation
$\tran{\cdot}$ from heaps and contract-value pairs into formulas
in solver $S$:
\[
\begin{array}{rcl}
\tran{\vec{(\mLab,\mcon)}} &=& \bigwedge \vec{\tran{\mLab:\mcon}} \\[2mm]
\tran{\mLab_1:\texttt{(>/c \(n\))}}
&=&
\mbox{\tt ASSERT \(\mLab_1\) > \(n\)}\\[1mm]
\mbox{\text{$\tran{\mLab_1:\texttt{(>/c \(\mLab_2\))}}$}}
&=&
\mbox{\tt ASSERT \(\mLab_1\) > \(\mLab_2\)}\\[1mm]
\mbox{\text{$\tran{\mLab:\texttt{(=/c (+ \(\mLab_1\)  \(\mLab_2\)))}}$}}
&=&
\mbox{\tt ASSERT \(\mLab\) = \(\mLab_1\) + \(\mLab_2\)}\\
 \multicolumn{3}{c}{\ldots}\\
\end{array}
\]
The translation of a heap is the conjunction of all formulas
generated from translatable refinements.
The function is partial,
and there are straightforward rules for translating specific pairs of $(\mLab:\mcon)$
where $\mcon$ are drawn from a small set of simple, well-understood contracts.
This mechanism is enough for the system to verify
many interesting programs because the analysis relies on evaluation
to break down complex, higher-order predicates.
Not having a translation for some contract $\mcon$ only reduces precision
and does not affect soundness.

Next, the
extension $(\vdash_S)$ is straightforward.  The old relation
$(\vdash)$ is refined by a solver $S$.  Whenever the basic relation
proves $\ambig\msto\mval\mcon$, we call out to the solver to try to
either prove or refute the claim:
\begin{mathpar}
\inferrule{\tran\msto \wedge \neg \tran{\mval:\mcon}\mbox{ is unsat}}
{\extProves\msto\mval\mcon}

\inferrule{\tran\msto \wedge \tran{\mval:\mcon}\mbox{ is unsat}}
{\extRefutes\msto\mval\mcon}
\end{mathpar}
The solver-aided relation uses refinements available on the heap to
generate premises $\tran\msto$.  Unsatisfiability of $\tran\msto
\wedge \neg \tran{\mval:\mcon}$ is equivalent to validity of
$\tran\msto \Rightarrow \tran{\mval:\mcon}$, hence value definitely
satisfies contract $\mcon$.  Likewise, unsatisfiability of $\tran\msto
\wedge \tran{\mval:\mcon}$ means {\tt \mval} definitely refutes {\tt
  \mcon}.  In any other case, we relate the value-contract pair to the
conservative answer.

\subsection{Program evaluation}
We give a reachable-states semantics to programs:
the initial program $\mprg$ is paired with an empty heap, and
\syntax{eval} produces all states in the
reflexive, transitive closure of the single-step reduction relation closed under
evaluation contexts.

\[
\begin{array}{l}
%% \syntax{inj\: }\colon\mprg\rightarrow\mPrg\\
%% \mbox{$\syntax{inj}({\vec\mmod}\mexp) = {\vec\mmod}\: \sclos\mexp\emptyset\: \emptyset$}\\[1mm]
%
\syntax{eval\: }\colon\mprg\rightarrow\pow(\mstate)\\
\mbox{$\syntax{eval}{(\vec\mmod \mexp)} =
   \{\mstate\ |\ {\vec\mmod}\ \vdash (\mexpo ; \mexp), \emptyset \multistdstep
                                               \mstate\}$}\\
\quad\mbox{where }\mexpo = \amb(\{\strue,\vec{\syntax{havoc} f}\})\mbox{, }
\syntax{\smod\mmodvar{\mval_c}\mval}\in\vec\mmod\\
%\mbox{and }\vec{\mmod}\ \mExp\ \msto = \mbox{inj(}\mprg\mbox{)}\\
\end{array}
\]
Modules with unknown definitions, which we call \emph{opaque},
complicate the definition of \syntax{eval}, since they may contain
references to concrete modules. If only the main module is considered,
an opaque module might misuse a concrete value in ways not visible to
the system.  We therefore apply \syntax{havoc} to each concrete module
before evaluating the main expression.

\subsection{Soundness}
\label{sec:soundness}
A program with unknown components is an abstraction of a fully-known
program. Thus, the semantics of the abstracted program should
approximate the semantics of any such concrete version.  In
particular, any behavior the concrete program exhibits should also be
exhibited by the abstract approximation of that program.

However, we must be precise as to which behaviors are relevant.
Suppose we have a single concrete module that links against a single
opaque module.  The semantics of this program should include all of
the possible behaviors, both good and bad, of the known module
assuming the opaque module always lives up to its contract.  We
exclude from consideration behaviors that cause the unknown module to
be blamed, since it is of course impossible to verify an unknown
program.
In other words, we try to verify the parts of the program that are
known, assuming arbitrary, but correct, behavior for the parts of the
program that are unknown.

For this reason, the precise semantic account of blame is crucial.
The demonic {\tt havoc} context can introduce blame of both the known
and unknown modules; since we can distinguish these parties, it is
easy to ignore blame of the unknown context.

In the remainder of this section, we formally define the approximation
relation and show that evaluation preserves the approximation,
i.e.~if program $q$ is an approximation of program $p$ ($q$ is like
$p$ but with potentially more unknowns), then the evaluation of $q$ is
an approximation of the evaluation of $p$.

\paragraph{Approximation:}
We define two approximation relations:
between modules and between pairs of expressions and heaps.

We write $\mstate \refines \mstate'$ to mean
``$\mstate'$ approximates $\mstate$,\!''
or ``$\mstate$ refines $\mstate'$\!,''
which intuitively means $\mstate'$ stands for a set of states
including $\mstate$.
For example, $(\syntax{1},\{\}) \refines (\mLab,\{\mLab\mapsto\opaque\})$.

One complication introduced by addresses is that a \emph{single} address in the abstract program
may accidentally approximate \emph{multiple} distinct values in the concrete one.
Such accidental approximations are not in general preserved by
reduction, as in the following example where
$(\mexp_1,\msto_1) \refines (\mexp_2,\msto_2)$:
\[
\begin{array}{lcll}
e_1 &=& \texttt{(if \sfalse\ 1 2)} & \sigma_1 = \{\}\\
e_2 &=& \texttt{(if $\mLab$\  $\mLab$ $\mLab$)} & \sigma_2 = \{\mLab \mapsto \bullet\}\\
\end{array}
\]
The abstract program does not continue to approximate the concrete one
in their next states:
\[
\begin{array}{lcll}
e_1 &\longmapsto& (\texttt{2}, \sigma_1') & \sigma_1' = \{\}\\
e_2 &\longmapsto& (\mLab, \sigma_2') & \sigma_2' = \{\mLab \mapsto \sfalse\}\\
\end{array}
\]

We therefore also define a ``strong'' version of the approximation relation, $\refines^{F}$,
where each address in the abstract program approximates exactly one
value in the concrete program, and this consistency is witnessed by
some function $F$ from addresses to values.  Then $e \refines e'$ means
that $\exists F. e \refines^{F} e'$
Since no such function exists between $e_1$ and $e_2$ above, $e_1
\not\refines^{F} e_2$ for any $F$, and therefore $e_1 \not\refines
e_2$.

Figure \ref{fig:approx} shows the important cases in the definition of $\refines^{F}$;
we omit structurally recursive rules.
All pre-values are approximated by $\opaque$,
and unknown values with contracts approximate  values that satisfy the
same contracts.
We extend the relation $\refines^F_{\vec\mmod}$ structurally to evaluation contexts $\mctx$,
point-wise to sequences, and to sets of program states.

In the following example, $(\mexp_1,\msto_1) \refines^F (\mexp_2,\msto_2)$,
where $F = \{\mLab_0 \mapsto \sfalse$,$\mLab_1 \mapsto 1$,$\mLab_2 \mapsto 2\}$:
\[
\begin{array}{@{}l@{\ }c@{\ }ll}
e_1 &=& \texttt{(if \sfalse\ 1 2)} & \sigma_1 = \{\}\\
e_2 &=& \texttt{(if $\mLab_1$\  $\mLab_2$ $\mLab_3$)} & \sigma_2 = \{\mLab_1 \mapsto
\bullet, \mLab_2 \mapsto \bullet, \mLab_3 \mapsto \bullet\}\\
\end{array}
\]

Notice that $F$'s domain is a superset of the domain of the heap
$\msto_2$.  In addition, our soundness result does not consider
additional errors that blame unknown modules or the \syntax{havoc}
module, and therefore we parameterize the approximation relation
$\refines^{F}_{\vec\mmod}$ with the module definitions $\vec\mmod$ to select the
opaque modules.  We omit these parameters where they are
easily inferred to ease notation.

\begin{figure}
\begin{mathpar}
\inferrule{(\with\mpreval{\vec\mval},\msto_1) \refines^F (\msto_2(\mLab),\msto_2)
\\\\ F(\mLab) = \with\mpreval{\vec\mval}}
          {(\with\mpreval{\vec\mval},\msto_1) \refines^F (\mLab,\msto_2)}
\qquad      
\inferrule{(\msto_1(\mLab_1),\msto_1) \refines^F (\msto_2(\mLab_2),\msto_2)
\\\\ F(\mLab_2) = \mLab_1}
          {(\mLab_1,\msto_1) \refines^F (\mLab_2,\msto_2)}

\inferrule{(\with{\mpreval_1}{\vec{\mval_1}},\msto_1)
            \refines^F
           (\with{\mpreval_2}{\vec{\mval_2}}),\msto_2)\\ 
           (\mval_c,\msto_1) \refines^F (\mval_d,\msto_2)}
          {(\with{\mpreval_1}{\vec{\mval_1} \cup \{\mval_c\}}, \msto_1)
           \refines^F
           (\with{\mpreval_2}{\vec{\mval_2} \cup \{\mval_d\}}, \msto_2)}
           
\inferrule{ }
          {(\mpreval,\msto_1) \refines^F (\opaque,\msto_2)}

\inferrule{(\with{\mpreval_1}{\vec{\mval_1}},\msto_1)
            \refines^F
           (\with{\mpreval_2}{\vec{\mval_2}},\msto_2)}
          {(\with{\mpreval_1}{\vec{\mval_1} \cup \mval}, \msto_1)
           \refines^F
           (\with{\mpreval_2}{\vec{\mval_2}}, \msto_2)}
            
\inferrule{\smod\mmodvar{\mpreval_c}\opaque \in \vec\mmod\\ or \\
           \mmodvar \in \{\dagger,\syntax{havoc}\}}
          {(\sblm\mmodvar\mmodvaro,\msto_1) \refines^F_{\vec\mmod} (\mexp,\msto_2)}

  %% Structural cases remain.
\end{mathpar}
\caption{Selected Approximation Rules}
\label{fig:approx}
\end{figure}

%\paragraph{Soundness of \lang}
%
With the definition of approximation in hand, we are now in a position
to state the main soundness theorem for the system.

\begin{theorem}[Soundness of Symbolic \lang]
\label{thm:main}\ \\
If $\mprg \refines^F_{\vec\mmod} \mprgo$ where $\mprgo = \mvmod \mexp$
and $\mstate \in \eval(\mprg)$,
then there exists some $\mstate' \in \eval(\mprgo)$
such that $\mstate \refines^{F'}_\mvmod \mstate'$\!.
\end{theorem}
\noindent
We defer all proofs to the
\iftechreport
appendix
\else
technical report
\fi
for  space.

\subsection{Verification and the blame theorem}

We can now define verification as a simple corollary of soundness.
First we defined when a module is \emph{verified} by our approach.

\begin{definition}[Verified module]
\label{def:verified}\ \\
A module $\smod{\mmodvar}{\mpreval_c}{\mpreval} \in \mprg$
 is \emph{verified in} $\mprg$
if $\mpreval \neq \opaque$ and $\eval(\mprg) \not\ni \sblm\mmodvar\relax$\!.
\end{definition}

Now, by soundness, $\mmodvar$ is always safe.

\begin{theorem}[Verified modules can't be blamed]\ \\
  If a module named $\mmodvar$ is verified in $\mprg$, then for any
  concrete program $\mprgo$ for which $\mprg$ is an abstraction, 
  $\eval(\mprgo) \not\ni \sblm\mmodvar\relax$\!.
\end{theorem}

\subsection{Taming the infinite state space}
\label{sec:termination}
A naive implementation of the above semantics will diverge for many programs.
Consider the following example:
\begin{alltt}
(define (fact n)
  (if (= n 0) 1 (* n (fact (- n 1)))))
(fact \(\opaque\))
\end{alltt}

\noindent
Ignoring error cases, it eventually reduces non-deterministically to all of the following:
\begin{align*}
\mbox{\tt 1} & \mbox{ if }\mLab\sb{n} \mapsto \mbox{\tt 0}\\
\mbox{\tt (* \(\mLab\sb{n}\) 1)} & \mbox{ if }\mLab\sb{n}\not\mapsto {\tt
  0}\text{, }\mLab\sb{n-1} \mapsto {\tt 0}\\
\mbox{\tt (* \(\mLab\sb{n}\) (* \(\mLab\sb{n-1}\) (fact
  \(\mLab\sb{n-1}\))))} & \mbox{ if } \mLab\sb{n}\text{, }\mLab\sb{n-1} \not\mapsto {\tt 0}
\end{align*}
where $\mLab_{n-1}$ is a fresh address resulting from subtracting $\mLab_n$ by one.
The process continues with $\mLab_{n-2}$, $\mLab_{n-3}$, etc.
This behavior from the analysis happens because it attempts to approximate
\emph{all} possible concrete substitutions to abstract values.
Although $\syntax{fact}$ terminates for all concrete naturals,
there are an infinite number of those: $\mLab_n$ can be $\syntax{0}$,
$\syntax{1}$, $\syntax{2}$, and so on.

To enforce termination for all programs,
we can resort to well-known techniques such as finite state or
pushdown abstractions~\cite{dvanhorn:VanHorn2012Systematic}.
But often those are overkill at the cost of precision.
Consider the following program:
\begin{alltt}
   (let* ([id (\(\slambda\) (x) x)] [y (id 0)] [z (id 1)])
     (< y z))
\end{alltt}
where a monovariant flow analysis such as 0CFA~\cite{dvanhorn:shivers-88} thinks $\syntax{y}$ and $\syntax{z}$ can be both $\syntax{0}$ and $\syntax{1}$,
and pushdown analysis thinks $\syntax{y}$ is 0 and $\syntax{z}$ is either $\syntax{0}$ or $\syntax{1}$.
For a concrete, straight-line program, such imprecision seems unsatisfactory.
We therefore aim for an analysis that provides exact execution for
non-recursive programs and
 retains enough invariants to verify interesting properties of recursive ones.
The analysis quickly terminates for a majority of programming patterns with decent precision,
although it is not guaranteed to terminate in the general case---see
section~\ref{sec:impl} for empirical results.

One technical difficulty is that the semantics of contracts prevents us from using
a recursive function's contract directly as a loop invariant,
because contracts are only boundary-level enforcement.
It is unsound to assume returned values of internal calls
can be approximated by contracts,
as in \texttt{f} below:
\begin{alltt}
(f : nat? \(\conarrow\) nat?)
(define (f n) (if (= n 0) "" (str-len (f (- n 1)))))
\end{alltt}
If we assume the expression \texttt{(f (- n 1))} returns a number as specified in the contract,
we will conclude \texttt{f} never returns, and is blamed
either for violating its own contract by returning a string,
or for applying $\syntax{str-len}$ to a number.
However, $\syntax{f}$ returns $\syntax{0}$ when applied to $\syntax{1}$.
To soundly and precisely approximate this semantics in the absence of types,
we recover  data type invariants by execution.

\begin{figure*}[t]
\begin{mathpar}
\inferrule{\mctx \neq \mctx_1[\srt{\msto_0}{\slam\mvar\mexp}{\mval_0}{\mctx_k}]
           \mbox{ for any } \mctx_1, \mctx_k, \msto_0, \mval_0}
          {\squadruple\mktab\mmtab{\mctx[\spapp{\slamp\mvar\mexp}{\mval}]}\msto \stdstep
           \squadruple\mktab\mmtab{\mctx[\srt{\msto}{\slam\mvar\mexp}{\mval}{\subst\mval\mvar\mexp}]}{\msto}}
          
\inferrule{\vec{\spair{\mLab_o}{\mLab_n}} = \mF\\
           \mstoo = \msto{\vec{[\mLab_n \mapsto \sapprox{\msto_0(\mLab_0)}{\msto(\mLab_0)}]}}}
          {\squadruple\mktab\mmtab{\mctx[\sblur\mF{\msto_0}{\mval_0}\mval]}\msto \stdstep
           \squadruple\mktab\mmtab{\mctx[{\sapprox{\mval_0}\mval}]}{\mstoo}}
  
\inferrule{\mctx = \mctx_1[\srt{\msto_0}{\slam\mvar\mexp}{\mval_0}{\mctx_k}]
           \mbox{ for some }{\mctx_1},{\mctx_k},{\msto_0},{\mval_0}\\
           \spair\msto\mval \not\refines \spair{\msto_0}{\mval_0}\\
           \mval_1 = \sapprox{\mval_0}\mval}
          {\squadruple\mktab\mmtab{\mctx[\spapp{\slamp\mvar\mexp}{\mval}]}\msto \stdstep
           \squadruple\mktab\mmtab{\mctx[\srt{\msto}{\slam\mvar\mexp}{\mval_1}{\subst{\mval_1}\mvar\mexp}]}\msto}
           
\inferrule{\mctx = \mctx_1[\srt{\msto_0}{\mval_f}{\mval_0}{\mctx_k}]
           \mbox{ for some }{\mctx_1},{\mctx_k},{\msto_0},{\mval_0}\\
           \spair\msto\mval \refines^F \spair{\msto_0}{\mval_0}\\
           \mktabo = \mktab \sqcup [\striple{\msto_0}{\mval_f}{\mval_0}
                            \mapsto \squadruple\mF\msto{\mctx_1}{\mctx_k}]\\
           \spair{\mval_a}{\msto_a} \in \mmtab[\striple{\msto_0}{\mval_f}{\mval_0}]\\
           \mstoo = \msto{\vec{[\mLab_n \mapsto \msto_a[\mLab_o]]}}
           \mbox{ where }{\vec{\spair{\mLab_o}{\mLab_n}} = \mF}}
		  {\squadruple\mktab\mmtab{\mctx[\spapp{\mval_f}{\mval}]}\msto \stdstep
		   \squadruple\mktabo\mmtab{\mctx_1[\srt{\msto_0}{\mval_f}{\mval_0}
		                          {\sblur{\mF}{\msto_a}{\mval_a}{\mctx_k[\mval_a]}}]}\mstoo}

\inferrule{\mmtabo = \mmtab \sqcup [\striple{\msto_0}{\mval_f}{\mval_0}
                                    \mapsto \spair\mval\msto]}
		  {\squadruple\mktab\mmtab{\mctx[\srt{\msto_0}{\mval_f}{\mval_0}\mval]}\msto
		  \stdstep
		  \squadruple\mktab\mmtabo{\mctx[\mval]}\msto}

\inferrule{\mmtabo = \mmtab \sqcup [\striple{\msto_0}{\mval_f}{\mval_0}
                                    \mapsto \spair\mval\msto]\\
           \squadruple\mF{\msto_k}{\mctx_1}{\mctx_k} \in
                   \mktab[\striple{\msto_0}{\mval_f}{\mval_0}]\\
           \mstoo_k = \msto_k{\vec{[\mLab_n \mapsto \msto(\mLab_o)]}}
           \mbox{ where }{\vec{\spair{\mLab_o}{\mLab_n}} = \mF}}
		  {\squadruple\mktab\mmtab{\mctx[\srt{\msto_0}{\mval_f}{\mval_0}\mval]}\msto
		  \stdstep
		  \squadruple\mktab\mmtabo{\mctx_1[\srt{\msto_0}{\mval_f}{\mval_0}{\sblur\mF\msto\mval{\mctx_k[\mval]}}]}{\mstoo_k}}

\end{mathpar}
\caption{Summarizing Semantics}
\label{fig:semantics-terminating}
\end{figure*}

\begin{figure}[t]
\[
\begin{array}{@{}l@{\quad}r@{\;}c@{\;}l}
\mbox{Expressions} & \mexp &+\!\!=& \srt\msto\mval\mval\mexp\ |\ \sblur\mF\msto\mval\mexp \\[.8mm]
\mbox{Values} & \mval &+\!\!=& \srec\mvar{\vec\mval}\ |\ \sref\mvar\\[.8mm]
\mbox{Evaluation contexts} & \mctx &+\!\!=& \srt\msto\mval\mval\mctx\ |\ \sblur\mF\msto\mval\mctx\\[1mm]
\mbox{Context memo tables} & \mktab &::=& \vec{((\msto,\mval,\mval), \vec{(\mF,\msto,\mctx,\mctx)})}\\[1mm]
\mbox{Value memo tables} & \mmtab &::=& \vec{((\msto,\mval,\mval), \vec{(\mval,\msto)})}\\[1mm]
\mbox{Renamings} & \mF &::=& \vec{\sconsc\mLab\mLab}
\end{array}
\]
\caption{Syntax extensions for approximation}
\label{fig:syntax-ext2}
\end{figure}

\paragraph{Summarizing function results:}
To accelerate convergence, we modify the application rules as follows.
At each application, we decide whether execution should step to
the function's body or wait for known results from other branches.
When an application \texttt{(f v)} reduces to
a similar application, we plug in known results instead of executing
\texttt{f}'s body again, avoiding the infinite loop.
Correspondingly, when \texttt{(f v)} returns,
we plug the new-found answer into contexts that need the result of \texttt{(f v)}.
The execution continues until it has a set soundly
describing the results of \texttt{(f v)}.

To track information about application results and waiting contexts,
we augment the execution with two global tables $\mmtab$ and $\mktab$
as shown in figure \ref{fig:syntax-ext2}.
We borrow the choice of metavariable names from work on concrete summaries~\cite{Johnson:HOPA}. %TODO

A value memo table $\mmtab$ maps each application to known results and accompanying refinements.
Intuitively, if $\mmtab(\msto,{\mval_f},{\mval_x}) \ni (\mval,\msto')$
then in some execution branch, there is an application $\spapp{\mval_f}{\mval_x},\msto
\multistdstep (\mval,\msto')$.

A context memo table $\mktab$ maps each application to contexts waiting for its result.
Intuitively, $\mktab(\msto,{\mval_f},{\mval_x}) \ni (\mF,\mstoo,{\mctx_1},{\mctx_k})$
means during evaluation,
some expression ${\mctx_1}[\srt\msto{\mval_f}{\mval_x}{[{\mctx_k}[\spapp{\mval_f}{\mval_z}]]}]$
with heap $\mstoo$ is paused because applying $\spapp{\mval_f\,}{\mval_z}$ under assumptions
in $\mstoo$ is subsumed by applying $\spapp{\mval_f\,}{\mval_x}$ under assumptions in
$\msto$ up to consistent address renaming specified by function $\mF$.

To keep track of function applications seen so far,
we extend the language with the expression $\srt\msto\mval\mvalo\mexp$,
which marks $\mexp$ as being evaluated as the result of applying
$\mval$ to $\mvalo$, but otherwise behaves like $\mexp$.
The expression $\sblur\mF\msto\mval\mexp$, whose detailed role is discussed below,
approximates  $\mexp$ under guidance from a ``previous'' value $\mval$.
%These two new expression forms result in corresponding new evaluation context forms.

Finally, we add  recursive contracts
$\srec\mvar{\vec\mval}$ and recursive  references $\sref\mvar$
% say "approximate" instead of "describe" inductive sets
for approximating inductive sets of values.
For example, $\srec\mvar{\{\sempty, \sconsc{\with\opaque\snatp}{\sref\mvar}\}}$
approximates all finite lists of naturals.

A state in the approximating semantics with summarization consists of
global tables $\mktab$, $\mmtab$, and a set $S$ of explored
states $\vec{\mstate}$\!.

Reduction now relates tables $\mktab$, $\mmtab$, and a set of states
$\vec{\mstate}$
to new tables $\mktabo$, $\mmtabo$ and a new set of states $\vec{\mstate}'$.
We define a relation $\striple{\mktab}{\mmtab}{\mstate} \longmapsto \striple{\mktab}{\mmtab}{\mstate}$, and then lift this relation point-wise to sets of states.
Figure \ref{fig:semantics-terminating} only shows rules that use the
global tables or new expression forms.

In the first rule, if an application $\spapp{\slamp\mvar\mexp}\mval$ is not previously seen,
execution proceeds as usual, evaluating expression $\mexp$ with
$\mvar$ bound to $\mval$, but marking this expression using \textsf{rt}.

Second, if a previous application of $\spapp{\slamp\mvar\mexp}{\mval_0}$ results in application
of the same function to a new argument $\mval$,
we approximate the new argument before continuing.
Taking advantage of knowledge of the previous argument,
we guess the transition from the $\mval_0$ to $\mval$ and
heuristically emulate
an arbitrary amount of such transformation using the $\oplus$ operator.
For example, if $\mval_0$ is $\sempty$ and $\mval$ is $\sconsc{\syntax{1}}\sempty$,
we approximate the latter to $\srec\mvar{\{\sempty, \sconsc{\syntax{1}}{\sref\mvar}\}}$,
denoting a list of \texttt{1}'s.
If a different number is later prepended to the list, it is approximated to a list of numbers.
The $\oplus$ operator should work well in common cases
and not hinder convergence in the general case.
Failure to give a good approximation to a value results in non-termination
but does not affect soundness.

Third, when an application results in a similar one
with potentially refined arguments,
we avoid stepping into the function body and use known results from table $\mmtab$ instead.
In addition, we refine the current heap to make better use of assumptions
about the particular ``base case''.
We also remember the current context as one waiting
for the result of such application.
To speed up convergence, apart from feeding a new answer $\mval_a$ to the context,
we wrap the entire expression inside $\sblur\mF\msto\mval{[\;]}$
to approximate the future result.

The fourth rule in figure \ref{fig:semantics-terminating} shows reduction
for returning from an application.
Apart from the current context, the value is also returned to any known
context waiting on the same application.
Besides, the value is also remembered in table $\mmtab$.
The resumption and refinement are analogous to the previous rule.

Finally, expression $\sblur\mF\msto{\mval_0}\mval$ approximates value $\mval$ under guidance
from the previous value $\mval_0$
and also approximates values on the heap from observation
of the previous case. 
Overall, the approximating operator $\oplus$ occurs in three places:
arguments of recursive applications, result of recursive applications,
and abstract values on the heap when recursive applications return.
%% The bottom of figure \ref{fig:semantics-terminating} shows representative
%% rules for implementing this operator.
%% In the worst case when no rule applies,
%% we do no approximation to the value.
Empirical results for our tool are presented in section~\ref{sec:impl}.

%For the \texttt{fact} example above, the two similar applications
%are $(\syntax{fact}\ \mLab_n)$ with heap $\{\mLab_n \mapsto \with\opaque\snatp\}$
%and $(\syntax{fact}\ \mLab_{n-1})$ with heap
%$\{\mLab_n \mapsto \with\opaque\snatp; \mLab_{n-1} \mapsto \with\opaque\snatp\}$.
%We pause the latter application and remember the following entry in table $\mktab$:
%
%\[
%\begin{array}{lr}
%\\multicolumn{2}{l}{
%\mktab(\{\mLab_n \mapsto \with\opaque\snatp\},\text{fact},\mLab_n) =
%(F,\mstoo,\mctx_0,\mctx_k)}\\
%\multicolumn{2}{r}{
%\mbox{where F = $\{\mLab_n \mapsto \mLab_{n-1}\}$}}\\
%\multicolumn{2}{r}{
%\mstoo = \{\mLab_n \mapsto \with\opaque\snatp;\ \mLab_{n-1} \mapsto \with\opaque\snatp\}}\\
%\multicolumn{2}{r}{
%\mctx_0 = [ ]}\\
%\multicolumn{2}{r}{
%\mctx_k = (*\ \mLab_n\ [\ ])}\\
%\end{array}
%\]

\paragraph{Soundness of summarization:}

A system $(\mktab,\mmtab,S)$ approximates a state $\mstate$
if that state can be recovered from the system through approximation
rules.  The crucial rule, given below, states that if the system
$(\mktab,\mmtab,S)$ already approximates expression $\mexp$ and the
application $\spapp{\mval_f}{\mval_x}$ is known to reduce to $\mexp$,
then $(\mktab,\mmtab,S)$ is an approximation of $\mctx_k[\mexp]$
where $\mctx_k$ is a waiting context for this application.
\begin{mathpar}
\inferrule{\srt{\_}\mval{\_}{[\;]} \notin \mctx_0\\
		   (\mval_x,\msto) \refines (\mval_z,\mstoo)\\
		   (\mval_y,\msto) \refines (\mval_z,\mstoo)\\
		   \mktab(\mstoo\!,\mval,\mval_z) \ni (F,\mstoo\!,\mctx'_0,\mctx'_k)\\
		   (\mctx_0,\msto) \refines (\mctx'_0,\mstoo)\\
		   (\mctx_k,\msto) \refines (\mctx'_k,\mstoo)\\
		   (\mctx_0[\srt{\msto_1}\mval{\mval_y}\mexp],\msto) \refines (\mktab,\mmtab,S)}
          {(\mctx_0[\srt{\msto_0}\mval{\mval_x}{\mctx_k[\srt{\msto_1}\mval{\mval_y}\mexp]}],\msto)
           \refines (\mktab,\mmtab,S)}
\end{mathpar}
As a consequence, summarization properly handles repetition of waiting contexts,
and gives results that approximate any number of recursive applications.
We refer readers to the appendix of \techrep for the full definition of the
approximation relation.

With this definition in hand, we can state the central lemma to
establish the soundness of the revised semantics that uses
summarization.
\begin{lemma}[Soundness of summarization]
\label{thm:summ}\ \\
If $\mstate \refines (\mktab,\mmtab,S)$
and $\mstate \stdstep \mstate'$\!,
then $(\mktab,\mmtab,S) \multistdstep (\mktabo,\mmtabo,S')$
such that $\mstate' \refines (\mktabo,\mmtabo,S')$.
\end{lemma}
\noindent

The proof is given in \techrep.
With this lemma in place, it is straightforward to replay the proof of
the soundness and blame theorems.

\begin{table*}[t]
\begin{center}
\begin{tabular}
{@{\;}l @{\;}|@{\;} r@{\;} |@{\;} r@{\;} | @{\;}c@{\;} |@{\;} c @{\;}|@{\;} r @{\;}|@{\;} c@{\;} }
   \multicolumn{2}{l}{} && \multicolumn{2}{@{\;}c@{\;}|@{\;}}{Simple} & \multicolumn{2}{@{\;}c@{\;}}{\SCV} \\ 
    Corpus & Lines & Checks & Time (ms) & False Pos. & Time (ms) &
    False Pos. \\ \hline 
\textrm{Occurrence Typing}%~\cite{dvanhorn:TobinHochstadt2010Logical}
 & 115 & 142 & 155.8 & 15 & 8.9 & 0\\
\textrm{Soft Typing}%~\cite{dvanhorn:Wright1997Practical}
 & 134 & 177 & 424.5 & 9 & 380.3 & 0\\
\textrm{Higher-order Recursion Schemes}%~\cite{dvanhorn:Kobayashi2011Predicate}
 & 301 & 467 & $\infty$ & $\leq 94$ & 3,253.8 & 4\\
\textrm{Dependent Refinement Types}%~\cite{dvanhorn:Terauchi2010Dependent}
& 69 & 116 & $\infty$ & $\leq 66$ & 193.0 & 1\\
\textrm{Higher-order Symbolic Execution}%~\cite{dvanhorn:TobinHochstadt2012Higherorder}
 & 236 & 319 &
$\infty$ & $\leq 19$ & 4,372.7 & 1\\[2mm]
{Student Video Games} & & & & &\\
\quad Snake & 202 & 270 & 9,452.5     &   0 & 3.008.8  & 0\\
\quad Tetris  & 308 & 351 & $\infty$ & - & 27,408.5 & 0\\
\quad Zombie & 249 & 393 & $\infty$ & - & 11,335.9 & 0
\end{tabular}
\end{center}
\caption{Summary benchmark results.  (See \techrep for detailed
  results.)}
\label{tbl:summary}
\end{table*}

\section{Implementation and evaluation}
\label{sec:impl}

To validate our approach, we implemented a static contract checking
tool, \SCV, based on the semantics presented in section~\ref{sec:concrete},
 along with a number of implementation
extensions for increased precision and performance.  We then applied 
\SCV to a wide selection of programs drawn from the literature on
verification of higher-order programs, and report on the results.

The source code for \SCV and all benchmarks are available along with instructions
on reproducing the results we report here:
\begin{center}
\href{https://github.com/philnguyen/soft-contract/}{\tt github.com/philnguyen/soft-contract}
\end{center}
In order to quantify the importance of the techniques presented in
this paper, we also created a simpler tool which omits the key contributions of this work.
This slimmed down system, which we refer to as ``Simple'' below, (a) does not
call out to a solver, but relies on remembering seen contracts, (b)
never refines the contracts associated with a heap address, but
splits disjunctive contracts and unrolls recursive contracts,
and (c) does not use our technique for summarizing repeated
context. To enable a full comparison on all benchmarks, the Simple tool
supports first-class contracts.
This simpler system is extremely similar to that presented by our
 earlier work~\cite{dvanhorn:TobinHochstadt2012Higherorder}, but
 works on all of our benchmarks.

\paragraph{Implementation extensions:}
\SCV supports an extended language beyond that presented
in section~\ref{sec:concrete} in order to handle realistic programs.
First, more base values and primitive operations are supported, such
as strings and symbols (and their operations), although we do not yet
use a solver to reason about values other than integers.
Second, data structure definitions are allowed at the top-level.
Each new data definition induces a corresponding (automatic) extension
to the definition of $\syntax{havoc}$ to deal with the new class of
data.
Third, modules have multiple named exports, to handle the examples
presented in section~\ref{sec:main-idea}, and can include local,
non-exported, definitions.
Fourth, functions can accept multiple arguments and can be defined to
have variable-arity, as with \texttt{+}, which accepts arbitrarily
many arguments.  This introduces new possibilities of errors from
arity mismatches.
Fifth, a much more expressive contract language is implemented with
$\syntax{and/c}$, $\syntax{or/c}$, $\syntax{struct/c}$, $\syntax{\(\mu\)/c}$
for conjunctive, disjunctive, data type, and recursive contracts, respectively.
Sixth, we provide solver back-ends for both
CVC4~\cite{dvanhorn:Barrett2011CVC4} and
Z3~\cite{dvanhorn:DeMoura2008Z3}.
%% As is standard in the semantics for disjunctive contracts, we
%% require that at most one disjunct is higher-order (contains a function
%% contract) and without loss of generality we take it to be the right disjunct.
%
%% Finally, we use CVC4~\cite{dvanhorn:Barrett2011CVC4} as an external solver
%% for proving arithmetic properties.

\paragraph{Evaluating on existing benchmarks:}
To evaluate the applicability of \SCV to a wide variety of
challenging higher-order contract checking problems, we collect
examples from the following sources: programs that make use of
control-flow-based typing from work on \textbf{occurrence
typing}~\cite{dvanhorn:TobinHochstadt2010Logical}, programs from work
on \textbf{soft typing}, which uses flow analysis to check the preconditions of
operations~\cite{dvanhorn:Cartwright1991Soft}, programs with
sophisticated specifications from work on model checking  \textbf{higher-order 
recursion schemes}~\cite{dvanhorn:Kobayashi2011Predicate}, programs from work on
inference of \textbf{dependent refinement
types}~\cite{dvanhorn:Terauchi2010Dependent}, and programs with rich
contracts from our prior
work on \textbf{higher-order symbolic execution}~\cite{dvanhorn:TobinHochstadt2012Higherorder}.
We also evaluate \SCV on three interactive student video games built for
a first-year programming course: \textbf{Snake}, \textbf{Tetris}, and \textbf{Zombie}. These
programs were all originally written as sample solutions, following
the style expected of students in the course. Of these,
Zombie is the most interesting: it was originally an object-oriented
program, translated using the encoding seen in
section~\ref{sec:putting-it-all-together-example}.

We present our results in summary form in table~\ref{tbl:summary}, grouping each of the above
sets of benchmark programs; expanded forms of the tables are provided
in \techrep.  The table shows total line count
(excluding blank lines and comments)
and the number of static occurrences of contracts and primitives requiring dynamic checks
such as function applications and primitive operations.  These
checks can be eliminated if we can show that they never fail; this has
proven to produce significant speedups in practice, even without
eliminating more expensive contract
checks~\cite{dvanhorn:TobinHochstadt2011Languages}.

The table reports time (in milliseconds) and the number of false
positives for \SCV and our
reduced system omitting the key contributions of this work (labeled ``Simple'');
 ``$\infty$'' indicates a timeout after 5 minutes.

A false positive is a contract violation reported by the analysis,
but by human inspection, cannot happen.
The programs we consider are all known not to have contract errors,
and thus all potential errors are false positives.

In cases where a tool times out, we give an upper bound on the number
of false positive error reports. For example, the Simple system times
out on two of the higher-order recursion scheme programs, meaning
that if it were to complete, it would report \emph{at most} 94 false
positives, counting all contract checks from the two programs on which
it times out, and the measured false positives on the programs where
it completes.

Execution times are measured on a Core i7 2.7GHz laptop with 8GB of RAM.

\paragraph{Discussion:}
  First, \SCV works on
a benchmarks for a range of previous static analyzers, from type systems to model checking to
program analysis.  

  Second, most programs are analyzed in a reasonable amount of time;
  the longest remaining analysis time is under 30 seconds.
This demonstrates that although the termination acceleration method of
section~\ref{sec:termination} is not fully general, it is effective for
many programming patterns.  For example, \SCV terminates with good
precision on \texttt{last} from
\citet{dvanhorn:Wright1997Practical}, which hides recursion behind the
Y combinator.

Third, across all benchmarks, over 99\% (2329/2335) of the contract
checks are statically verified, enabling the elimination both of small
checks for primitive operations and expensive contracts; see below for
timing results.  This result emphasizes the value of static contract
checking: gaining confidence about correctness from expensive contracts
without actually incurring their cost.

Overall, our experiments show that our approach is able to discover and
use invariants implied by conditional flows of control and contract
checks.  Obfuscations such as multiple layers of abstractions or
complex chains of aliases do not impact precision (a common
shortcoming of flow analysis).

Our approach does not yet give a way to prove deep structural
properties expressed as dependent contracts such as ``map over a list
preserves the length'' or ``all elements in the result of filter
satisfy the predicate'', resulting in the false positives seen in
table~\ref{tbl:summary}. However, it can already be used to verify
many interesting programs because often safety questions depend only
on knowledge of top-level constructors.  Examples of these patterns
appear in programs from \citet{dvanhorn:Kobayashi2011Predicate} for programs such as {\tt reverse}
(see also \S\ref{sec:approxctxs}), {\tt nil}, and {\tt mem}.

Finally, soft contract verification is  more broadly applicable than the systems from
which our benchmarks are drawn, which typically are successful only on
their own benchmarks.  For example, type systems such as occurrence
typing~\cite{dvanhorn:TobinHochstadt2010Logical} cannot verify any
non-trivial contracts, and most soft typing systems do not consider
contracts at all.  Systems based on higher-order
model-checking~\cite{dvanhorn:Kobayashi2011Predicate}, and dependent
refinement types~\cite{dvanhorn:Terauchi2010Dependent} assume a typed
language; encoding our programs using large disjoint unions produces
unverifiable results.

This broad applicability is why we are not able to directly
compare \SCV to these other systems across all benchmarks.
Instead, the Simple system serves as a benchmark for a system which
does not contain our primary contributions.

\paragraph{Contract optimization:} 
We also report speedup results for the three most complex programs in
our evaluation, which are interactive games designed for first-year
programming courses (Snake, Tetris, and Zombie).  For each, we
recorded a trace of input and timer events while playing the game, and
then used that trace to re-run the game (omitting all graphical
rendering) both with the contracts that we verified, and with the 
contracts manually removed. Each
game was run 100 times in both modes; the total time is presented below.
%\begin{table}[b]
\begin{center}
\begin{tabular}{l|r|r}
    Program & Contracts On (ms) & Contracts Off (ms)\\ \hline
    snake & 475,799 & 59 \\
    tetris & 1,127,591 & 186 \\
    zombie & 12,413 & 1,721 \\
\end{tabular}
\end{center}
%\caption{Contract overhead}
%\label{tbl:overhead}
%\end{table}

The timing results are quite striking---speedup ranges from over 5x to
over 5000x. This does not indicate, of course, that speedups of these
magnitudes are achievable for real programs.  Instead, it shows that
programmers avoid the rich contracts we are able to verify, because of
their unacceptable performance overhead. Soft contract verification
therefore enables programmers to write these specifications without
the run-time cost.

The difference in timing between Zombie and the other two games is
intriguing because Zombie uses higher-order dependent contracts
extensively, along the lines of \texttt{vec/c} from
section~\ref{sec:putting-it-all-together-example}, which intuitively should
be more expensive.
An investigation reveals that most of the cost comes from monitoring
flat contracts, especially those that apply to data structures.  For
example, in Snake, disabling {\tt posn/c}, a simple contract that
checks for a {\tt posn} struct with two numeric fields, cuts the
run-time by a factor of 4.  This contract is repeatedly applied to
every such object in the game.
In contrast, higher-order contracts, as in the object encodings used
in Zombie, delay contracts and avoid this repeated checking.

\section{Related work}
\label{sec:related}

% \cite{meseguer-rosu-2013-ic}

In this section, we relate our work to four related
strands of research:
soft-typing, static contract verification, refinement types, and model
checking of recursion schemes.

\paragraph{Soft typing:}
Verifying the preconditions of primitive operations can be seen as a
weak form of contract verification and soft typing is a well studied
approach to this kind of
verification~\cite{dvanhorn:Cartwright1996Program}.  There are two
predominant approaches to soft-typing: one is based on a
generalization of Hindley-Milner type
inference~\cite{dvanhorn:Cartwright1991Soft,
  dvanhorn:Wright1997Practical, dvanhorn:Aiken1994Soft}, which views
an untyped program as being embedded in a typed one and attempts to
safely eliminate coercions~\cite{dvanhorn:Henglein1994Dynamic}.  The
other is founded on set-based abstract interpretation of programs
\cite{dvanhorn:Flanagan1996Catching,dvanhorn:Flanagan1999Componential}.
Both approaches have proved effective for statically checking
preconditions of primitive operations, but the
approach does not scale to checking pre- and post-conditions of
arbitrary contracts.
For example, Soft Scheme~\cite{dvanhorn:Cartwright1991Soft}  is not path-sensitive
and does not reason about arithmetic, thus it is unable to
verify many of the occurrence-typing or higher-order recursion scheme examples considered in the evaluation.

%% In particular, these approaches do not allow path-sensitive
%% reasoning with complex usage of type predicates and aliasing
%% \cite{dvanhorn:TobinHochstadt2010Logical}
%% and do not verify richer properties such as arithmetics,
%% or those requiring cross-function reasoning such as \texttt{reverse}
%% or \texttt{isnil} from table \ref{tbl:hors}.

\paragraph{Contract verification:}
Following in the set-based analysis tradition of soft-typing, there
has been work extending set-based analysis to languages with
contracts~\cite{dvanhorn:Meunier2006Modular}.  This work shares the
overarching goal of this paper: to develop a static contract checking
approach for components written in untyped languages with contracts.
However the work fails to capture the control-flow-based type
reasoning essential to analyzing untyped programs and is unsound (as
discussed by~\citet{dvanhorn:TobinHochstadt2012Higherorder}).
Moreover, the set-based formulation is complex and 
difficult to extend to features considered here.
%(e.g.~data structures, first-class contracts).

Our prior work~\cite{dvanhorn:TobinHochstadt2012Higherorder}, as
discussed in the introduction, also performs soft contract
verification, but with far less sophistication and success. 
  As our empirical results show, the contributions of this paper are required to
tackle the arithmetic relations, flow-sensitive reasoning, and
complex recursion found in our benchmarks.

An alternative approach has been applied to checking contracts in
Haskell and OCaml~\cite{dvanhorn:Xu2012Hybrid,dvanhorn:Xu2009Static},
which is to inline monitors into a program following a transformation
by \citet{dvanhorn:Findler2002Contracts} and then simplify the
program, either using the compiler, or a specialized symbolic engine
equipped with an SMT solver.  The approach would be applicable to
untyped languages except for the final step dubbed
\emph{logicization}, a type-based transformation of program
expressions into first-order logic (FOL).  A related approach used for
Haskell is to use a denotational semantics that can be mapped into
FOL, which is then model checked~\cite{dvanhorn:Vytiniotis2013HALO},
but this approach is highly dependent on the type structure of a
program.  Further, these approaches assume a different semantics for
contract checking that monitors recursive calls.  This allows the use
of contracts as inductive hypotheses in recursive calls.  In contrast,
our approach can naturally take advantage of this stricter semantics
of contract checking and type systems, but can also accommodate the
more common and flexible checking policy.
Additionally, our approach does not rely on type information, the lack of
which makes these approaches inapplicable to many of our benchmarks.

Contract verification in the setting of typed, first-order contracts
is much more mature.  A prominent example is the work on verifying C\#
contracts as part of the Code Contracts
project~\cite{dvanhorn:Fahndrich2011Static}.

\paragraph{Refinement type checking:}
Refinement types are an
alternative approach to statically verifying pre- and post-conditions
in a higher-order functional language.  There are several approaches
to checking type refinements; one is to restrict the computational
power of refinements so that checking is decidable at type-checking
time~\cite{dvanhorn:Freeman1991Refinement}; another is allow
unrestricted refinements as in contracts, but to use a solver to
attempt to discharge
refinements~\cite{dvanhorn:Knowles2010Hybrid,dvanhorn:Rondon2008Liquid,dvanhorn:Vazou2013Abstract}.
In the latter approach, when a refinement cannot be discharged, some
systems opt to reject the
program~\cite{dvanhorn:Rondon2008Liquid,dvanhorn:Vazou2013Abstract},
while others such as hybrid type-checking residualize a run-time check
to enforce the refinement~\cite{dvanhorn:Knowles2010Hybrid}, similar
to the way soft-typing residualizes primitive pre-condition checks.
The end result of our approach most closely resembles that of hybrid
checking, although the technique applies regardless of the type
discipline and approaches the problem using different tools.

DJS~\cite{Chugh:POPL, Chugh:OOPSLA} supports expressive refinement specification and
verification for stateful JavaScript programs, including sophisticated
dependent specifications which \SCV cannot verify.
However, most dependent properties require heavy annotations.
Moreover, \texttt{null} inhabits every object type.
Thus the approach cannot give the same guarantees about programs
such as \texttt{reverse} (\S\ref{sec:approxctxs})
without significantly more annotation burden.
Additionally, it relies on whole program annotation, type-checking, and analysis.

\paragraph{Model checking higher-order recursion schemes:}
Much of the recent work on model checking of higher-order programs
relies on the decidability of model checking trees generated by
higher-order recursion schemes
(HORS)~\cite{dvanhorn:Ong2006ModelChecking}.  A HORS
is essentially a program in the simply-typed $\lambda$-calculus
with recursion and finitely inhabited base types that generates
(potentially infinite) trees.  Program verification is accomplished by
compiling a program to a HORS in which the generated tree represents
program event
sequences~\cite{dvanhorn:Kobayashi2009Types,dvanhorn:Kobayashi2010Higherorder}.
This method is sound and complete for the simply typed
$\lambda$-calculus with recursion and finite base types, but the gap
between this language and realistic languages is
significant.  Subsequently, an untyped variant of HORS has been
developed~\cite{dvanhorn:Tsukada2010Untyped}, which has applications
to languages with more advanced type systems, but
%% such as Hindley-Milner
%% polymorphism and rank-2 intersection types, although other type
%% systems such as System F have undecidable model checking problems and
despite the name it does not lead to a model checking procedure for
the untyped $\lambda$-calculus.  A subclass of untyped HORS is the
class of recursively typed recursion schemes, which has applications
to typed object-oriented
programs~\cite{dvanhorn:Kobayashi2013ModelChecking}.  In this setting,
model checking is undecidable, but relatively complete with a certain
recursive intersection type system (anything typable in this system
can be verified).  To cope with infinite data domains such as
integers, counter-example guided abstraction refinement (CEGAR)
techniques have been developed~\cite{dvanhorn:Kobayashi2011Predicate}.
The complexity of model checking even for the simply typed case is
$n$-EXPTIME hard (where $n$ is the rank of the recursion scheme), but
progress on decision procedures
\cite{dvanhorn:Kobayashi2009Type,dvanhorn:Kobayashi2009Modelchecking}
has lead to  verification engines that can verify a number of
``small but tricky higher-order functional programs in less than a
second.''

In comparison, the HORS approach can verify some specifications which
\SCV cannot, but in a simpler (typed) setting, whereas our
lightweight method applies to richer languages.  Our approach
handles untyped higher-order programs with sophisticated language
features and infinite data domains.  Higher-order program invariants
may be stated as behavioral contracts, while the HORS-based systems
only support assertions on first order data.  Our work is also able to
verify programs with unknown external functions, not just unknown
integer values, which is important for modular program verification,
and we are able to verify many of the small but tricky programs
considered in the HORS work.  
%% On the other hand, the method cannot be
%% applied to verify some of our examples containing higher-order
%% abstract components such as \texttt{insertion-sort}.
%
%% The HORS approach also relies on type
%% assumptions, and therefore are unable to check untyped programs or
%% programs in type systems other than their own.

% Inference (relevant?):
%\cite{dvanhorn:Kawaguchi2010Dsolve}
%\cite{dvanhorn:DBLP:conf/vmcai/ZhuJ13}

\section{Conclusions and perspective}
\label{sec:conclusion}

We have presented a lightweight method and prototype implememtation
for static contract checking using a non-standard reduction semantics
that is capable of verifying higher-order modular programs with
arbitrarily omitted components.
Our tool, \SCV, scales to realistic language features such as recursive
data structures and modular programs, and verifies programs written in
the idiomatic style of dynamic languages.
The analysis  proves the absence of run-time errors
without excessive reliance on programmer help.
With zero annotation, \SCV already helps programmers find 
unjustified usage of partial functions with high precision and could
even be modified to suggest inputs that break the program.
With explicit contracts, programmers can enforce rich specifications to
their programs and have those optimized away
without incurring the significant run-time overhead
entailed by dynamic enforcement.  
%
%% Future work will extend our techniques to effectful programs, as well as
%% extend the SMT integration to more sophisticated
%% data types.
%% We plan to improve our work by extending the approach to work on
%% effectful programs and contracts, and to by integrating sophisticated
%% solvers for reasoning about base values.

While in this paper, we have addressed the problem of soft contract
verification, the technical tools we have introduced apply beyond this
application.  For example, a run of \SCV can be seen as a modular
program analysis---it soundly predicts which functions are called at
any call site. Moreover it can be composed with whole-program analysis
techniques to derive modular
analyses~\cite{dvanhorn:VanHorn2010Abstracting}. A small modification
to \textsf{blur} to cause it to pick a small set of concrete values
would turn our system into a concolic execution
engine~\cite{dvanhorn:Larson2003High}.  Adding temporal
contracts~\cite{dvanhorn:Disney2011Temporal} to our system would
produce a model checker for higher-order languages. This breadth of
application follows directly from the semantics-based nature of our
approach.

\acks We thank Carl Friedrich Bolz, Jeffrey S.~Foster, Michael Hicks, J.~Ian
Johnson, Lindsey Kuper, Aseem Rastogi, and Matthew Wilson for comments.
We thank the anonymous reviewers of ICFP 2014 for their detailed
reviews, which helped to improve the presentation and technical
content of the paper.
This material is based on research sponsored by the NSF under award 1218390, the
NSA under the Science of Security program, and DARPA under the
programs Automated Program Analysis for Cybersecurity
(FA8750-12-2-0106).  The U.S.  Government is authorized to reproduce
and distribute reprints for Governmental purposes notwithstanding any
copyright notation thereon.

\balance
\bibliographystyle{abbrvnat}
{%\footnotesize
\bibliography{dvh-bibliography,sth-bibliography,local}
}

\iftechreport
  \input{appendix}
\else
  \relax
\fi

\end{document}

%% file: preamble.tex
%% Requires:
%% \usepackage{color}
%% \usepackage{calc}
%% \usepackage{xspace}
%% \usepackage{pifont}
%% \usepackage{upgreek}

% \usepackage{MnSymbol}

\setlength{\bibsep}{1pt}

\newtheorem{theorem}{Theorem}
\newtheorem{lemma}{Lemma}
\newtheorem{definition}{Definition}

\definecolor{gray}{rgb}{0.9,0.9,0.9}
\definecolor{red}{rgb}{1,0,0}
\newcommand{\graybox}[1]{\mbox{\setlength{\fboxsep}{0.5pt}%
    \colorbox{gray}{$#1$}}}

\newcommand{\fixme}[1][\relax]{{\color{red}{FIXME: #1}}}

\newcommand{\ma}[1]{\ensuremath{#1}\xspace}

\newcommand\fakeparagraph[1]{\vspace{-1em}\subsubsection*{#1}}
\newcommand\eval{\mathit{eval}}
\newcommand\aval{\smash[l]{\widehat{\mathit{aval}}}}

\newcommand\refines{\sqsubseteq}
\newcommand\abstracts{\sqsupseteq}

\newcommand\asblm[3]{\ablm{#1}{#2}{#3}{\slambda}{#3}}
%\mmodvar\mmodvar\mval\mcon\mval
% (\syntax{wrong}^{#1}\:#3)}}

\newcommand\toplevel{\ensuremath{\dagger}}
\newcommand\allocname{\ensuremath{\mathsf{alloc}}}

\renewcommand{\vec}[1]{\overrightarrow{#1}}

\newcommand\pow{\ensuremath{\mathcal{P}}}

\newcommand\pluseq{+\!\!=}

\newcommand\ie{\emph{i.e.}}
\newcommand\rstep{{\ensuremath{\mathbf{r}}}}
\newcommand\betav{{\ensuremath{\mathbf{v}}}}
\newcommand\avstep{{\ensuremath{\widehat{\mathbf{v}}}}}
\newcommand\uvstep{{\ensuremath{\avstep\vstep}}}
\newcommand\uvmstep{{\ensuremath{\avstep\vstep{\Delta_M}}}}
\newcommand\betaext{{\ensuremath{\beta_{ext}}}}
\newcommand\abetav{{\ensuremath{\widehat{\mathbf{v}}}}}

\newcommand\opaque{\bullet}
\newcommand\topaque{\bullet^\mtyp}

\newcommand{\dotcup}{\mathrel{\ensuremath{\mathaccent\cdot\cup}}}

\newcommand\with[2]{{#1}/{#2}}
\newcommand\remcon[2]{\with{#1}{#2}}%{\textsc{rem}(#1,#2)}
\newcommand\demonic[2]{\textsc{demonic}(#1,#2)}
\newcommand{\closed}[1]{\dot{#1}}

\newcommand\sclos[2]{\langle{#1},{#2}\rangle}

% M is for Meta
\newcommand\mans{\ensuremath{a}}
\newcommand\mtyp{\ensuremath{T}}
\newcommand\mbasetyp{\ensuremath{B}}
\newcommand\menv{\ensuremath{\rho}}
\newcommand\menvo{\ensuremath{\varrho}}
\newcommand\mprg{\ma{p}}
\newcommand\mprgo{\ensuremath{q}}
\newcommand\mexp{\ensuremath{e}}
\newcommand\maexp{\ensuremath{f}} % Currently indistinguished from \mexp.
\newcommand\mexpo{\ensuremath{e'}}
\newcommand\mexpoo{\ensuremath{e''}}
\newcommand\mvexp{{\ensuremath{\vec{\mexp}}}}
\newcommand\mvexpo{{\ensuremath{\vec{\mexpo}}}}
\newcommand\mvmod{{\ensuremath{\vec{\mmod}}}}
\newcommand\mvmodo{{\ensuremath{\vec{\mmodo}}}}
\newcommand\mvmodvar{{\ensuremath{\vec{\mmodvar}}}}
\newcommand\mmod{\ensuremath{m}}
\newcommand\mmodo{\ensuremath{m'}}
\newcommand\mmodvar{\ensuremath{f}}
\newcommand\mmodvaro{\ensuremath{g}}
\newcommand\mmodvaroo{\ensuremath{h}}
\newcommand\mcon{\ensuremath{c}}
\newcommand\mcono{\ensuremath{d}}
\newcommand\mbase{\ensuremath{b}}
\newcommand\mpreval{\ensuremath{u}}
\newcommand\mprevalo{\ensuremath{u'}}
\newcommand\mval{\ensuremath{v}\xspace}
\newcommand\mvalo{\ensuremath{v'}} %% FIXME
\newcommand\mvaloo{\ensuremath{v''}}
\newcommand\mvalset{\ensuremath{\vec\mval}}
\newcommand\maval{\ensuremath{v}} % Indistinguished from \mval
\newcommand\mvar{\ensuremath{x}}
\newcommand\mvaro{\ensuremath{y}}
\newcommand\mvaroo{\ensuremath{z}}
\newcommand\mnum{\ensuremath{n}}
\newcommand\mnumo{\ensuremath{n'}}
\newcommand\mlam{\ensuremath{l}}
\newcommand\mfun{\ensuremath{w}}
\newcommand\mkont{\ensuremath{\kappa}}
\newcommand\mclos{\ensuremath{D}}
\newcommand\msbl{\ensuremath{S}}
\newcommand\maddr{\ensuremath{a}}
\newcommand\maddro{\ensuremath{b}}
\newcommand\maddroo{\ensuremath{a'}}
\newcommand\maddrtc{\maddro}
\newcommand\mden{\ensuremath{d}}
\newcommand\msto{\ensuremath{\sigma}}
\newcommand\mstoo{\ensuremath{\sigma'}}
\newcommand\masto{\ensuremath{\hat\sigma}}
\newcommand\mblm{\ensuremath{B}}
\newcommand\mstate{\ensuremath{\varsigma}}
\newcommand\mastate{\ensuremath{\hat\varsigma}}
\newcommand\mop{o}
\newcommand\moppred{o_?}
\newcommand\mopone{o_1}
\newcommand\moptwo{o_2}
\newcommand\mlab{\ell}
\newcommand\mlabo{\ell'}
\newcommand\mlaboo{\ell''}
\newcommand\mctx{\ensuremath{\mathcal{E}}}
\newcommand\mctxo{\ensuremath{\mathcal{E'}}}
\newcommand\mvval{\ensuremath{\vec\mval}}
% meta-variables for closed terms
\newcommand\mPrg{\ma{P}}
\newcommand\mPrgo{\ensuremath{Q}}
\newcommand\mAns{\ensuremath{A}}
\newcommand\mAnso{\ensuremath{A'}}
\newcommand\mCon{\ensuremath{C}}
\newcommand\mCono{\ensuremath{D}}
\newcommand\mExp{\ensuremath{E}}
\newcommand\mExpo{\ensuremath{E'}}
\newcommand\mExpoo{\ensuremath{E''}}
\newcommand\mVal{\ensuremath{V}}
\newcommand\mConset{\ensuremath{\mathcal{C}}}
\newcommand\mPreval{\ensuremath{U}}
\newcommand\mLab{\ensuremath{a}} % value alias on heap
\newcommand\mLabo{\ensuremath{a'}}
\newcommand\mktab{\ensuremath{\Xi}}
\newcommand\mktabo{\ensuremath{\Xi'}}
\newcommand\mmtab{\ensuremath{M}}
\newcommand\mmtabo{\ensuremath{M'}}
\newcommand\mF{\ensuremath{F}}

\newcommand\proves[3]{{{#1}\vdash{#2}:{#3}\mbox{\,\yay}}}%\ding{51}}}}
\newcommand\refutes[3]{{#1}\vdash{#2}:{#3}\mbox{\,\nay}}%\ding{55}}}
\newcommand\ambig[3]{{#1}\vdash{#2}:{#3}\,\mbox{\bf ?}} % neither proves, nor refutes
\newcommand\extProves[3]{{{#1}{\;\vdash_{S}\;}{#2}:{#3}\mbox{\,\yay}}}%\ding{51}}}}
\newcommand\extRefutes[3]{{{#1}{\;\vdash_{S}\;}{#2}:{#3}\mbox{\,\nay}}}%\ding{55}}}}
\newcommand\extAmbig[3]{{{#1}{\;\vdash_{S}\;}{#2}:{#3}\mbox{\bf ?}}}
\newcommand\tran[1]{\{\!\!\{#1\}\!\!\}_S}

\newcommand\absval{\with\opaque\mconset}

\newcommand\metapp[3]{\metappname({#1},{#2},{#3})}
\newcommand\metappname{\textsc{app}}

\newcommand\MSet[1]{\ensuremath{\mathsf{#1}}}
\newcommand\Mmodvar{\MSet{MVar}}
\newcommand\Mvar{\MSet{Var}}
\newcommand\Mval{\MSet{Val}}
\newcommand\Mcon{\MSet{Con}}
\newcommand\Mprg{\MSet{Prog}}
\newcommand\Mlam{\MSet{Lam}}
\newcommand\Mfun{\MSet{Wrap}}
\newcommand\Mblm{\MSet{Blame}}
\newcommand\Mexp{\MSet{Expr}}

\newcommand\confont[1]{\ensuremath{\mathsf{#1}}}

\newcommand\conarrow{\ensuremath\mbox{\tt →}}
\newcommand\slambda{{\tt λ}}
\newcommand\lparen{\mbox{\texttt{(}}}
\newcommand\rparen{\mbox{\texttt{)}}}

% S is for Surface
%\newcommand\sbegin[2]{(\mathsf{begin}\:#1\:#2)}
\newcommand\sbegin[2]{#1;\ #2}
\newcommand\sreclamnp[3]{\ensuremath{\slambda_{#2} #1.#3}}
\newcommand\slamnp[2]{\sreclamnp{#1}{\relax}{#2}}
\newcommand\smod[3]{\lparen\ensuremath{\syntax{module}\:#1\:#2\:#3}\rparen}
\newcommand\stlam[3]{\slamnp{#1\!:\!#2}{#3}}
\newcommand\slam[2]{\slamnp{#1}{#2}}
\newcommand\slamp[2]{\lparen\slam{#1}{#2}\rparen}
\newcommand\slamc[3]{\ensuremath{#2\,\conarrow\, \slam{#1}#3}}
\newcommand\sreclam[3]{(\sreclamnp{#1}{#2}{#3})}
\newcommand\strec[3]{\srec{{#1}\!:\!{#2}}{#3}}
\newcommand\srec[2]{\mu{#1}.{#2}}
\newcommand\sapp[2]{\ensuremath{#1\:#2}}
\newcommand\spapp[2]{\ensuremath{\lparen#1\:#2\rparen}}
\newcommand\sapptwo[3]{\ensuremath{(#1\:#2\:#3)}}
\newcommand\sif[3]{\mathsf{if}\:#1\:#2\:#3}
\newcommand\szerop[1]{\mathsf{zero?}(#1)}
\newcommand\sint{\ensuremath{\syntax{int}}}
\newcommand\sany{\ensuremath{\syntax{any}}}
\newcommand\sarr[2]{\ensuremath{#1\,\conarrow\,  #2}}
\newcommand\sdep[3]{\ensuremath{#1\,\conarrow\,\slambda#2.#3}}
\newcommand\stdep[4]{\sdep{#1}{#2\!:\!#3}{#4}}
\newcommand\spred[1]{\ensuremath{\liftpred{#1}}}
\newcommand\scons[2]{(\syntax{cons}\:#1\:#2)}
\newcommand\sconsop{\syntax{cons}}
\newcommand{\stabcons}[2]{(\syntax{cons}\:\begin{tabular}[t]{@{}l}{$#1$}\\{${#2})$}\end{tabular}}
\newcommand\szero{\syntax{zero?}}
\newcommand\ssucc{\syntax{add1}}
\newcommand\snatp{\syntax{nat?}}
\newcommand\snump{\syntax{num?}}
\newcommand\sintp{\syntax{int?}}
\newcommand\sequalp{\syntax{=}}
\newcommand\splus{\syntax{+}}
\newcommand\soptwo[3]{{#2}({#1},{#3})}
\newcommand\sopone[2]{{#1}({#2})}
\newcommand\sadd[2]{{#1}+{#2}}
\newcommand\smin[2]{{#1}-{#2}}
\newcommand\sfalse{\syntax{false}}
\newcommand\strue{\syntax{true}}
\newcommand\sboolp{\syntax{bool?}}
\newcommand\sempty{\syntax{empty}}
\newcommand\semptyp{\syntax{empty?}}
\newcommand\scar{\syntax{car}}
\newcommand\scdr{\syntax{cdr}}
\newcommand\sprocp{\syntax{proc?}}
\newcommand\sfalsep{\ensuremath{\mbox{\tt{false?}}}\xspace}
\newcommand\struep{\syntax{true?}}
\newcommand\sconsp{\syntax{cons?}}
\newcommand\sdepp{\syntax{dep?}}
\newcommand\sflatp{\syntax{flat?}}
\newcommand\srt[4]{\lparen\syntax{rt}_{\striple{#1}{#2}{#3}}\ #4\rparen}
\newcommand\sblur[4]{\lparen\syntax{blur}_{\striple{#1}{#2}{#3}}\ #4\rparen}
\newcommand\sapprox[2]{#1\ \oplus\ #2}
\newcommand\sref[1]{!#1}

\newcommand\sext[3]{#1[#2 \mapsto #3]}
\newcommand\sdom[1]{\ensuremath{\mathit{dom}}(#1)}
\newcommand\refine[3]{\ensuremath{\mathit{refine}}(#1, #2, #3)}

\newcommand\syntax[1]{\ensuremath{\mbox{\tt{#1}}}} % \ma{\mathsf{#1}}}

\newcommand\sand[2]{\sif{#1}{#2}{\sfalse}}
\newcommand\sor[2]{\sif{#1}{\strue}{#2}}

% Surface contracts
\newcommand\sanyc{\ensuremath{\syntax{any/c}}}
\newcommand\snatc{\ensuremath{\syntax{nat/c}}}
\newcommand\sconsc[2]{\langle#1,\!#2\rangle}
\newcommand\sandc[2]{#1\wedge#2}
\newcommand\sorc[2]{#1\vee#2}
\newcommand\srecc[2]{\mu #1.#2}
\newcommand\sboolc{\syntax{bool/c}}
\newcommand\semptyc{\syntax{empty/c}}
\newcommand\spair[2]{\langle#1,\!#2\rangle}
\newcommand\striple[3]{\langle#1,#2,#3\rangle}
\newcommand\squadruple[4]{\langle#1,#2,#3,#4\rangle}

\newcommand\sopc[1]{\spred{\slam\mvar{(\sapp{#1}\mvar)}}}
\newcommand\snotopc[1]{\spred{\slam\mvar{\sapp\neg{(\sapp{#1}\mvar)}}}}

% T is for Type
\newcommand\tvar{\ensuremath{A}}
\newcommand\tarr[2]{\ensuremath{#1}\rightarrow{#2}}

% A is for Annotated
\newcommand\amod[2]{\lparen\ensuremath{\syntax{module}\:#1\:#2}\rparen}
\newcommand\achk[6]{\chk{#1}{#3}{#4}{#5}{#2}}
% order consistent with redex model
% Dropping potential HO value for simplicity
%\newcommand\chk[5]{\ensuremath{#1\!\Leftarrow^{#2,#3}_{#4}\! #5}}
\newcommand\chk[5]{\ensuremath{\syntax{mon}^{#2,#3}_{#4}\mbox{\tt
      (}}#1, #5\mbox{\tt )}}
\newcommand\achksimple[2]{\ensuremath{\syntax{mon}\mbox{\tt
      (}#1,{#2}\mbox{\tt )}}}
\newcommand\afchksimple[2]{\ensuremath{\syntax{fmon}(#1,{#2})}}
%\newcommand\abless[2]{\ensuremath{(#1\!\dashrightarrow\! #2)}}
% #1 broke the contract on #2, contract was "look it up" but value was #3,
% first-order check was #4 but value was #5
% SIMPLIFIED BY DROPPING POTENTIAL HO VALUE AND WITNESS VALUE
%\newcommand\ablm[5]{\ensuremath{(\syntax{blame}^{#1}_{#2}\:#3\:#4\:#5)}}
\newcommand\ablm[5]{\ensuremath{\syntax{blame}^{#1}_{#2}}}
\newcommand\simpleblm[2]{\ablm{#1}{#2}\relax\relax\relax}

\newcommand\fc[3]{\ensuremath{\syntax{fc}_{#2}(#1, #3)}}
\newcommand\assume[2]{\ensuremath{\syntax{assume}\lparen{#1},
    #2\rparen}}
\newcommand\sArr[5]{\ensuremath{(#1 \syntax{\Rightarrow}^{#2,#3}_{#4} #5)}}
\newcommand\simpleArr[2]{\ensuremath{\syntax{Arr}(#1, #2)}}

\newcommand\sblm[2]{\ensuremath{\syntax{blame}^{#1}_{#2}}}

\newcommand\stypbool{\syntax{B}}
\newcommand\stypnum{\syntax{N}}
\newcommand\starr[2]{{#1}\rightarrow{#2}}
\newcommand\stcon[1]{\syntax{con}({#1})}
\newcommand\sflat[1]{\ensuremath{\syntax{flat}(#1)}}

\newcommand\svar[1]{\syntax{#1}}

\newcommand\depbless[7]{
  \ma{\slam{#2}{\achk{#3}{(\sapp{#7}{\achk{#1}{#2}{#5}{#4}{#6}{}})}{#4}{#5}{#6}{}}}}
\newcommand\bless[6]{
  \ensuremath{\chk{(#1\!\dashrightarrow\! #2)}{#3}{#4}{#5}{#6}}}

\newcommand\simpledepbless[4]{
  \ma{\slam{#2}{\achksimple{#3}{(\sapp{#4}{\achksimple{#1}{#2}})}}}}
\newcommand\simplebless[3]{
  \ma{\slam\mvar{\achksimple{#2}{(\sapp{#3}{\achksimple{#1}{\mvar}})}}}}

\newcommand\deltamap[5]{\delta({#1},{#2},{#3})\ni({#4},{#5})}
\newcommand\absdeltamap[3]{\absdelta({#1},{#2})\ni{#3}}

\newcommand\alloc[1]{\ensuremath{\mathsf{alloc}(#1)}}

\newcommand\subst[3]{\ensuremath{[#1/#2]#3}}
\newcommand\redrule[1]{\ensuremath{\mbox{\sc{#1}}}}
\newcommand\step{\rightarrow}
\newcommand\vstep{\mathbf{v}}       % CBV notion of reduction
\newcommand\cstep{\mathbf{c}}       % Contract notion of reduction
\newcommand\sstep{\mathbf{s}}       % Env and state
\newcommand\acstep{{\widehat{\mathbf{c}}}} % Contract-as-value notion of reduction
\newcommand\stdstep{\longmapsto}
\newcommand\multistdstep{\longmapsto\!\!\!\!\!\rightarrow}
\newcommand\mtenv{\ensuremath{\emptyset}}
\newcommand\initstore{[\maddr_0 \mapsto \kmt]}
\newcommand\unload{\ensuremath{\mathcal{U}}}

\newcommand\close[2][\rho]{\ensuremath{\langle{#2},{#1}\rangle}}
\newcommand\mstep{\ensuremath{\stdstep\quad}}
\newenvironment
    {machine}[2][\relax]
    {\begin{display*}[#1]{#2}{\textwidth}\[
\begin{tabular}{@{}>{$}p{2.1in}<{$}>{$}p{.2in}<{$}>{$}p{2.2in}<{$}>{$}p{1.8in}<{$}}
}
    {\end{tabular}\]\end{display*}}   

\newcommand\plaindelta{\cn\delta}
\newcommand\absdelta{\delta}

\newcommand\liftpred[1]{\sflat{#1}\relax}

\newcommand\projleft[1]{\pi_1{#1}}
\newcommand\projright[1]{\pi_2{#1}}
\newcommand\proj\pi

\newcommand\mcache{\mathcal{CV}}

\newcommand\vcons[2]{({#1},{#2})}

% A machine state
\newcommand\s[4]{\langle{#1},{#2},{#4},{#3}\rangle}
% Only expression is interesting
\newcommand\se[1]{\s{#1}\menv\msto\mcont}
% Value in the hole, only continuation is interesting
\newcommand\sk[1]{\s\mval\menv\msto{#1}}

\newcommand\mcont\kappa
\newcommand\mconto\iota

\newcommand\cont[1]{\textsf{#1}}

\newcommand\ab[1]{\widehat{#1}}
\newcommand\cn[1]{\widetilde{#1}}

\newcommand\kchk[6]{\cont{chk}^{#3,#4}_{#5}({#1},{#2},{#6})}
\newcommand\kchkor[8]{\cont{chk-or}^{#5,#6}_{#7}({#1},{#2},{#3},{#4},{#8})}
\newcommand\kchkcons[8]{\cont{chk-cons}^{#3,#4}_{#5}({#1},{#2},{#6},{#7},{#8})}
\newcommand\kchkconso[8]{\cont{chk-cons}^{#5,#6}_{#7}({#1},{#2},{#3},{#4},{#8})}
\newcommand\kfn[4]{\cont{fn}^{#3}({#1},{#2},{#4})}
\newcommand\kap[4]{\cont{ar}^{#3}({#1},{#2},{#4})}
\newcommand\kif[4]{\cont{if}({#1},{#2},{#3},{#4})}
\newcommand\kopone[3]{\cont{op}^{#2}({#1},{#3})}
\newcommand\koptwo[5]{\cont{opr}^{#4}({#1},{#2},{#3},{#5})}
\newcommand\koptwol[5]{\cont{opl}^{#4}({#1},{#2},{#3},{#5})}
\newcommand\kdem[3]{\cont{begin}({#1},{#2},{#3})}
\newcommand\kmt{\cont{mt}}

\newcommand\mkaddr{k}
\newcommand\mkaddro{i}

\newcommand\lang{$\lambda_{\text{C}}$\xspace}
\newcommand\etal{\emph{et al.}}

\newcommand\leftcurly{\ensuremath{\{}}
\newcommand\rightcurly{\ensuremath{\}}}

\newcommand\mloc{\ensuremath{L}}
\newcommand\amb{\mathit{amb}}

\def\TirName{\textit}

\newif\iftechreport

\newcommand\techrep{the accompanying technical report\xspace}
\newcommand\techrepcite{\techrep~\cite{techreport}}

\newcommand\SCV{\textsf{SCV}\xspace}

%% file: appendix.tex
\clearpage

\appendix
\section{Full Formalism}
This section presents full materials omitted from the paper.
Figure \ref{fig:semantics-full} shows the complete operational semantics of Symbolic \lang.
Figure \ref{fig:proof-system-full} shows the full implementation of the basic proof system
(without calling out to the SMT solver).
Any triple of $\striple\msto\mval\mcon$ not applicable to these rules
are implicitly related by $\ambig\msto\mval\mcon$.
Figure \ref{fig:delta-full} shows the implementation of basic operations.
Figure \ref{fig:approximation-full} defines approximations between expression-store pairs
as well as program states.
Figure \ref{fig:oplus-full} displays definitions for operator $\oplus$
which approximates a value with knowledge of a ``prior'' value.
Figure \ref{fig:sum-approx-full} defines approximation between a state $\spair\mexp\msto$
and a system $\striple\mktab\mmtab{S}$.
We abbreviate simple predicates of the form $\slam\mvar{(\sapp\moppred\mvar)}$
as $\moppred$.

%% As an example, consider the {\tt fibonacci} program:
%% \begin{alltt}
%% (define (fib n)
%%   (if (< n 2) n (+ (fib (- n 1)) (fib (- n 2)))))
%% (fib •/int?)
%% \end{alltt}
%% When executed, table $\mktab$ eventually sees that
%% applying {\tt fib} to an opaque integer results in the following two
%% waiting contexts:
%% \begin{alltt}
%%     (+ [] (fib (- •/int? 2)))
%%     (+ •/int? [])
%% \end{alltt}
%% where {\tt []} is the hole to plug in.

%% Now consider one expression reachable from application {\tt (fib 7)}
%% by standard reduction:
%% \begin{alltt}
%% (+ (if (< 6 2) 6 (+ (fib 5) (fib (- 6 2)))) (fib (- 7 2)))
%% \end{alltt}
%% First, we can establish the approximation of the subexpression
%% {\tt (if (< 6 2) \ldots)}, because the summarizing semantics
%% always applies a function at least once.
%% Next, By wrapping context {\tt (+ [] (fib (- •/int? 2)))}
%% around this expression, we can approximate the above
%% expansion of {\tt (fib 7)}.
%% By wrapping the waiting contexts an arbitrary number of times,
%% we can approximate any expression reachable from {\tt (fib •/int?)}.

\begin{figure*}[t]

\begin{mathpar}

\inferrule{\smod\mvar{\mpreval_c}\mpreval \in \mvmod}{\mvar^\mvar,\msto \stdstep \mpreval,\msto}

\inferrule{\smod\mvar{\mpreval_c}\mpreval \in \mvmod
\\ \mvar\neq\mlab}
{\mvar^\mlab,\msto
\stdstep
 \chk{\mpreval_c}{\mvar}{\mlab}{\mvar}{\mpreval},\msto}

\inferrule{ }{\sapp{(\slam\mvar\mexp)}{\mpreval},\msto
\stdstep
\subst\mpreval\mvar\mexp,\msto}

\inferrule{\mpreval\neq\opaque}
          {\mpreval,\msto
            \stdstep            
            \with\mpreval\emptyset,\msto}

\inferrule{\deltamap\msto\sprocp\mval\sfalse{\msto_t}}
{\sapp{\mval}\mval'^\mlab,\msto
\stdstep
\simpleblm\mlab\Lambda,{\msto_t}}

\inferrule{\mLab\notin\msto}
          {\opaque,\msto
            \stdstep
            \mLab,\sext\msto\mLab{\with\opaque\emptyset}}

\inferrule{\deltamap\msto\mop{\vec\mval}\mans{\msto_a}}
          {\sapp\mop{\vec\mval}, \msto
            \stdstep
            \mans,\msto_a }

\inferrule{(\mstoo,\mvalo) = \refine\msto\mval{\mval_c}}
          {\assume\mval{\mval_c},\msto \stdstep \mvalo,\mstoo}

\inferrule{\deltamap\msto\sfalsep\mval\sfalse{\msto_t}}
          {\sif\mval{\mexp_1}{\mexp_2},\msto \stdstep \mexp_1,\msto_t }

\inferrule{\deltamap\msto\sfalsep\mval\strue{\msto_f}}
          {\sif\mval{\mexp_1}{\mexp_2},\msto \stdstep \mexp_2,\msto_f}

\inferrule{\deltamap\msto\sdepp{\mval_c}\sfalse\mstoo\\
  \proves\mstoo\mval{\mval_c}}
          {\chk{\mval_c}\mlab\mlabo\mlaboo\mval, \msto
            \stdstep
            \mval, \mstoo}

\inferrule{\deltamap\msto\sdepp{\mval_c}\sfalse\mstoo \\
  \refutes\mstoo\mval{\mval_c}}
          {\chk{\mval_c}\mlab\mlabo\mlaboo\mval, \msto
            \stdstep
            \simpleblm\mlab\mlaboo, \mstoo}

\inferrule{\deltamap\msto\sdepp{\mval_c}\sfalse\mstoo\\
  \ambig\mstoo\mval{\mval_c}}
          {\chk{\mval_c}\mlab\mlabo\mlaboo\mval, \msto
            \stdstep
            \sif{{(\sapp{\mval_c}\mval)}^{\mlaboo}}{\!\!\assume\mval{\mval_c}}{\simpleblm\mlab\mlaboo}, \mstoo}

\inferrule{\deltamap\msto\sprocp\mval\strue{\msto_1}}
          {\chk{\slamc\mvar{\mval_c}\mcono}\mlab\mlabo\mlaboo\mval, \msto
            \stdstep
            \slam\mvar{\chk\mcono\mlab\mlabo\mlaboo{\sapp\mval{\chk{\mval_c}\mlabo\mlab\mlaboo\mvar}}}, {\msto_1}}

\inferrule{\deltamap\msto\sdepp\mLab\strue{\msto_1}\\
  \deltamap{\msto_1}\sprocp\mval\strue{\msto_2}\\
           {\msto_2}(\mLab) = \slamc\mvar{\mval_c}\mcono }
  {\chk\mLab\mlab\mlabo\mlaboo\mval, \msto
    \stdstep
    \slam\mvar{\chk\mcono\mlab\mlabo\mlaboo{\sapp\mval{\chk{\mval_c}\mlabo\mlab\mlaboo\mvar}}}, {\msto_2}}

\inferrule{\deltamap\msto\sdepp{\mval_c}\strue{\msto_1}\\
  \deltamap{\msto_1}\sprocp\mval\sfalse{\msto_2}}
          {\chk{\mval_c}\mlab\mlabo\mlaboo\mval, \msto
            \stdstep
            \simpleblm\mlab\mlaboo, {\msto_2}}

\inferrule{\deltamap\msto\sprocp\mLab\strue\mstoo}
          {\sapp\mLab\mval, \msto
            \stdstep
            \mLab_a,\sext\mstoo{\mLab_a}\opaque}
          
\inferrule[\mbox{\scshape Havoc}]{\deltamap\msto\sprocp\mLab\strue\mstoo}
          {\sapp\mLab\mval, \msto \stdstep
            \sapp{\syntax{havoc}}\mval, \mstoo}
\end{mathpar}
\caption{Operational semantics for Symbolic \lang}
\label{fig:semantics-full}
\end{figure*}

\begin{figure*}[]

\begin{mathpar}
\inferrule{ }{\proves\msto\mnum\sintp}

\inferrule{ }{\proves\msto\sfalse\sfalsep}

\inferrule{ }{\proves\msto\sempty\semptyp}

\inferrule{ }{\proves\msto{\sconsc{\mval_1}{\mval_2}}\sconsp}

\inferrule{ }{\proves\msto{\slam\mvar\mexp}\sprocp}

\inferrule{ }{\proves\msto{\sdep\mval\mvar\mcon}\sdepp}

\inferrule{\proves\msto\mval{\moppred'}\\
           \moppred' \neq \moppred}{\refutes\msto\mval\moppred}

\inferrule{\mval_c \in \vec\mval}{\proves\msto{\with\mpreval{\vec\mval}}{\mval_c}}

\inferrule{\proves\msto\mval{\slam\mvar\mexp}}{\refutes\msto\mval{\slam\mvar{(\sapp\neg\mexp)}}}

\inferrule{\refutes\msto\mval{\slam\mvar\mexp}}{\proves\msto\mval{\slam\mvar{(\sapp\neg\mexp)}}}

\inferrule{{\slam\mvar{(\sapp\neg\mexp)} \in \vec{\mval_c}}}
          {\refutes\msto{\with\mpreval{\mval_c}}{\slam\mvar\mexp}}

\inferrule{\proves\msto{\msto(\mLab)}{\mval_c}}{\proves\msto\mLab{\mval_c}}

\inferrule{\refutes\msto{\msto(\mLab)}{\mval_c}}{\refutes\msto\mLab{\mval_c}}

%\inferrule{\ambig\msto{\msto(\mLab)}{\mval_c}}{\ambig\msto\mLab{\mval_c}}
\end{mathpar}

\caption{Basic Proof System}
\label{fig:proof-system-full}
\end{figure*}

\begin{figure*}[]

\begin{mathpar}
% predicates
\inferrule{\proves\msto\mval\moppred}{\deltamap\msto\moppred\mval\msto\strue}

\inferrule{\refutes\msto\mval\moppred}{\deltamap\msto\moppred\mval\msto\sfalse}

\inferrule{\ambig\msto\mval\moppred\\
           (\mstoo,\_) = \refine\msto\mval\moppred}
          {\deltamap\msto\moppred\mval\mstoo\strue}
          
\inferrule{\ambig\msto\mval\moppred\\
           (\mstoo,\_) = \refine\msto\mval{\neg\moppred}}
          {\deltamap\msto\moppred\mval\mstoo\sfalse}

% plus
\inferrule{ }{\deltamap\msto\splus{{\mnum_1},{\mnum_2}}{\mnum_1 \splus \mnum_2}\msto}

\inferrule{\deltamap\msto\sintp{\mval_1}{\msto_1}\sfalse}
          {\deltamap\msto\splus{{\mval_1},{\mval_2}}{\simpleblm{}\Lambda}{\msto_1}}
          
\inferrule{\deltamap\msto\sintp{\mval_1}{\msto_1}\strue\\
           \deltamap{\msto_1}\sintp{\mval_2}{\msto_2}\sfalse}
          {\deltamap\msto\splus{{\mval_1},{\mval_2}}{\simpleblm{}\Lambda}{\msto_2}}

\inferrule{\mval_1 \neq \mnum\ or\ \mval_2 \neq \mnum\\
           \deltamap\msto\sintp{\mval_1}{\msto_1}\strue\\
           \deltamap{\msto_1}\sintp{\mval_2}{\msto_2}\strue\\
           \mLab \notin \msto_2}
          {\deltamap\msto\splus{{\mval_1},{\mval_2}}\mLab{\sext{\msto_2}\mLab{\with\opaque\sintp}}}
          
% equal
\inferrule{ }{\deltamap\msto\sequalp{\mnum,\mnum}\strue\msto}

\inferrule{\mnum_1 \neq \mnum_2}{\deltamap\msto\sequalp{{\mnum_1},{\mnum_2}}\sfalse\msto}

\inferrule{\deltamap\msto\sintp{\mval_1}{\msto_1}\sfalse}
          {\deltamap\msto\sequalp{{\mval_1},{\mval_2}}{\simpleblm{}\Lambda}{\msto_1}}
          
\inferrule{\deltamap\msto\sintp{\mval_1}{\msto_1}\strue\\
           \deltamap{\msto_1}\sintp{\mval_2}{\msto_2}\sfalse}
          {\deltamap\msto\sequalp{{\mval_1},{\mval_2}}{\simpleblm{}\Lambda}{\msto_2}}

\inferrule{\deltamap\msto\sintp{\mval_1}{\msto_1}\strue\\
           \deltamap{\msto_1}\sintp{\mval_2}{\msto_2}\strue\\
           \mval_1\neq\mval_2}
          {\deltamap\msto\sequalp{{\mval_1},{\mval_2}}\sfalse{\msto_2}}
          
\inferrule{\deltamap\msto\sintp{\mval_1}{\msto_1}\strue\\
           \deltamap{\msto_1}\sintp{\mval_2}{\msto_2}\strue}
          {\deltamap\msto\sequalp{{\mval_1},{\mval_2}}\strue{\msto_2}}
          
% cons
\inferrule{ }{\deltamap\msto\sconsop{{\mval_1},{\mval_2}}{\sconsc{\mval_1}{\mval_2}}\msto}
          
% car
\inferrule{ }{\deltamap\msto\scar{\sconsc{\mval_1}{\mval_2}}{\mval_1}\msto}

\inferrule{\deltamap\msto\sconsp\mLab\strue\mstoo\\
           \mstoo(\mLab) = \sconsc{\mval_1}{\mval_2}}
          {\deltamap\msto\scar\mLab{\mval_1}\msto}
          
\inferrule{\deltamap\msto\sconsp\mLab\sfalse\mstoo}
          {\deltamap\msto\scar\mval{\simpleblm{}\Lambda}\msto}
          
% cdr
\inferrule{ }{\deltamap\msto\scdr{\sconsc{\mval_1}{\mval_2}}{\mval_2}\msto}

\inferrule{\deltamap\msto\sconsp\mLab\strue\mstoo\\
           \mstoo(\mLab) = \sconsc{\mval_1}{\mval_2}}
          {\deltamap\msto\scdr\mLab{\mval_2}\msto}
          
\inferrule{\deltamap\msto\sconsp\mLab\sfalse\mstoo}
          {\deltamap\msto\scdr\mval{\simpleblm{}\Lambda}\msto}
          
\end{mathpar}

\caption{Basic Operations}
\label{fig:delta-full}
\end{figure*}

\begin{figure*}[]
\begin{mathpar}
\inferrule{ }{\mnum \refines^F \with\opaque\sintp}

\inferrule{(\mval_1,\msto) \refines^F (\mvalo_1,\msto')\\
           (\mval_2,\msto) \refines^F (\mvalo_2,\msto')}
          {(\sconsc{\mval_1}{\mval_2},\msto) \refines^F (\sconsc{\mvalo_1}{\mvalo_2},\msto')}

\inferrule{(\with\mpreval{\vec\mval},\msto_1) \refines^F (\msto_2(\mLab),\msto_2)}
          {(\with\mpreval{\vec\mval},\msto_1) \refines^F (\mLab,\msto_2)}
          { F(\mLab) = \with\mpreval{\vec\mval}}
          
\inferrule{(\msto_1(\mLab_1),\msto_1) \refines^F (\msto_2(\mLab_2),\msto_2)}
          {(\mLab_1,\msto_1) \refines^F (\mLab_2,\msto_2)}
          { F(\mLab_2) = \mLab_1}
          
\inferrule{(\with{\mpreval_1}{\vec{\mval_1}},\msto_1)
            \refines^F
           (\with{\mpreval_2}{\vec{\mval_2}}),\msto_2)\\ 
           (\mval_c,\msto_1) \refines^F (\mval_d,\msto_2)}
          {(\with{\mpreval_1}{\vec{\mval_1} \cup \{\mval_c\}}, \msto_1)
           \refines^F
           (\with{\mpreval_2}{\vec{\mval_2} \cup \{\mval_d\}}, \msto_2)}
           
\inferrule{ }
          {(\mpreval,\msto_1) \refines^F (\opaque,\msto_2)}

\inferrule{(\with{\mpreval_1}{\vec{\mval_1}},\msto_1)
            \refines^F
           (\with{\mpreval_2}{\vec{\mval_2}},\msto_2)}
          {(\with{\mpreval_1}{\vec{\mval_1} \cup \mval}, \msto_1)
           \refines^F
           (\with{\mpreval_2}{\vec{\mval_2}}, \msto_2)}
            
\inferrule{\smod\mmodvar{\mpreval_c}\opaque \in \vec\mmod\\ or \\
           \mmodvar \in \{\dagger,\syntax{havoc}\}}
          {(\sblm\mmodvar\mmodvaro,\msto_1) \refines^F_{\vec\mmod} (\mexp,\msto_2)}
          
% structural case:
\inferrule{ }{(\mvar,\msto) \refines^F (\mvar,\mstoo)}

\inferrule{(\mexp,\msto) \refines^F (\mexpo,\mstoo)\\
           (\mexp_x,\msto) \refines^F (\mexpo_x,\mstoo)}
          {(\sapp{\mexp}{\mexp_x},\msto) \refines^F (\sapp\mexpo\mexpo_x,\mstoo)}
          
\inferrule{(\mexp_i,\msto) \refines^F (\mexpo_i,\mstoo)
           \mbox{ for each $\mexp_i$ in $\vec\mexp$, $\mexpo_i$ in $\vec\mexpo$}}
          {(\sapp\mop{\vec\mexp},\msto) \refines^F (\sapp\mop{\vec\mexpo},\mstoo)}
          
\inferrule{(\mexp_1,\msto) \refines^F (\mexp'_1,\msto')\\
 		  (\mexp_2,\msto) \refines^F (\mexp'_2,\msto')\\
		  (\mexp_3,\msto) \refines^F (\mexp'_3,\msto')}
          {(\sif{\mexp_1}{\mexp_2}{\mexp_3},\msto) \refines^F (\sif{\mexp'_1}{\mexp'_2}{\mexp'_3},\msto')}
          
\inferrule{(\mcon,\msto) \refines^F (\mcon',\msto')\\
           (\mcono,\msto) \refines^F (\mcono',\msto')}
	      {(\sdep\mcon\mvar\mcono,\msto) \refines^F (\sdep\mcon'\mvar\mcono',\msto')}
	      
\inferrule{(\mcon,\msto) \refines^F (\mcon',\msto')\\
           (\mexp,\msto) \refines^F (\mexp',\msto')}
          {(\achk\mcon\mlab\mlabo\mlaboo\mexp,\msto) \refines^F (\achk\mcon'\mlab\mlabo\mlaboo\mexp',\msto')}
          
\inferrule{ }{(\asblm\mlab\mlabo,\msto) \refines^F (\asblm\mlab\mlaboo,\mstoo)}

\end{mathpar}
\caption{Approximation}
\label{fig:approximation-full}
\end{figure*}

\begin{figure}
\[
\begin{array}{@{}r@{\;}c@{\;}l@{\ }l@{}}

\sapprox{\sconsc{\mval_0}{\mval_1}}{\sconsc{\mval_2}{\mval_3}}
&=&
\multicolumn{2}{l}{
\sconsc{\sapprox{\mval_0}{\mval_2}}{\ \sapprox{\mval_1}{\mval_3}}}\\[1mm]

\sapprox{\srec\mvar{\vec{\mval_0}}}{\srec\mvaro{\vec{\mval_1}}}
&=&
\multicolumn{2}{l}{
\srec\mvar{(\sapprox{\vec{\mval_0}}{\vec{\subst{\sref\mvar}{\sref\mvaro}{\mval_1}})}}}\\[1mm]

\sapprox{\srec\mvar{\vec{\mval_0}}}\mval
&=&
\multicolumn{2}{l}{
\srec\mvar{(\sapprox{\vec{\mval_0}}{\subst{\sref\mvar}{\srec\mvar{\vec{\mval_0}}}\mval)}}}\\[1mm]

\sapprox{\mval_0}{\mval_1}
&=&
\multicolumn{1}{l}{
\srec\mvar{\{\mval_0, \subst{\sref\mvar}{\mval_0}{\mval_1}\}}}
&{\mbox{if }\mval_0 \in_s \mval_1}\\[1mm]

\sapprox{\mnum_0}{\mnum_1}
&=&
\multicolumn{1}{l}{\with\opaque\sintp}
&{\mbox{if }\mnum_0 \neq \mnum_1}\\[1mm]

\sapprox{\mval_0}{\mval_1}
&=&
\multicolumn{1}{l}{\mval_1}
&{\mbox{otherwise}}
\end{array}
\]
\caption{Approximation}
\label{fig:oplus-full}
\end{figure}

\begin{figure}
\begin{mathpar}
\inferrule{\mstate \in (\mktab,\mmtab,S)}
          {\mstate \refines (\mktab,\mmtab,S)}

\inferrule{\mstate \refines \mstate'\\
           \mstate' \refines (\mktab,\mmtab,S)}
          {\mstate \refines (\mktab,\mmtab,S)}
          
\inferrule{\srt{\_}\mval{\_} \notin \mctx_0\\
		   (\mval_x,\msto) \refines (\mval_z,\mstoo)\\
		   (\mval_y,\msto) \refines (\mval_z,\mstoo)\\
		   \mktab(\mstoo,\mval,\mval_z) = (F,\mstoo,\mctxo_0,\mctxo_k)\\
		   (\mctx_0,\msto) \refines (\mctxo_0,\mstoo)\\
		   (\mctx_k,\msto) \refines (\mctxo_k,\mstoo)\\
		   (\mctx_0[\srt{\msto_1}\mval{\mval_y}],\msto) \refines (\mktab,\mmtab,S)}
          {(\mctx_0[\srt{\msto_0}\mval{\mval_x}\mctx_k[\srt{\msto_1}\mval{\mval_y}\mexp]],\msto)
           \refines (\mktab,\mmtab,S)}
          
\end{mathpar}
\caption{Approximation of Summarizing Semantics}
\label{fig:sum-approx-full}
\end{figure}

\section{Proofs}
This section presents proofs for theorems in the paper.
Lemmas \ref{lem:main} and \ref{lem:summ} prove theorems \ref{thm:main}
and \ref{thm:summ}, respectively.
Other lemmas support these main ones.

%% DVH: Omit ^F since we can understand naked \refines as \exists
%% F, \refines^F.

\begin{lemma}[Soundness of $\delta$]
\label{lem:delta}\ \\
If $(\vec{\mval_1},\msto_1) \refines (\vec{\mval_2},\msto_2)$
then $\delta(\msto_1,\mop,\vec{\mval_1}) \refines
\delta(\msto_2,\mop,\vec{\mval_2})$.
\end{lemma}
\begin{proof}
By inspection of $\delta$ and cases of $\mop$
and $(\vec{\mval_1},\msto_1) \refines (\vec{\mval_2},\msto_2)$.
\end{proof}

\begin{lemma}[Soundness of function application]
\label{lem:app}\ \\
If $(\mval_f,\msto_1) \refines (\mval_g,\msto_2)$,
$(\mval_x,\msto_1) \refines (\mval_y,\msto_2)$, and
$\spapp{\mval_f}{\mval_x},\msto_1 \stdstep \mstate$,
then 
$\spapp{\mval_g}{\mval_y},\msto_2 \stdstep \mstate'$
such that
$\mstate_1 \refines \mstate_2$.
\end{lemma}
\begin{proof}

By case analysis on $\mexp_1$
, where $(\mexp_1, \msto'_1) = \mstate_1$ and $(\mexp_1, \msto'_2) = \mstate_2$
\begin{itemize}
\item{Case 1} $\mexp_1 = \simpleblm\relax\relax$ because
              $\deltamap{\msto_1}\sprocp{\mval_{f1}}\sfalse{\mstoo_1}$

    By lemma \ref{lem:delta}, $\deltamap{\msto_2}\sprocp{\mval_{f2}}\sfalse{\mstoo_2}$,
    which means $\spapp{\mval_{f2}}{\mval_{x2}},\msto_2 \stdstep (\simpleblm\relax\relax,\mstoo_2)$.

\item{Case 2} $\mexp_1 \neq \simpleblm\relax\relax$ because
  $\deltamap{\msto_1}\sprocp{\mval_{f1}}\strue{\mstoo_1}$.
  
  Case analysis on the derivation of
  $(\mval_{f1},\msto_1) \refines
  (\mval_{f2},\msto_2)$
  \begin{itemize}
  \item{Case 2a}
    Both $\mval_{f1}$ and $\mval_{f2}$ are lambdas and step to
    structurally similar states.
  % previous 2b case is obsolete
  \item{Case 2b}
    $(\mval_{f1},\msto_1) \refines (\mLab, [\mLab \mapsto \opaque])$.
    Then $(\sapp{\mval_{f2}}{\mval_{x2}})$ steps to $(\mLab', [\mLab' \mapsto \opaque])$.
  \end{itemize}
\end{itemize}
\end{proof}

\begin{lemma}[Soundness of flat contract checking]
\label{lem:fc}\ \\
If  $(\mval_c,\msto_1) \refines (\mval_d,\msto_2)$
and $(\mval_1,\msto_1) \refines (\mval_2,\msto_2)$, then
\begin{itemize}
\item $\proves{\msto_1}{\mval_1}{\mval_c}$ implies
      $(\sapp{\mval_d}{\mval_2}, \msto_2) \multistdstep (\mval_t,\msto'_2)$
      such that $\deltamap{\msto'_2}{\sfalsep}{\mval_t}\sfalse{\msto''_2}$
      and $\msto_1 \refines \msto''_2$, and
\item $\refutes{\msto_1}{\mval_1}{\mval_c}$ implies
      $(\sapp{\mval_d}{\mval_2}, \msto_2) \multistdstep (\mval_t,\msto'_2)$
      such that $\deltamap{\msto'_2}{\sfalsep}{\mval_t}\strue{\msto''_2}$
      and $\msto_1 \refines \msto''_2$.
\end{itemize}
\end{lemma}
\begin{proof}
By case analysis on the derivation of $(\mval_c,\msto_1) \refines (\mval_d,\msto_2)$
and inspection of provability relation.
\end{proof}

\begin{lemma}[Heap substitution]
\label{lem:preservation}\ \\
If $(\mexp_1,\msto_1) \refines (\mexp_2,\msto_2)$
and $\mstoo_1 \refines \mstoo_2$
then $(\mexp_1,\mstoo_1) \refines (\mexp_2,\mstoo_2)$.
\end{lemma}
\begin{proof}
By induction on $\mexp_2$.
\end{proof}

\begin{lemma}[Soundness of reduction relation]
\label{lem:main}\ \\
If $\mprg \refines \mprgo$ and
$\mprg \stdstep \mprg'$ where $\mprg' \neq \sblm\relax\relax$,
then $\mprgo \multistdstep \mprgo'$
such that $\mprg' \refines \mprgo'$
for some $q'$.
\end{lemma}
\begin{proof}
By case analysis on the derivation of $\mprg \stdstep \mprg'$.
\begin{itemize}

\item{Case 1}

    $\mprg = {\vec{\mmod_1}}\:{\mctx_1}[\mexp_1], \msto_1 \stdstep {\mctx_1}[\mexp'_1], \mstoo_1$
    
    $\mprgo = {\vec{\mmod_2}}\:{\mctx_2}[\mexp_2], \msto_2 \stdstep {\mctx_2}[\mexp'_2], \mstoo_2$
    
    where $(\mctx_1,\msto_1)\:{\refines_{\vec\mmod_2}}\:(\mctx_2,\msto_2)$
    and $(\mexp_1,\msto_1)\:{\refines_{\vec\mmod_2}}\:(\mexp_2,\msto_2)$.
    \begin{itemize}
    \item{Case 1a:} $\mexp_1 = \mmodvar^\mmodvaro$
    
            If $\mmodvar = \mmodvaro$, they both step to the simple value in module $\mmodvar$,
            so $\mexp'_1 \refines \mexp'_2$ because ${\vec{\mmod_1}} \refines {\vec{\mmod_2}}$
            
            If $\mmodvar \neq \mmodvaro$, $\mexp'_1$, $\mexp'_2$ is
            $\achksimple{\mpreval_c}{\mpreval_1}$ and $\mExp'_2$ is $\achksimple{\mpreval_c}{\mpreval_2}$
            where $\mpreval_1 \refines \mpreval_2$.
            
    \item{Case 1b:} $\mexp_1$ is a function application (lemma \ref{lem:app}).
    \item{Case 1c:} $\mexp_1$ is a primitive application (lemma \ref{lem:delta}).
    \item{Case 1d:} $\mexp_1 = \sif{\mval_1}{\mexp_{1a}}{\mexp_{1b}}$
                   and $\mexp_2 = \sif{\mval_2}{\mexp_{2a}}{\mexp_{2b}}$
                   
            If both $\mval_1$ and $\mval_2$ are non-labels, there are no change
            to the heaps, and the case is structural.
            Otherwise, $\mval_2$ is a label, the case follows by lemma \ref{lem:preservation}.
            
    \item{Case 1e:} $\mexp_1 = \achksimple{\mval_c}{\mval_1}$
                   and $\mexp_2 = \achksimple{\mval_d}{\mval_2}$
                   and $\mcon$ is flat.
            
            We do case analysis on (${\msto_1}\vdash{\mval_1}:{\mval_c}$).
            
            If $\ambig{\msto_1}{\mval_1}{\mval_c}$,
            the next states for both are structurally similar.
            
            If $\proves{\msto_1}{\mval_1}{\mval_c}$ or
            $\refutes{\msto_1}{\mval_1}{\mval_c}$,
            we rely on lemma \ref{lem:fc}.
            
            Otherwise,  $\mexp_1$ and $\mexp_2$ step to structurally similar states.
            
    \item{Case 1*:} All other cases involve stepping to structurally similar states
            or looking up environments, which are straightforward.
    \end{itemize}
\item{Case 2}

    $\mPrg = {\vec{\mmod_1}}\:{\mctx_{1a}}[{\mctx_{1b}}[\mexp_1]], \msto_1 \stdstep
                              {\mctx_{1a}}[{\mctx_{1b}}[\mexp'_1]], \mstoo_1$
    and $\mPrgo = {\vec{\mmod_2}}\:{\mctx_{2a}}[{\mctx_{2b}}[\mexp_2]], \msto_2$
    where $\mctx_{1a}$ is the largest context such that $\mctx_{1a}\:{\refines_{\vec{\mmod_2}}}\:\mctx_{2a}$
    but $\mctx_{1b}\:{\not\refines_{\vec{\mmod_2}}}\:\mctx_{2b}$.
    This means $\mctx_{1b}[\mexp_1]\:\refines\:\mctx_{2b}[\mexp_2]$ by one of the non-structural
    rules for $\refines$, which are insensitive to $\mexp_1$'s content.
    Hence $\mctx_{1b}[\mexp'_1]\:\refines\:\mctx_{2b}[\mexp_2]$

\end{itemize}
\end{proof}

\begin{lemma}[Soundness of havoc]
\label{lem:havoc}\ \\
If there exists a context $\mctx$ such that
\[
\mbox{\emph{$\mvmod\;\smod\mmodvar{\mpreval_c}{\mpreval}\vdash\mctx[\mmodvar],\msto \multistdstep
\simpleblm\mmodvar\mmodvaro$}}
\]
then
\[
\mbox{\emph{$\mvmod\;\smod\mmodvar{\mpreval_c}{\mpreval}\vdash(\sapp{\syntax{havoc}}{\mmodvar}),\msto
\multistdstep
\simpleblm\mmodvar\mmodvaro$.}}
\]
\end{lemma}
\begin{proof}
Any context triggering a blame of $\mmodvar$ can be reduced to a canonical form of:
$$
\mctx ::= [\ ]\ |\ (\mctx\ \mpreval)\ |\ (\scar\ \mctx)\ |\ (\scdr\
\mctx)\text,
$$
which can be approximated by contexts of the form:
$$
\mctx' ::= [\ ]\ |\ (\mctx'\ \opaque)\ |\ (\scar\ \mctx')\ |\ (\scdr\
\mctx')\text.
$$
By induction on $\mctx'$\!, for any $\mval$,
if $\mctx'[\mval] \multistdstep \mans$,
then 
$(\syntax{havoc}\ \mval) \multistdstep (\syntax{havoc}\ \mans)$.
\end{proof}

\begin{lemma}[Return frame]\ \\
If $(\emptyset,\emptyset,\{(\mexp,\emptyset)\}) \multistdstep (\mktab,\mmtab,S)$
and $\mctx[\srt\msto{\mval_f}{\mval_x}\mexpo] \in S$,
then $\mctx[\spapp{\mval_f}{\mval_x}] \in S$.
\begin{proof} By induction on 
$(\emptyset,\emptyset,\{(\mexp,\emptyset)\}) \multistdstep (\mktab,\mmtab,S)$
and case analysis on the last reduction to $\mctx[\srt\msto{\mval_f}{\mval_x}\mexpo]$,
assuming programmers cannot write expressions of the form $\srt\msto\mval\mval\mexp$.
\end{proof}
\end{lemma}

\begin{lemma}
\label{lem:rt_base}\ \\
If $\spair{\mctx[\srt{\msto_0}{\mval_f}{\mval_x}\mval]}{\msto}
   \refines \striple\mktab\mmtab{S}$
where $\srt{\_}{\mval_f}{\_}{[\;]} \notin \mctx$,
then there is $\mstate \in S$
such that $\spair{\mctx[\srt{\msto_0}{\mval_f}{\mval_x}\mval]}{\msto}
           \refines \mstate$
\begin{proof}
By case analysis on the last step deriving
\[
\spair{\mctx[\srt{\msto_0}{\mval_f}{\mval_x}\mval]}{\msto}
   \refines \striple\mktab\mmtab{S}
\]
\end{proof}
\end{lemma}

\begin{lemma}[Soundness of summarization]
\label{lem:summ}\ \\
If $\mstate_1 \refines (\mktab_1,\mmtab_1,S_1)$ and
$\mstate_1 \stdstep \mstate_2$,
then $(\mktab_1,\mmtab_1,S_1) \multistdstep (\mktab_2,\mmtab_2,S_2)$
such that $\mstate_2 \refines (\mktab_2,\mmtab_2,S_2)$.
\end{lemma}
\begin{proof}
By induction on the derivation of $\mstate_1 \refines (\mktab_1,\mmtab_1,S_1)$
and case analysis on the reduction $\mstate_1 \stdstep \mstate_2$.

Let $(\mexp_1,\msto_1) = \mstate_1$, $(\mexp_2,\msto_2) = \mstate_2$,
$\mexp_1$ = $\mctx[\mexpo_1]$ and $\mexp_2$ = $\mctx[\mexpo_2]$
where $\mexpo_1$ is a redex.
\begin{itemize}
\item{Case 1:} % Base case for \refines
  $\mstate_1 \refines (\mktab_1,\mmtab_1,S_1)$ because
  $\mstate_1 \in S_1$
  
  Case analysis on $\mstate_1 \stdstep \mstate_2$:
  \begin{itemize}
             
    \item {Sub-case: $\mexpo_1 = \spapp{\slam\mvar\mexp}{\mval}$}

    We have $\mctx[\sapp{(\slam\mvar\mexp)}\mval] \stdstep
              \mctx[\srt{\msto_1}{\slam\mvar\mexp}\mval{\subst\mval\mvar\mexp}]$
              
     \begin{itemize}
     \item{Sub-sub-case 1:}
     $\striple{\mktab_1}{\mmtab_1}{\mstate_1}$ proceeds with the same reduction
     , which straightforwardly approximates $\mstate_2$ 
     
     \item{Sub-sub-case 2:}
      $\striple{\mktab_1}{\mmtab_1}{\mstate_1}$ proceeds with widened argument
      , which also straightforwardly approximates $\mstate_2$
      
      \item{Sub-sub-case 3:}
      $(\mktab_1,\mmtab_1,S_1) \stdstep (\mktabo_1,\mmtab_1,S'_1)$
      
      where
      $\mctx = \mctx_0[\srt{\msto_0}{\slam\mvar\mexp}{\mval_0}{\mctx_k}]$
      
      and $(\mval_0,\msto_0) \refines^F (\mvalo,\msto_1)$
      
      and $\mktabo_1 = \mktab_1 \sqcup [(\mstoo_1,\slam\mvar\mexp,\mvalo)
           \mapsto
          (F,\msto_0,\mctx_0,\mctx_k)]$
      
      and $S'_1 \supseteq S_1$
      
      But because $\mctx_0[\srt{\msto_0}{\slam\mvar\mexp}{\mval_0}\mctx_k] \in S'_1$,
      it follows that
      ${\mctx_0[\sapp{(\slam\mvar\mexp)}{\mval_0}]}$ is also in $S'_1$.
      
      Hence $(\mktabo_1,\mmtab_1,S'_1) \multistdstep (\mktabo_1,\mmtab_1,S''_1)$
      such that
      $\mctx_0[\srt{\msto_0}{\slam\mvar\mexp}{\mval_0}{\subst{\mval_0}\mvar\mexp}] \in S''_1$.
      
      Thus
      $\mctx_0[\srt{\msto_0}{\slam\mvar\mexp}{\mval_0}{\mctx_k[\subst{\mval_0}\mvar\mexp]}]
      \refines (\mktabo_1,\mmtab_1,S''_1)$,
      or $\mctxo[\subst{\mval_0}\mvar\mexp] \refines (\mktabo_1,\mmtab_1,S''_1)$
      
     \end{itemize}
     
   \item{Other sub-cases are straightforward}
    
  \end{itemize}
  
\item{Case 2:}
  $\mstate_1 \refines (\mktab_1,\mmtab_1,S_1)$ because:
  $\mstate_1 \refines \mstate'_1$ and
  $\mstate'_1 \refines (\mktab_1,\mmtab_1,S_1)$
  
  By lemma \ref{lem:main} (straightforwardly extended with $\syntax{rt}$ expressions),
  $\mstate_1 \stdstep \mstate_2$ implies
  $\mstate'_1 \multistdstep \mstate'_2)$ such that
  $\mstate_2 \refines \mstate'_2$
  
  By induction hypothesis,
  $(\mktab_1,\mmtab_1,S_1) \multistdstep (\mktab_2,\mmtab_2,S_2)$ such that
  $\mstate'_2 \refines (\mktab_2,\mmtab_2,S_2)$
  
  Therefore $\mstate_2 \refines (\mktab_2,\mmtab_2,S_2)$
  
\item{Case 3:} % inductive case for \refines
  $\mstate_1 \refines (\mktab_1,\mmtab_1,S_1)$ because:
  \begin{itemize}
    \item $\mctx[\mexpo_1]$ = $\mctx_0[\srt{\msto_x}\mval{\mval_x}{\mctx_k[\srt{\msto_y}\mval{\mval_y}\mexpo]}]$
    \item $\mctx_0[\srt{\msto_y}\mval{\mval_y}\mexpo] \refines (\mktab_1,\mmtab_1,S_1)$
    \item $(F,\mstoo_k,\mctxo_0,\mctxo_k) \in \mktab(\mstoo,\mval,\mval_z)$
    \item $(\mval_x,\msto_x) \refines (\mstoo,\mval_z) \abstracts (\mval_y,\msto_y)$
    \item $(\mctx_k,\msto_1) \refines (\mctxo_k,\mstoo_k)$
    \item $(\mctx_0,\msto_1) \refines (\mctxo_0,\mstoo_k)$
  \end{itemize}
  There are 2 subcases: either $\mexpo$ is a value or not.
  
  \begin{itemize}
  \item{Sub-case 1:} $\mexpo = \mval_a$
  
  This means $\mexp_1 = \mctx_0[\srt{\msto_x}\mval{\mval_x}{{\mctx_k}[\srt{\msto_y}\mval{\mval_y}{\mval_a}]}]$
  and $\mexp_2 = \mctx_0[\srt{\msto_x}\mval{\mval_x}{{\mctx_k}[\mval_a]}]$
  
  By lemma \ref{lem:rt_base},
  there exists $\mctxo'_0[\srt{\mstoo_y}\mval{\mvalo_y}{\mvalo_a}] \in S_1$
  such that $\mctxo_0[\srt{\msto_y}\mval{\mval_y}{\mval_a}]
             \refines \mctxo'_0[\srt{\mstoo_y}\mval{\mvalo_y}{\mvalo_a}]$.
  Then $(\mktab_1,\mmtab_1,S_1) \multistdstep (\mktabo_1,\mmtabo_1,S'_1)$
  such that $S'_1 \ni (\mctxo_0[\srt{\mstoo}\mval{\mval_z}{{\mctxo_k}[{\mval_a}]}])$.
    
  \item{Sub-case 2:} $\mexpo = \mctx_l[\mexpo_1]$ (where $\mexpo_1$ is the redex)
  
  We have
  \[
  \mctx_0[\srt{\msto_0}{\mval_f}{\mval_x}{\mctx_l[\mexpo_1]}]
  \stdstep \mctx_0[\srt{\msto_0}{\mval_f}{\mval_x}{\mctx_l[\mexpo_2]}]\text.
  \]
  
  By induction hypothesis,
  $(\mktab_1,\mmtab_1,S_1) \multistdstep (\mktab_2,\mmtab_2,S_2)$,
  such that
  $\mctx_0[\srt{\msto_0}{\mval_f}{\mval_x}{\mctx_l[\mexpo_2]}] \refines (\mktab_2,\mmtab_2,S_2)$.
  
  Because $\mktab_2 \supseteq \mktab_1$, 
  \[
  \mctx_0[\srt{\msto_0}{\mval_f}{\mval_x}{\mctxo_k\mctx_l[\mexpo_2]}]
  \refines (\mktab_2,\mmtab_2,S_2)\]
   follows.
  \end{itemize}

\end{itemize}
\end{proof}

\section{Detailed evaluation results}
This section shows detailed evaluation results for benchmarks
collected from different verification papers.
All are done on a Core i7 @ 2.70GHz laptop running Ubuntu 13.10 64bit.
Analysis times are averaged over 10 runs.
For benchmarks that time out, we display a range estimating
number of false positives.

{\footnotesize
\begin{center}
\begin{tabular}
{l@{\;\;}| @{\;}l @{\;}|@{\;} r@{\;} |@{\;} r@{\;} | @{\;}c@{\;} |@{\;} c @{\;}|@{\;} c @{\;}|@{\;} c@{\;} }
&   & & & \multicolumn{2}{@{\;}c@{\;}|@{\;}}{Simple} & \multicolumn{2}{@{\;}c@{\;}}{\SCV} \\ 
&    Program & Lines & Checks & Time (ms) & False/+ & Time (ms) & False/+\\ \hline
\parbox[t]{2mm}{\multirow{12}{*}{\rotatebox[origin=c]{90}{Occurrence Typing Examples}}}
&    ex-01 & 6 & 4 & 2.9 & 1 & 0.3 & 0\\
&    ex-02 & 6 & 8 & 6.8 & 0 & 0.5 & 0\\
&    ex-03 & 10 & 12 & 27.3 & 0 & 1.8 & 0\\
&    ex-04 & 11 & 12 & 20.6 & 1 & 0.7 & 0\\
&    ex-05 & 6 & 6 & 5.4 & 2 & 0.4 & 0\\
&    ex-06 & 8 & 11 & 10.0 & 0 & 0.6 & 0\\
&    ex-07 & 8 & 7 & 5.3 & 2 & 0.5 & 0\\
&    ex-08 & 6 & 11 & 12.9 & 1 & 0.7 & 0\\
&    ex-09 & 14 & 12 & 20.5 & 1 & 0.8 & 0\\
&    ex-10 & 6 & 8 & 3.7 & 1 & 0.3 & 0\\
&    ex-11 & 8 & 8 & 11.6 & 1 & 0.8 & 0\\
&    ex-12 & 5 & 11 & 6.1 & 2 & 0.4 & 0\\
&    ex-13 & 9 & 11 & 10.1 & 1 & 0 & 0\\
&    ex-14 & 12 & 20 & 12.6 & 2 & 1.1 & 0\\
    \hline
&    \textbf{Total} & 115 & 142 & 155.8 & 15 & 8.9 & 0
\end{tabular}
\end{center}

%\vspace*{-2mm}

\begin{center}
\begin{tabular}
{l@{\;\;}| @{\;}l @{\;}|@{\;} r@{\;} |@{\;} r@{\;} | @{\;}c@{\;} |@{\;} c @{\;}|@{\;} c @{\;}|@{\;} c@{\;} }
&   & & & \multicolumn{2}{@{\;}c@{\;}|@{\;}}{Simple} & \multicolumn{2}{@{\;}c@{\;}}{\SCV} \\ 
&    Program & Lines & Checks & Time (ms) & False/+ & Time (ms) & False/+\\ \hline
\parbox[t]{2mm}{\multirow{12}{*}{\rotatebox[origin=c]{90}{Soft Typing Examples }}}
&    append & 8 & 15 & 12.8 & 0 & 2.0 & 0\\
&    cpstak & 23 & 15 & 195.9 & 2 & 157.9 & 0\\
&    flatten & 12 & 24 & 16.3 & 1 & 177.7 & 0\\
&    last-pair & 7 & 9 & 4.5 & 2 & 0.6 & 0\\
&    last & 17 & 21 & 12.9 & 1 & 2.5 & 0\\
&    length-acc & 10 & 14 & 74.2 & 1 & 3.3 & 0\\
&    length & 8 & 13 & 46.4 & 1 & 1.6 & 0\\
&    member & 8 & 15 & 5.6 & 0 & 2.3 & 0\\
&    rec-div2 & 9 & 17 & 6.4 & 0 & 1.9 & 0\\
&    subst* & 11 & 12 & 4.8 & 1 & 13.3 & 0\\
&    tak & 12 & 14 & 9.8 & 0 & 0.7 & 0\\
&    taut & 9 & 8 & 34.9 & 0 & 16.5 & 0\\
    \hline % TODO sum
&    \textbf{Total} & 134 & 177 & 424.5 & 9 & 380.3 & 0
\end{tabular}
\end{center}

%\vspace*{-2mm}
\footnotesize
\begin{center}
\begin{tabular} 
{l@{\;\;}| @{\;}l @{\;}|@{\;} r@{\;} |@{\;} r@{\;} | @{\;}c@{\;} |@{\;} c @{\;}|@{\;} c @{\;}|@{\;} c@{\;} }
&   & & & \multicolumn{2}{@{\;}c@{\;}|@{\;}}{Simple} & \multicolumn{2}{@{\;}c@{\;}}{\SCV} \\ 
&    Program & Lines & Checks & Time (ms) & False/+ & Time (ms) & False/+\\ \hline
\parbox[t]{2mm}{\multirow{25}{*}{\rotatebox[origin=c]{90}{Higher-order Recursion Scheme Examples }}}
&   intro1 & 10 & 11 & 12.7 & 1 & 0.6 & 0\\
&    intro2 & 10 & 11 & 12.8 & 1 & 12.7 & 0\\
&    intro3 & 10 & 12 & 13.9 & 1 & 13.0 & 0\\
&    sum & 7 & 12 & 20.7 & 2 & 362.2 & 0\\
&    mult & 9 & 20 & 12.2 & 2 & 55.9 & 0\\
&    max & 13 & 11 & 47.7 & 2 & 40.2 & 0\\
&    mc91 & 8 & 15 & 15.1 & 2 & 96.6 & 1\\
&    ack & 9 & 16 & 22.1 & 2 & 6.7 & 0\\
&    repeat & 7 & 11 & 34.4 & 2 & 2.7 & 0\\
&    fhnhn & 15 & 15 & 31.5 & 1 & 3.2 & 0\\
&    fold-div & 18 & 34 & $\infty$ & - & 1047.2 & 0\\
&    hrec & 9 & 13 & 148.6 & 3 & 26.8 & 0\\
&    neg & 18 & 15 & 30.2 & 1 & 37.7 & 0\\
&    l-zipmap & 17 & 31 & $\infty$ & - & 1403.7 & 2\\
&    hors & 19 & 17 & 47.6 & 2 & 73.5 & 0\\
&    r-lock & 17 & 19 & 10.1 & 2 & 0.8 & 0\\
&    r-file & 44 & 62 & 76.6 & 0 & 5.1 & 0\\
&    reverse & 11 & 28 & 24.9 & 1 & 12.7 & 0\\
&    isnil & 8 & 17 & 13.4 & 2 & 4.1 & 0\\
&    mem & 11 & 28 & 28.0 & 2 & 3.3 & 0\\
&    nth0 & 15 & 27 & 10.2 & 0 & 2.3 & 0\\
&    zip & 16 & 42 & 29.6 & 1 & 42.8 & 1\\
\hline % TODO sum
&    \textbf{Total} & 301 & 467 & (642.3,$\infty$) & (29,94) & 3253.8 & 4
\end{tabular}
\end{center}

\begin{center}
\begin{tabular}
{l@{\;\;}| @{\;}l @{\;}|@{\;} r@{\;} |@{\;} r@{\;} | @{\;}c@{\;} |@{\;} c @{\;}|@{\;} c @{\;}|@{\;} c@{\;} }
&   & & & \multicolumn{2}{@{\;}c@{\;}|@{\;}}{Simple} & \multicolumn{2}{@{\;}c@{\;}}{\SCV} \\ 
&    Program & Lines & Checks & Time (ms) & False/+ & Time (ms) & False/+\\ \hline
\parbox[t]{2mm}{\multirow{7}{*}{\rotatebox[origin=c]{90}{Dep. Type
      Inf. }}}
&    boolflip & 10 & 17 & 1.5 & 0 & 1.7      & 0 \\
&    mult-all & 10 & 18 & 1.2  & 0 & 0.4      & 0 \\
&    mult-cps & 12 & 20 & $\infty$  & - & 146.9 & 0 \\
&    mult     & 10 & 17 & $\infty$ & - & 20.4     & 0 \\
&    sum-acm  & 10 & 15 & $\infty$ & - & 11.27    & 1 \\
&    sum-all  & 9  & 15 & 11.3 & 0 & 0.4      & 0 \\
&    sum      & 8  & 14 & $\infty$  & - & 11.9     & 0 \\
\hline
&    \textbf{Total} & 69 & 116 & (14,$\infty$) & (0,66) &
193.0 & 1
\end{tabular}
\end{center}

%\vspace*{-2mm}
\footnotesize
\begin{center}
\begin{tabular}
{l@{\;\;}| @{\;}l @{\;}|@{\;} r@{\;} |@{\;} r@{\;} | @{\;}c@{\;} |@{\;} c @{\;}|@{\;} c @{\;}|@{\;} c@{\;} }
&   & & & \multicolumn{2}{@{\;}c@{\;}|@{\;}}{Simple} & \multicolumn{2}{@{\;}c@{\;}}{\SCV} \\ 
&    Program & Lines & Checks & Time (ms) & False/+ & Time (ms) & False/+\\ \hline
\parbox[t]{2mm}{\multirow{25}{*}{\rotatebox[origin=c]{90}{Symbolic Execution Examples }}}
&    all & 9 & 16 & 16.7 & 0 & 3.2 & 0\\
&    even-odd & 10 & 11 & 9.7 & 0 & 27.1 & 0\\
&    factorial-acc & 10 & 9 & 13.8 & 0 & 2.8 & 0\\
&    factorial & 7 & 8 & 5.9 & 0 & 74.4 & 0\\
&    fibonacci & 7 & 11 & 7.3 & 0 & 1632.1 & 0\\
&    filter-sat-all & 11 & 18 & 38.0 & 1 & 11.4 & 1\\
&    filter & 11 & 17 & 25.2 & 0 & 6.2 & 0\\
&    foldl1 & 9 & 17 & 11.8 & 0 & 2.4 & 0\\
&    foldl & 8 & 10 & 10.4 & 0 & 2.1 & 0\\
&    foldr1 & 9 & 11 & 10.6 & 0 & 4.1 & 0\\
&    foldr & 8 & 10 & 15.5 & 0 & 3.3 & 0\\
&    ho-opaque & 10 & 14 & 11.7 & 0 & 1.2 & 0\\
&    id-dependent & 8 & 3 & 2.6 & 1 & 0.2 & 0\\
&    insertion-sort & 14 & 30 & 65.8 & 1 & 24.6 & 0\\
&    map-foldr & 11 & 20 & 19.5 & 1 & 3.8 & 0\\
&    mappend & 11 & 31 & 29.6 & 0 & 2.2 & 0\\
&    map & 10 & 13 & 25.0 & 1 & 3.3 & 0\\
&    max-var & 13 & 11 & 76.0 & 0 & 12.7 & 0\\
&    recip-contract & 7 & 9 & 20.2 & 0 & 0.5 & 0\\
&    recip & 8 & 15 & 20.0 & 2 & 0.6 & 0\\
&    rsa & 14 & 5 & 35.0 & 0 & 2.3 & 0\\
&    sat-7 & 20 & 12 & $\infty$ & - & 2548.1 & 0\\
&    sum-filter & 11 & 18 & 49.9 & 0 & 4.1 & 0\\
\hline % TODO sum
&    \textbf{Total} & 236 & 319 & (520,$\infty$) & (7,19) & 4372.7 & 1
% &    snake & 196 & 256 & 6324 & 0 & 8696 & 0\\
% &    tetris & 308 & 339 & $\infty$ & - & 11636 & 0\\
% &    zombie & 250 & 392 & $\infty$ & - & 10756 & 0 \\
% \hline
% &    \textbf{Total} & 754 & 987 & 6324+ & 0+ & 31088 & 0
\end{tabular}
\end{center}

%% file: main.bbl
\begin{thebibliography}{49}
\providecommand{\natexlab}[1]{#1}
\providecommand{\url}[1]{\texttt{#1}}
\expandafter\ifx\csname urlstyle\endcsname\relax
  \providecommand{\doi}[1]{doi: #1}\else
  \providecommand{\doi}{doi: \begingroup \urlstyle{rm}\Url}\fi

\bibitem[Aiken et~al.(1994)Aiken, Wimmers, and
  Lakshman]{dvanhorn:Aiken1994Soft}
A.~Aiken, E.~L. Wimmers, and T.~K. Lakshman.
\newblock Soft typing with conditional types.
\newblock POPL, 1994.

\bibitem[Austin et~al.(2011)Austin, Disney, and
  Flanagan]{dvanhorn:Austin2011Virtual}
T.~H. Austin, T.~Disney, and C.~Flanagan.
\newblock Virtual values for language extension.
\newblock OOPSLA, 2011.

\bibitem[Barrett et~al.(2011)Barrett, Conway, Deters, Hadarean, Jovanovi\'{c},
  King, Reynolds, and Tinelli]{dvanhorn:Barrett2011CVC4}
C.~Barrett, C.~Conway, M.~Deters, L.~Hadarean, D.~Jovanovi\'{c}, T.~King,
  A.~Reynolds, and C.~Tinelli.
\newblock {CVC4}.
\newblock CAV. 2011.

\bibitem[Cartwright and Fagan(1991)]{dvanhorn:Cartwright1991Soft}
R.~Cartwright and M.~Fagan.
\newblock Soft typing.
\newblock PLDI, 1991.

\bibitem[Cartwright and Felleisen(1996)]{dvanhorn:Cartwright1996Program}
R.~Cartwright and M.~Felleisen.
\newblock Program verification through soft typing.
\newblock \emph{ACM Comput. Surv.}, 1996.

\bibitem[Chugh et~al.(2012{\natexlab{a}})Chugh, Herman, and
  Jhala]{Chugh:OOPSLA}
R.~Chugh, D.~Herman, and R.~Jhala.
\newblock Dependent types for {JavaScript}.
\newblock In \emph{OOPSLA}, 2012{\natexlab{a}}.

\bibitem[Chugh et~al.(2012{\natexlab{b}})Chugh, Rondon, and Jhala]{Chugh:POPL}
R.~Chugh, P.~M. Rondon, and R.~Jhala.
\newblock Nested refinements: A logic for duck typing.
\newblock In \emph{POPL}, 2012{\natexlab{b}}.

\bibitem[{De Moura} and Bj{\o}rner(2008)]{dvanhorn:DeMoura2008Z3}
L.~{De Moura} and N.~Bj{\o}rner.
\newblock {Z3}: an efficient {SMT} solver.
\newblock TACAS, 2008.

\bibitem[Dimoulas et~al.(2011)Dimoulas, Findler, Flanagan, and
  Felleisen]{dvanhorn:Dimoulas2011Correct}
C.~Dimoulas, R.~B. Findler, C.~Flanagan, and M.~Felleisen.
\newblock Correct blame for contracts: no more scapegoating.
\newblock POPL, 2011.

\bibitem[Disney(2013)]{local:contracts-coffee}
T.~Disney.
\newblock contracts.coffee, July 2013.
\newblock URL \url{http://disnetdev.com/contracts.coffee/}.

\bibitem[Disney et~al.(2011)Disney, Flanagan, and
  McCarthy]{dvanhorn:Disney2011Temporal}
T.~Disney, C.~Flanagan, and J.~McCarthy.
\newblock Temporal higher-order contracts.
\newblock ICFP, 2011.

\bibitem[F{\"{a}}hndrich and Logozzo(2011)]{dvanhorn:Fahndrich2011Static}
M.~F{\"{a}}hndrich and F.~Logozzo.
\newblock Static contract checking with abstract interpretation.
\newblock FoVeOOS, 2011.

\bibitem[Findler and Felleisen(2002)]{dvanhorn:Findler2002Contracts}
R.~B. Findler and M.~Felleisen.
\newblock Contracts for higher-order functions.
\newblock ICFP, 2002.

\bibitem[Flanagan and Felleisen(1999)]{dvanhorn:Flanagan1999Componential}
C.~Flanagan and M.~Felleisen.
\newblock Componential set-based analysis.
\newblock \emph{ACM Trans. Program. Lang. Syst.}, 1999.

\bibitem[Flanagan et~al.(1996)Flanagan, Flatt, Krishnamurthi, Weirich, and
  Felleisen]{dvanhorn:Flanagan1996Catching}
C.~Flanagan, M.~Flatt, S.~Krishnamurthi, S.~Weirich, and M.~Felleisen.
\newblock Catching bugs in the web of program invariants.
\newblock PLDI, 1996.

\bibitem[Freeman and Pfenning(1991)]{dvanhorn:Freeman1991Refinement}
T.~Freeman and F.~Pfenning.
\newblock Refinement types for {ML}.
\newblock PLDI, 1991.

\bibitem[Henglein(1994)]{dvanhorn:Henglein1994Dynamic}
F.~Henglein.
\newblock Dynamic typing: syntax and proof theory.
\newblock \emph{Science of Computer Programming}, 1994.

\bibitem[Hickey et~al.(2013)Hickey, Fogus, and
  contributors]{local:clojure-contracts}
R.~Hickey, M.~Fogus, and contributors.
\newblock core.contracts, July 2013.
\newblock URL \url{https://github.com/clojure/core.contracts}.

\bibitem[Johnson and {Van Horn}(2014)]{Johnson:HOPA}
J.~I. Johnson and D.~{Van Horn}.
\newblock Abstracting abstract control.
\newblock \emph{CoRR}, 2014.
\newblock URL \url{http://arxiv.org/abs/1305.3163}.

\bibitem[Knowles and Flanagan(2010)]{dvanhorn:Knowles2010Hybrid}
K.~Knowles and C.~Flanagan.
\newblock Hybrid type checking.
\newblock \emph{ACM Trans. Program. Lang. Syst.}, 2010.

\bibitem[Kobayashi(2009{\natexlab{a}})]{dvanhorn:Kobayashi2009Modelchecking}
N.~Kobayashi.
\newblock Model-checking higher-order functions.
\newblock PPDP, 2009{\natexlab{a}}.

\bibitem[Kobayashi(2009{\natexlab{b}})]{dvanhorn:Kobayashi2009Types}
N.~Kobayashi.
\newblock Types and higher-order recursion schemes for verification of
  higher-order programs.
\newblock POPL, 2009{\natexlab{b}}.

\bibitem[Kobayashi and Igarashi(2013)]{dvanhorn:Kobayashi2013ModelChecking}
N.~Kobayashi and A.~Igarashi.
\newblock {Model-Checking} {Higher-Order} programs with recursive types.
\newblock ESOP, 2013.

\bibitem[Kobayashi and Ong(2009)]{dvanhorn:Kobayashi2009Type}
N.~Kobayashi and C.~H.~L. Ong.
\newblock A type system equivalent to the modal {Mu-Calculus} model checking of
  {Higher-Order} recursion schemes.
\newblock LICS, 2009.

\bibitem[Kobayashi et~al.(2010)Kobayashi, Tabuchi, and
  Unno]{dvanhorn:Kobayashi2010Higherorder}
N.~Kobayashi, N.~Tabuchi, and H.~Unno.
\newblock Higher-order multi-parameter tree transducers and recursion schemes
  for program verification.
\newblock POPL, 2010.

\bibitem[Kobayashi et~al.(2011)Kobayashi, Sato, and
  Unno]{dvanhorn:Kobayashi2011Predicate}
N.~Kobayashi, R.~Sato, and H.~Unno.
\newblock Predicate abstraction and {CEGAR} for higher-order model checking.
\newblock PLDI, 2011.

\bibitem[Larson and Austin(2003)]{dvanhorn:Larson2003High}
E.~Larson and T.~Austin.
\newblock High coverage detection of input-related security faults.
\newblock USENIX Security, 2003.

\bibitem[Meunier et~al.(2006)Meunier, Findler, and
  Felleisen]{dvanhorn:Meunier2006Modular}
P.~Meunier, R.~B. Findler, and M.~Felleisen.
\newblock Modular set-based analysis from contracts.
\newblock In \emph{POPL '06}, POPL, 2006.

\bibitem[Meyer(1991)]{dvanhorn:meyer-eiffel}
B.~Meyer.
\newblock \emph{Eiffel : The Language}.
\newblock 1991.

\bibitem[{Nguy\~{\^e}n} et~al.(2014){Nguy\~{\^e}n}, Tobin-Hochstadt, and Van{
  }Horn]{techreport}
P.~C. {Nguy\~{\^e}n}, S.~Tobin-Hochstadt, and D.~Van{ }Horn.
\newblock Soft contract verification.
\newblock \emph{CoRR}, 2014.
\newblock URL \url{http://arxiv.org/abs/1307.6239}.

\bibitem[Ong(2006)]{dvanhorn:Ong2006ModelChecking}
C.~H.~L. Ong.
\newblock On {Model-Checking} trees generated by {Higher-Order} recursion
  schemes.
\newblock LICS, 2006.

\bibitem[Plosch(1997)]{dvanhorn:Plosch1997Design}
R.~Plosch.
\newblock Design by contract for {P}ython.
\newblock 1997.
\newblock APSEC/ICSC'97.

\bibitem[Rondon et~al.(2008)Rondon, Kawaguci, and
  Jhala]{dvanhorn:Rondon2008Liquid}
P.~M. Rondon, M.~Kawaguci, and R.~Jhala.
\newblock Liquid types.
\newblock PLDI, 2008.

\bibitem[Shivers(1988)]{dvanhorn:shivers-88}
O.~Shivers.
\newblock Control flow analysis in {S}cheme.
\newblock PLDI, 1988.

\bibitem[Strickland et~al.(2012)Strickland, Tobin-Hochstadt, Findler, and
  Flatt]{dvanhorn:Strickland2012Chaperones}
T.~S. Strickland, S.~Tobin-Hochstadt, R.~B. Findler, and M.~Flatt.
\newblock Chaperones and impersonators: run-time support for reasonable
  interposition.
\newblock OOPSLA, 2012.

\bibitem[Terauchi(2010)]{dvanhorn:Terauchi2010Dependent}
T.~Terauchi.
\newblock Dependent types from counterexamples.
\newblock POPL, 2010.

\bibitem[Tobin-Hochstadt and
  Felleisen(2010)]{dvanhorn:TobinHochstadt2010Logical}
S.~Tobin-Hochstadt and M.~Felleisen.
\newblock Logical types for untyped languages.
\newblock ICFP, 2010.

\bibitem[Tobin-Hochstadt and {Van
  Horn}(2012)]{dvanhorn:TobinHochstadt2012Higherorder}
S.~Tobin-Hochstadt and D.~{Van Horn}.
\newblock Higher-order symbolic execution via contracts.
\newblock OOPSLA, 2012.

\bibitem[Tobin-Hochstadt et~al.(2011)Tobin-Hochstadt, St-Amour, Culpepper,
  Flatt, and Felleisen]{dvanhorn:TobinHochstadt2011Languages}
S.~Tobin-Hochstadt, V.~St-Amour, R.~Culpepper, M.~Flatt, and M.~Felleisen.
\newblock Languages as libraries.
\newblock PLDI, 2011.

\bibitem[Tsukada and Kobayashi(2010)]{dvanhorn:Tsukada2010Untyped}
T.~Tsukada and N.~Kobayashi.
\newblock Untyped recursion schemes and infinite intersection types.
\newblock FoSSaCS, 2010.

\bibitem[Van{ }Horn and Might(2010)]{dvanhorn:VanHorn2010Abstracting}
D.~Van{ }Horn and M.~Might.
\newblock Abstracting abstract machines.
\newblock ICFP, 2010.

\bibitem[Van{ }Horn and Might(2012)]{dvanhorn:VanHorn2012Systematic}
D.~Van{ }Horn and M.~Might.
\newblock Systematic abstraction of abstract machines.
\newblock \emph{Journal of Functional Programming}, 2012.

\bibitem[Vazou et~al.(2013)Vazou, Rondon, and
  Jhala]{dvanhorn:Vazou2013Abstract}
N.~Vazou, P.~Rondon, and R.~Jhala.
\newblock Abstract refinement types.
\newblock ESOP, 2013.

\bibitem[Vazou et~al.(2014)Vazou, Seidel, Jhala, Vytiniotis, and
  Peyton-Jones]{local:vazou-icfp2014}
N.~Vazou, E.~L. Seidel, R.~Jhala, D.~Vytiniotis, and S.~Peyton-Jones.
\newblock Refinement types for haskell.
\newblock ICFP, 2014.

\bibitem[Vytiniotis et~al.(2013)Vytiniotis, {Peyton Jones}, Claessen, and
  Ros\'{e}n]{dvanhorn:Vytiniotis2013HALO}
D.~Vytiniotis, S.~{Peyton Jones}, K.~Claessen, and D.~Ros\'{e}n.
\newblock {HALO}: {H}askell to logic through denotational semantics.
\newblock POPL, 2013.

\bibitem[Wright and Cartwright(1997)]{dvanhorn:Wright1997Practical}
A.~K. Wright and R.~Cartwright.
\newblock A practical soft type system for {Scheme}.
\newblock \emph{ACM Trans. Program. Lang. Syst.}, 1997.

\bibitem[Xu(2012)]{dvanhorn:Xu2012Hybrid}
D.~N. Xu.
\newblock Hybrid contract checking via symbolic simplification.
\newblock PEPM, 2012.

\bibitem[Xu et~al.(2009)Xu, {Peyton Jones}, and
  Claessen]{dvanhorn:Xu2009Static}
D.~N. Xu, S.~{Peyton Jones}, and S.~Claessen.
\newblock Static contract checking for {H}askell.
\newblock POPL, 2009.

\bibitem[Zhu and Jagannathan(2013)]{dvanhorn:DBLP:conf/vmcai/ZhuJ13}
H.~Zhu and S.~Jagannathan.
\newblock Compositional and lightweight dependent type inference for {ML}.
\newblock 2013.

\end{thebibliography}
